\documentclass[11pt]{article}
\usepackage[utf8]{inputenc}
\usepackage{amsmath,amsthm,amsfonts,fullpage, amssymb, xcolor,hyperref,bbm,graphicx}
\usepackage[ruled,noend,noline]{algorithm2e}
\usepackage{authblk}
\usepackage{mathtools}
\usepackage{cleveref}
\usepackage{natbib}

\usepackage{amsthm}
\usepackage{aliascnt}
\usepackage{cleveref}
\theoremstyle{plain}
\newtheorem{theorem}{Theorem} 

\newcommand{\makenamedtheorem}[2]{%
  \newaliascnt{#1}{theorem}
  \newtheorem{#1}[#1]{#2}
  \aliascntresetthe{#1}
}

\newtheorem*{theorem*}{Theorem}
\makenamedtheorem{corollary}{Corollary}
\makenamedtheorem{conjecture}{Conjecture}
\makenamedtheorem{lemma}{Lemma}
\makenamedtheorem{proposition}{Proposition}
\makenamedtheorem{protocol}{Protocol}
\makenamedtheorem{claim}{Claim}
\makenamedtheorem{fact}{Fact}
\makenamedtheorem{assumption}{Assumption}

\makenamedtheorem{example}{Example}
\makenamedtheorem{problem}{Problem}
\makenamedtheorem{definition}{Definition}
\makenamedtheorem{intuition}{Intuition}
\makenamedtheorem{idea}{Idea}
\makenamedtheorem{exercise}{Exercise}
\makenamedtheorem{remark}{Remark}

\crefname{theorem}{Theorem}{Theorems}
\Crefname{theorem}{Theorem}{Theorems}

\crefname{lemma}{Lemma}{Lemmas}
\Crefname{lemma}{Lemma}{Lemmas}

\crefname{corollary}{Corollary}{Corollaries}
\Crefname{corollary}{Corollary}{Corollaries}

\crefname{conjecture}{Conjecture}{Conjectures}
\Crefname{conjecture}{Conjecture}{Conjectures}

\crefname{proposition}{Proposition}{Propositions}
\Crefname{proposition}{Proposition}{Propositions}

\crefname{protocol}{Protocol}{Protocols}
\Crefname{protocol}{Protocol}{Protocols}

\crefname{claim}{Claim}{Claims}
\Crefname{claim}{Claim}{Claims}

\crefname{fact}{Fact}{Facts}
\Crefname{fact}{Fact}{Facts}

\crefname{assumption}{Assumption}{Assumptions}
\Crefname{assumption}{Assumption}{Assumptions}

\crefname{example}{Example}{Examples}
\Crefname{example}{Example}{Examples}

\crefname{problem}{Problem}{Problems}
\Crefname{problem}{Problem}{Problems}

\crefname{definition}{Definition}{Definitions}
\Crefname{definition}{Definition}{Definitions}

\crefname{intuition}{Intuition}{Intuitions}
\Crefname{intuition}{Intuition}{Intuitions}

\crefname{idea}{Idea}{Ideas}
\Crefname{idea}{Idea}{Ideas}

\crefname{exercise}{Exercise}{Exercises}
\Crefname{exercise}{Exercise}{Exercises}

\crefname{remark}{Remark}{Remarks}
\Crefname{remark}{Remark}{Remarks}

\newcommand{\Del}[2]{\textbf{Del}_{#1}(#2)}
\newcommand{\Bin}[2]{\textbf{Bin}(#1,#2)}

\newcommand{\Exp}{\operatorname*{\mathbb{E}}}

\newcommand{\abs}[1]{\left|#1\right|}

\renewcommand{\mod}{\text{ mod }}

\renewcommand{\Pr}{\textbf{Pr}}

\newcommand{\eps}{\varepsilon}

\title{Quasipolynomial Trace Reconstruction}
\author[ ]{Arnav Burudgunte, Paul Valiant, Hongao Wang}\affil[ ]{Purdue University}
\begin{document}

\maketitle
\begin{abstract}
We show that trace reconstruction on $n$-bit strings is possible using a quasipolynomial number of traces, for any retention probability $p$ that is at least inverse polylogarithmic in $n$.
\end{abstract}
\section{Introduction}

Given an $n$-bit string $x$, a \emph{deletion channel with retention probability $p$} deletes each bit of $x$ with probability $1-p$ and returns the remaining bits, which is called a \emph{trace}. 
The trace reconstruction problem asks, how many traces from an unknown string $x$ are needed to reconstruct $x$?

This problem relates to some of the most fundamental issues in information theory, and has been studied---along with many proposed variants---for decades, both for its own sake and for the sake of applications. See Section~\ref{sec:related} for a discussion of related work.

For constant retention probability $p$, the best lower bound says that $\tilde{\Omega}(n^{3/2})$ traces are necessary~\cite{chase2021new}; the best previously known upper bound is exponentially higher, saying that $\exp(\tilde{O}(n^{1/5}))$ traces suffice~\cite{chase2021separating}.

We substantially improve our algorithmic understanding of trace reconstruction, showing a quasipolynomial upper bound.

\begin{theorem*}
There exists a constant $c>0$ such that for any retention probability $p>0$, trace reconstruction on $n>1$ bit strings can be done from $e^{p^{-7/3}(\log_2 n)^c}$ traces.
\end{theorem*}

Our techniques bypass known barriers for the trace reconstruction problem. The strongest lower bound for the standard setting of trace reconstruction is the ``local statistical query (SQ)'' bound of~\cite{chen24trace}, which we briefly introduce. An $\ell$-local query specifies a function $f$ on $\ell$ consecutive bits of a trace $U$, to which an oracle responds with a $\delta$-accurate estimate of the expected value of $f$ across traces from the unknown string. The surprising lower bound from this work is: any trace reconstruction algorithm that makes $\tilde{O}(n^{1/5})$-local queries must have tolerance $\delta = \exp(-\tilde{\Omega}(n^{1/5}))$. Further, the $\exp(\tilde{O}(n^{1/5}))$ upper bound of~\cite{chase2021separating} can be reinterpreted in a $\tilde{O}(n^{1/5})$-local guise~\cite{chen24trace}. Thus, any algorithm that beats the $\exp(\tilde{O}(n^{1/5}))$ upper bound must take advantage of global structure in traces, in a way that previous algorithms have not.

One of the simplest global queries asks: what is the expected value of the product of bits $s_1,\ldots,s_k$ of a trace $U\sim \Del{p}{x}$. We call this a $k^\textrm{th}$ order statistic, and will denote it $x_p^{(s)}$ (see Definition~\ref{def:k-th}). Our analysis will show how to identify a string from its $polylog(n)$ order statistics.

It is well known that, up to $O(n)$ factors in sample complexity, the problem of recovering an unknown $n$-bit string $x$ is equivalent to the problem of \emph{distinguishing} a given pair $x,y\in\{0,1\}^n$, given repeated traces from one of the strings. It is often simpler to think of trace reconstruction as synonymous with this problem of distinguishing $x$ from $y$.

Overall, our approach to distinguishing traces from $x$ versus $y$ can be described as ``zooming out around the point $d$ of first discrepancy between $x,y$.'' Trace reconstruction from very short strings is easy, even with exponential approaches, so if we artificially pretend we had traces from a tiny window around $d$, then it is easy to find statistics $s$ that strongly distinguish $x$ from $y$, even simple statistics with $k=1$ that look at the mean value of a single bit of the trace.
Our idea is to iteratively transform a statistic $s$ that only works on traces from a very local window around $d$ into a new statistic $S$ that survives being used on a more zoomed-out window around $d$. We thus iteratively transform a trivial $k=1$ local statistic into a successively more global statistic of higher order until we have ultimately recovered a statistic that distinguishes $x$ from $y$ on traces from the entire string.

Intuitively, given a window of size $R$, the trace will produce $\Bin{R}{p}$ bits, and thus, roughly, the location of individual bits will end up binomially ``blurred'' by radius $\approx\sqrt{Rp}$. At each step of the induction, we will roughly \emph{square} the size of the window, squaring the blurring. The challenge of the induction step then is: given a statistic $s$ that distinguishes structure from $x$ versus $y$ when blurred by $\approx\sqrt{Rp}$, can we somehow construct a new statistic $S$ whose distinguishing signal survives blurring by $\approx Rp$? 

To understand the effect of blurring, consider the sequence of all shifts of statistic $s$ on some small window $W$ of $x$: let $f(j)=x^{(s+j)}_{W,p}$. Consider traces from some larger window; each trace lets us construct a noisy estimate of the sequence $f$, but \emph{shifted} by a large binomially-distributed offset. Can we construct a statistic $S$ on this larger window that lets us reconstruct $f$ (up to shift) on the smaller window, or at least emulate its power to distinguish $x$ from $y$? Phrased in this way, our challenge is similar to the problem of \emph{multiple reference alignment}, abstracted from an electron microscopy setting, where, given many noisy randomly-translated images of a molecule, you want to reconstruct the molecule (up to a shift). This was the algorithmic challenge at the center of the 2017 Nobel Prize in Chemistry.

The solution is to show that $f$ is uniquely defined (up to shift) by its second and third order statistics, even when summed/blurred over all shifts. Inspired by BLR linearity testing, we show in Section~\ref{sec:sequences}, roughly, that there is a three-point test to distinguish statistics from $x$ versus $y$ on the zoomed-out window: we show that the following is nonnegligible for some offsets $\ell_0,\ell_1,\ell_2$ \[\left|\sum_j x_{W,p}^{(s+j+\ell_0)}x_{W,p}^{(s+j+\ell_1)}x_{W,p}^{(s+j+\ell_2)}-y_{W,p}^{(s+j+\ell_0)}y_{W,p}^{(s+j+\ell_1)}y_{W,p}^{(s+j+\ell_2)}\right|\]

The focus of Section~\ref{sec:alpha} is to reexpress a triple product of $k^{\textrm{th}}$ order statistics in terms of a linear combination of statistics of order $\leq 3k$, to complete the induction. Crucial towards this goal, Lemma~\ref{lem:simulation} shows that we can actually simulate three \emph{independent} traces from a single trace; the cost of this simulation is that the retention probability triples, going from $p$ to $P\in[p,3p]$.

See Section~\ref{sec:overview} for a much more detailed discussion of the technical components of this paper, including many important steps we have omitted here.

\medskip\noindent{\bf The trajectory of our induction:}

The successive zoom-out idea leads to an induction step where, given a window of size $R$ around the point of first discrepancy $d$, and a statistic $s$ of order $k$ with some nonnegligible discrepancy $\tau$ between $x,y$ when applied on this window: we seek a statistic $S$ of order $\leq 3k$ that, for a window of size $\approx R^2$ has discrepancy $\geq \tau^c$ (for constant $c$), when the retention probability $p$ increases to $P\in[p,3p]$.

Since it takes $\log\log n$ squarings to go from some small window size to a window size $n$, we repeat the induction step $\approx\log\log n$ times. 

Over the course of $\log\log n$ steps of induction: the order $k$ will increase by $polylog(n)$ factor; $\log\tau<0$ will get scaled by $polylog(n)$ factor meaning that $\tau$ will decay inverse quasipolynomially; and the retention probability $p$ might increase by $polylog(n)$ factor.

The $polylog(n)$ growth of retention probability puts us in an interesting situation: at the end of the induction, we want a conclusion about deletion channels for the entire strings $x,y$ and for $p$ in a (typically) constant range; thus the start of the induction must be in the somewhat weird regime where the retention probability $p$ is inverse polylogarithmic.

Lemma~\ref{lem:base} shows the base case for our induction. For technical reasons, we start the induction with a window size $R=polylog(n)$. We show that, for two strings $x,y$ with a point of first discrepancy that is $polylog(n)$ close to the start of the window, and a retention probability that is $1/polylog(n)$, then there is an order $k=1$ (``mean based'') statistic on the window that distinguishes $x,y$ with inverse quasipolynomial signal. This is proved with mostly standard techniques, though we adapt these techniques to get tight bounds on the \emph{location} in the trace at which the means of traces from $x,y$ differ, in terms of the point of first discrepancy.

\medskip\noindent{\bf Run time vs. sample complexity:}

Our main result shows the quasipolynomial sample complexity of trace reconstruction. It has been well known that maximum likelihood estimation (MLE) is the algorithm with essentially optimal sample complexity for this problem (see~\cite{kuan24on}) and thus our paper implies that MLE also succeeds from quasipolynomial traces. In the context of distinguishing two arbitrary strings $x,y$ (as opposed to reconstructing one arbitrary string $x$), there is a natural quadratic-time dynamic programming algorithm to compute the likelihood of generating a particular trace $U$ from a particular string $x$. And thus we conclude there is a quasipolynomial time algorithm for distinguishing $x,y$ from traces.

For the case of reconstructing an arbitrary string, MLE suggests computing the likelihood of generating the given traces from each of the exponential number of $x\in\{0,1\}^n$, reusing the quasipolynomial number of traces. However it is unclear whether there is an adaptation of our approach that additionally has quasipolynomial running time.

\subsection{Outline}
We discuss related work in Section~\ref{sec:related}. Section~\ref{sec:definitions} contains notation and definitions. Section~\ref{sec:overview} contains an in-depth technical overview of some of the global relations between technical sections in this paper, and we strongly suggest readers review this before going into more specific technical sections. Section~\ref{sec:sequences} proves our analog of ``3-point linearity testing'' combined with analysis of Fourier properties of sequences with ``abrupt start''. Section~\ref{sec:alpha} shows how to simulate several traces from one trace of higher retention probability, and analyzes the relation between statistics of traces simulated in this manner. Section~\ref{sec:deconvolution} shows how to deconvolve a sequence by the pdf of a binomial, while maintaining very tight control of the support and Fourier decay. Section~\ref{sec:induction} shows our main induction step in Proposition~\ref{prop:induction}. Section~\ref{sec:main} contains the base case in Lemma~\ref{lem:base}, and the main result in Theorem~\ref{thm:main}. Section~\ref{sec:misc} contains a few small technical tools.

\subsection{Related Work}
\label{sec:related}

The trace reconstruction problem is motivated by the multiple sequence alignment problem in computational biology and was introduced to the TCS community by the work of Batu et al. \cite{Batu2004reconstruct}. Since then, there has been a line of work (\cite{holenstein08trace, de2019optimal,nazarov16Trace,holden20Lower,chase2021new,chase2021separating}) improving the upper or lower bounds bounds on the sample complexity. However, after decades of effort, the best prior upper bound shown by \cite{chase2021separating} is $\exp(\tilde{O}(n^{1/5}))$ and the best lower bound by \cite{chase2021new} is $\tilde{\Omega}(n^{3/2})$; these bounds have an exponential gap. 

In this paper, we consider $k^{\textrm{th}}$ order statistics corresponding to the expectation of the product of $k$ bits in a trace. When $k=1$, these statistics measure the expected value of a single bit of a trace and are called \emph{mean-based statistics}. Mean-based statistics for the trace reconstruction problem have been thoroughly studied. It is known that $\exp(\Theta(n^{1/3}))$ traces are necessary and sufficient to reconstruct a string via mean-based statistics \cite{de2019optimal,nazarov16Trace}. The base case of our analysis uses a stronger fact about mean-based statistics: if two strings $x,y$ have a point of first discrepancy $d$, they can be distinguished by a mean-based statistic using $\exp(O(d^{1/3})$ traces. A similar result was proved by \cite{peres17Average}, though with looser bounds on the location of this statistic.

Many known algorithms for worst-case trace reconstruction (e.g., \cite{holenstein08trace, de2019optimal,nazarov16Trace}) can be interpreted through the \emph{statistical query} model in which the algorithm accesses traces only by estimating the expectation of some function $f$ on the trace process---which includes approaches like ours that rely on $k^{\textrm{th}}$ order statistics. A strong lower bound is known for \emph{local} statistical query algorithms, those algorithms which query functions of at most $k$ \textit{consecutive} bits~\cite{kuan24on,chen24trace}. The lower bound states that $O(n^{1/5})$-local statistical query algorithm can only succeed by generating estimates which are accurate to tolerance $2^{-\tilde{\Omega}(n^{1/5})}$, which requires exponentially many samples. Significantly, the previous upper bound of \cite{chase2021separating} can be reexpressed as an $O(n^{1/5})$-local statistical query algorithm~\cite{chen24trace}. In contrast, while our algorithm only queries $polylog(n)$ bits in the trace, these bits could be \emph{anywhere}, and not consecutive.

The average-case version of trace reconstruction---reconstruction of a string $x$ chosen at random---has also been studied extensively. There are sublinear upper bounds for this case \cite{peres17Average,pmlr-v75-holden18a,holden2020subpolynomial}; the best upper bound is $\exp(\tilde{O}(\log^{1/5}n))$ \cite{rubinstein2023average}, and the best lower bound is $\tilde{\Omega}(\log^{5/2}n)$ \cite{chase2021new}. Notably, the algorithmic upper bounds for average-case reconstruction cannot be expressed as statistical queries; they generally rely on global alignment between traces rather than summarizing traces via statistics. 

Many variants of trace reconstruction have been productively investigated: circular trace reconstruction \cite{narayanan2021circular, burudgunte_et_al:LIPIcs.ITCS.2026.30}, coded trace reconstruction \cite{cheraghchi2020coded,brakensiek2020coded}, approximate trace reconstruction \cite{daviesApproximateTraceReconstruction2021,chenApproximateTraceReconstruction2023}, matrix trace reconstruction  \cite{krishnamurthy2021trace}, and population recovery \cite{ban2019beyond, 3458064.3458141, rivkin25a}. 

As a key component of our analysis, in Section~\ref{sec:sequences} we study the problem of distinguishing arbitrary real sequences via summed second and third order statistics. This problem has analogs in the field of multiple reference alignment, in which one seeks to reconstruct a real-valued sequence from measurements which add a random shift and Gaussian noise \cite{perry2019sample, bandeira2023estimation}. Our proofs in Section~\ref{sec:sequences}, particularly the idea of ``three-point tests" for matching third order statistics, are partially inspired by the classic BLR test for linearity \cite{10.1145/100216.100225}, which has been generalized to linearity testing for all finite groups (including non-abelian groups)~\cite{https://doi.org/10.1002/rsa.20182}, and to the domain of integers \cite{devadas2016self}, which corresponds to our setting.

\section{Notation and Deletion Channel Basics}\label{sec:definitions}

We denote the binomial distribution, representing $n$ flips of a $p$-biased coin, as $\Bin{n}{p}$. The probability of getting exactly $j$ heads from this distribution we denote with the 3-argument function $bin(n,j,p):=\binom{n}{j}p^j(1-p)^{1-j}$, though taking values 0 when $j\notin \{0,\ldots,n\}$.

For a binary string $x\in\{0,1\}^n$ we 1-index it via indices from $\{1,\ldots,n\}$. We may refer to a substring of $x$ using notation $x_{[i:j]}$ which refers to the string $(x_i,\ldots,x_j)$. For two binary strings $x\neq y$, we will often focus on the ``point of first discrepancy'', $d:=\min \{i: x_i\neq y_i\}$.

We formally define the deletion channel and the trace reconstruction problem now.
\begin{definition}
Given a binary string $x\in\{0,1\}^n$, and a probability $p$, the deletion channel with retention probability $p$ is defined as the probabilistic process that, for each bit of $x$, retains it with probability $p$ independently, discarding the other bits. The result of this process is a binary string $U$ of length between 0 and $n$, which we call a trace. We denote the process of drawing a trace from the deletion channel by $U\sim \Del{p}{x}$.
\end{definition}

\begin{definition}[Trace reconstruction problem]
    Let $x \in \{0,1\}^n$. Given sample access to traces generated by the deletion channel $\Del{p}{x}$, an algorithm $\mathcal{A}$ which returns $x'$ solves the \textit{trace reconstruction problem} if, with probability at least $2/3$, $x'=x$. 
\end{definition}

The main object of study in this paper is ``low-order statistics'' of the deletion channel, which we define now.

\begin{definition}\label{def:k-th}
Given a tuple $s$ of $k$ locations in $\{1,\ldots,\sigma\}$ and a probability $p$ and an $n$-bit binary string $x$, we define the notation $x^{(s)}_p$ to denote the expected value of product of locations $s_1,\ldots,s_k$ in a trace from $x$, where by convention, any accesses outside the range of the trace are interpreted as being 0: \[x^{(s)}_p:=\Exp_{U\sim \Del{p}{x}}[U_{s_1}U_{s_2}\cdots U_{s_k}]\]

We say that statistic $s$ has order $k$ and span $\sigma$.

Further, we will apply the above notation to traces of \emph{portions} of larger strings: given indices $i<j$, we use $x_{[i,j],p}^{(s)}$ to denote the expected value of the statistic $s$, when applied to the portion of $x$ from $i$ to $j$. Treating $s$ as a vector, if $\ell$ is an integer then $s+\ell$ is interpreted as adding $\ell$ to all entries of $s$, and thus we may use $x^{(s+\ell)}_p$ to denote the statistic of $s$ shifted by $\ell$. When $p$ is implicit, we may drop the subscript $p$.
\end{definition}

The analysis makes repeated use of the Fourier transform:

\begin{definition}
Given a sequence $f:\mathbb{Z}\rightarrow \mathbb{R}$ such that $\sum_j |f(j)|$ is finite, we define its Fourier transform $F:(-\pi,\pi]\rightarrow\mathbb{C}$, with argument considered as an angle mod $2\pi$, by
\[F(\xi):= \sum_{j=-\infty}^{\infty} f(j)\cdot e^{ij\xi}\]
and for such $F$, we may invert the Fourier transform as
\[f(j)=\frac{1}{2\pi}\int_{-\pi}^{\pi} F(\xi)\cdot e^{-ij\xi}\,d\xi\]
\end{definition}

We now introduce and prove a basic lemma that relates statistics on a substring to statistics on a larger string, showing, essentially, that ``zooming out'' to a larger string a has the effect of adding a binomially-distributed shift to a statistic, essentially blurring it. ``Blurring'' more technically means convolution with the pdf of a binomial distribution.

\begin{lemma}\label{lem:binomial-shift} Let $x \in \{0,1\}^n$. For any statistic $s$ and $0 \leq R \leq n$, we have
\[x_{[1:n]}^{(s)}=\eps+\sum_{j=0}^{R} bin(R,j,p) \cdot x_{[R+1:n]}^{(s-j)}\]
where $0\leq \eps\leq Pr[\Bin{R}{p}\geq \min_i s_i]$.
\end{lemma}
We note that for $j>\min_i s_i-1$, the statistic $s$ will be shifted to have invalid indices, smaller than 1; and in these cases implicitly returns value 0. We thus could equivalently view the sum in the lemma as having upper bound $\min_i s_i-1$.
\begin{proof}
    A trace from the string $x_{[1:n]}$ can be viewed as the concatenation of a trace from $x_{1:R}$ and a trace from $x_{R+1:n}$. We write $x^{(s) \mid j}_{1:n}$ to denote the expectation of $s$ conditioned on precisely $j$ bits being retained from $x_{1:R}$. If $j < \min_i s_i$, the bits contributing to $s$ come exclusively from $x_{R+1:n}$, and therefore $x^{(s) \mid j}_{[1:n]} = x_{[R+1,n]}^{(s-j)}$. Since the number of retained bits $j$ is distributed as $\Bin{R}{p}$, we have 
    \begin{align*}
        x_{[1:n]}^{(s)} &= \sum_{j=0}^{R} bin(R,j,p) x^{(s) \mid j}_{[1:n]} = \sum_{j=s_i}^{R} bin(R,j,p) x^{(s) \mid j}_{[R+1:n]} + \sum_{j=0}^{\min_i s_i-1} bin(R,j,p) x^{(s-j)}_{[R+1:n]} 
    \end{align*}
    To complete the proof, we note that $x^{s|j} \in (0,1)$ for all $j$, so the first term is at least $0$ and at most $\Pr[\Bin{R}{p} \geq \min_i s_i]$, as desired. 
\end{proof}

\section{Proof Strategy and Discussion}\label{sec:overview}

At a high level, this paper shows how to distinguish two strings $x,y\in\{0,1\}^n$ using quasipolynomial traces via an induction argument that successively zooms out around the location $d$ where $x$ and $y$ first disagree. Each step of induction (roughly), given a statistic $s$ that could distinguish $x$ from $y$ given traces from a window starting a small distance $R$ before $d$, shows how to construct a new statistic $S$ that can distinguish $x$ from $y$ on a much larger (``zoomed out'') window starting $R'$ farther back, where $R'$ is nearly $R^2$. The induction ends after roughly $\log\log n$ iterations when we have zoomed out to the full size of $x,y$, so that we have constructed a statistic that distinguishes traces on the original strings $x,y$.

As expressed in the introduction, we are guided by the intuition provided by Lemma~\ref{lem:binomial-shift} that formally explains how to compare the value of a statistic $s$ on a large window $x_{[d-R-R':n]}$ to the value of various shifts $s+j$ on the small window $x_{[d-R:n]}$. Suppose $s$ consists of indices far enough to the right so that $x^{(s)}_{[d-R-R':n],P}$ is likely to query indices entirely within the small window $x_{[d-R:n]}$. Then $x^{(s)}_{[d-R-R':n],P}$ essentially equals an average of shifted statistics $x_{[d-R:n]}^{(s+j),P}$ on the small window (weighted by the pdf of the binomial distribution $\Bin{R'}{P}$, though for the purpose of this high-level overview it will be simpler to ignore the weights and regard this as an unweighted average). The conceptual task needed to achieve an induction step is now: given a statistic $s$ that distinguishes $x$ from $y$ given traces on the small window $\{d-R,\ldots,n\}$, can we construct a related statistic $S$ that distinguishes $x$ from $y$ \emph{even when summed over all offsets} $j$? Namely, given $s$ for which $\left|x_{[d-R:n]}^{(s)}-y_{[d-R:n]}^{(s)}\right|$ is large, can we find an $S$ for which $\left|\sum_j x_{[d-R:n]}^{(S+j)}-y_{[d-R:n]}^{(S+j)}\right|$ is large? This will enable us to apply Lemma~\ref{lem:binomial-shift} to shift the window start back by $\approx R^2$, thus zooming out and completing one step of the induction.

As it turns out, we cannot do this directly, and instead ask a related question: are there offsets $\ell_0=0,\ell_1,\ell_2$ such that the following is large \[\left|\sum_j x_{[d-R:n]}^{(s+j+\ell_0)}x_{[d-R:n]}^{(s+j+\ell_1)}x_{[d-R:n]}^{(s+j+\ell_2)}-y_{[d-R:n]}^{(s+j+\ell_0)}y_{[d-R:n]}^{(s+j+\ell_1)}y_{[d-R:n]}^{(s+j+\ell_2)}\right|\]
We then connect this back to the previous question by showing, in Section~\ref{sec:alpha} (see below for more discussion) that a product of three expected statistics $s+j+\ell_0,s+j+\ell_1,s+j+\ell_2$ can be expressed as a linear combination of expected statistics $S$ of slightly higher order.

Phrased abstractly, if we define the sequence $f(j):=x^{(s+j)}_{[d-R:n]}$ to denote the values of all shifts of the statistic $s$---where $f$ will be in the range $[0,1]$ since it is expectations of a product of bits from a trace of a binary string---and define $g$ analogously as the statistics of $y$ then: under the assumption that the sequence $f$ is nonnegligibly different from $g$, can we construct some triple of offsets $\ell_0=0,\ell_1,\ell_2$ such that \begin{equation}\label{eq:third-order-example-early}\Big|\sum_j f(j+\ell_0)f(j+\ell_1)f(j+\ell_2)-g(j+\ell_0)g(j+\ell_1)g(j+\ell_2)\Big|\end{equation}
is nonnegligible. Namely, given that the sequences of statistics $f,g$ are nonnegligibly different, is there a third-order ``statistic of the statistics'' that distinguishes $f$ from $g$ even when summed over all offsets?

Resolving this challenge is the point of Section~\ref{sec:sequences}. We note that Section~\ref{sec:sequences} can be viewed as entirely about sequences, and that while it is designed to fit precisely into the induction structure of our trace reconstruction argument, the results of this section do not use the deletion channel in any way. An early lemma of this section is one of the central technical tools of this paper: Lemma~\ref{lem:contrapositive} shows that, subject to certain conditions discussed soon, \textbf{either} there is a summed second or third order statistic that distinguishes $f,g$ in the sense of Equation~\ref{eq:third-order-example-early}, \textbf{or} $f$ is essentially a shifted version of $g$ (in which case $f,g$ are indistinguishable to statistics that are summed over all shifts). 

Lemma~\ref{lem:contrapositive} can be viewed as a stronger form of one of the central results of the field of ``multiple reference alignment'', where summed second and third order statistics are referred to as the autocorrelation function and the bispectrum respectively, and the classic result is that these statistics suffice to reconstruct a sequence up to a shift, \emph{provided} that the sequence has entirely nonzero Fourier transform. The multiple reference alignment field justifies this assumption from an average-case analysis, where vanishingly small Fourier coefficients are vanishingly unlikely for random inputs; however, we need a worst-case result. 
In the setting of \emph{cyclic} trace reconstruction, we had shown in prior work that \emph{sixth} order statistics can always distinguish \emph{integer} sequences~\cite{burudgunte_et_al:LIPIcs.ITCS.2026.30}; but we cannot use this result here because our sequences $f,g$ are themselves expected statistics of traces and thus real-valued. Perhaps the closest related work to our Lemma~\ref{lem:contrapositive} comes from the celebrated 3-point linearity testing underlying the original proof of the PCP theorem~\cite{10.1145/100216.100225}, which can be interpreted as saying ``if a function passes most 3-point linearity tests, then it can be error-corrected to a linear function.''

The perhaps unexpected way this intuition shows up here is, defining $H(\xi)=\log\frac{F(\xi)}{G(\xi)}$ to be the log of the ratios of the Fourier transforms of $f,g$, we show that failed 3-point linearity tests on $H$ precisely correspond to summed third-order statistics on which $f,g$ differ (see Lemma~\ref{lem:projection-slice}), except for cases where $H$ is undefined because the Fourier transforms of $f$ or $g$ are 0. Thus, under the assumption that $f,g$ have Fourier transforms that mostly stay away from 0, then if $H$ passes the remaining linearity tests, $H$ must be linear on the points for which $F,G$ stay away from 0. Rephrased: if $f,g$ have very similar summed second and third order statistics, then $H(\xi)=\log\frac{F(\xi)}{G(\xi)}$ must be essentially a linear function of $\xi$ when it is defined; basic Fourier analysis lets us reinterpret the conclusion $F(\xi)\approx G(\xi)\cdot e^{i\xi\alpha}$ as saying that $f$ is a shifted version of $g$.

This is essentially the form of our Lemma~\ref{lem:contrapositive}, proven in terms of a 3-point linearity testing result, Lemma~\ref{lem:linearity-tester}, morally following the guidance of the BLR linearity tester but disguised by its adaptation to our setting.

The remainder of Section~\ref{sec:sequences} consists of building the tool of Lemma~\ref{lem:mostly-big-laurent} to guarantee that the Fourier transform of $f-g$ will be nonnegligible on $\geq 90\%$ of its domain, as required by the input condition of Lemma~\ref{lem:contrapositive}; combining Lemmas~\ref{lem:contrapositive} and~\ref{lem:mostly-big-laurent} into Lemma~\ref{lem:combined-contrapositive}; and proving Corollary~\ref{cor:shift} that will ultimately let us disambiguate the case that $f,g$ are shifted versions of each other.

\medskip\noindent{\bf Proving Fourier Coefficients are Nonnegligible $\geq 90\%$ of the Time, From ``Abruptness'':}

Guaranteeing that the Fourier transform of the difference in our statistics, $f-g$, stays mostly nonnegligible reflects some global decisions in our proof  strategy that might otherwise appear mysterious if not explained in this context, so we discuss this now.

An example of a sequence $f:\mathbb{Z}\rightarrow \mathbb{R}$ with mostly negligible Fourier transform is any extremely smooth function, such as the pdf of a binomial of high variance. This kind of function not only violates the input conditions of our lemmas but is also fatal to their conclusions: let $f(j)$ consist of the even entries of $bin(n,j,\frac{1}{2})$, taking values 0 for odd $j$; and let $g(j)$ consist of the odd entries. One can show that \emph{all} constant-order statistics of $f$ are exponentially close to those of $g$, despite $f,g$ being nonnegative functions that are not shifts of each other. We need to show that this kind of situation will not arise in our trace reconstruction setting. We achieve this via Lemma~\ref{lem:mostly-big-laurent}, which roughly says that if a sequence starts \emph{abruptly}, instead of having a long left tail, then its Fourier transform must be mostly nonnegligible.

From a broader perspective, we avoid vanishing Fourier transforms by focusing on discriminating $f,g$ via statistics that originate \emph{as close to the point of first discrepancy, $d$, as possible}. Namely, $x,y$ are identical to each other to the left of location $d$ by definition; and location $d$ in $x$ will end up near location $Rp$ in a trace from $\Del{p}{x_{[d-R:n]}}$. Thus for any statistic $s$ with small span, its discrepancy $|f(j)-g(j)|=\left|x^{(s+j)}_{[d-R:n]}-y^{(s+j)}_{[d-R:n]}\right|$ will decay extremely rapidly for $j$ to the left of location $Rp$, while being nonnegligible at some location $\ell$ near $Rp$ by assumption. We intuitively summarize this by saying that $f(j)-g(j)$ ``starts abruptly'' near location $Rp$.

In short, Lemma~\ref{lem:mostly-big-laurent} will formalize the intuition that ``sequences that start abruptly have mostly nonnegligible Fourier transform.'' The repeated use of this crucial lemma in our induction is enabled by our strategy of successively ``zooming out'' around location $d$, repeatedly focusing on statistics $s$ that are located as close to $d$ as possible while still distinguishing $x$ from $y$: guaranteeing sharp left tail decay of $f-g$ just to the left of a location $\ell$ where the difference is nonnegligible.

To motivate Lemma~\ref{lem:mostly-big-laurent}, we start by presenting the simplest version of it that formalizes the intuition ``if a sequence starts \emph{abruptly} then its Fourier transform must be mostly nonnegligible.'' In the below lemmas, the sequence $f$ is meant to correspond to the difference in statistics that we have been denoting as $f-g$. The following lemma considers the most drastic possible ``abrupt start'', where a sequence has values 0 for negative inputs, and then has significant value $f(0)$. We lower-bound the average value of the log of the magnitude of the Fourier transform by $\log|f(0)|$.

\begin{lemma}\label{lem:mostly-big}
Given a sequence $f:\mathbb{Z}\rightarrow\mathbb{R}$ supported only on $j\geq 0$ then: letting $F(\xi)$ be the Fourier transform of $f$, we have that $\frac{1}{2\pi}\int_{-\pi}^{\pi}\log|F(\xi)|\,d\xi\geq \log|f(0)|$.
\end{lemma}

If we denote $|f(0)|$ as $\tau$ and add the assumption that $\sum_j |f(j)|\leq 1$, which imples $|F(\xi)|\leq 1$ everywhere, then Markov's inequality lets us conclude that, for $\geq 90\%$ of the domain, $|F(\xi)|\geq \tau^{10}$. Namely, ``if $f$ is $\tau$-abrupt, then its Fourier transform has magnitude at least $\tau^{10}$ for $\geq 90\%$ of its domain''.

However, we might not be able to prove that $f-g$ vanishes \emph{immediately} to the left of the location $\ell$ where it is nonnegligible. The following lemma captures how these results degrade as the location $\ell$ at which we can lower-bound $|f|$ moves farther from 0, the location at which $f$ vanishes.

\begin{lemma}\label{lem:mostly-big-exponential}
Given a sequence $f:\mathbb{Z}\rightarrow\mathbb{R}$ supported only on $j\geq 0$, where $\sum_j |f(j)|\leq 1$ and $|f(\ell)|\geq \tau$ for some $\ell\geq 0$ then: letting $F(\xi)$ be the Fourier transform of $f$, we have that $|F(\xi)|$ is $\geq \tau^c e^{-c \ell}$ on $\geq 90\%$ of the frequencies $\xi$, for some universal constant $c>0$.
\end{lemma}

Lemma~\ref{lem:mostly-big-exponential} yields lower bounds on the Fourier transform that decay \emph{exponentially} with $\ell$. However, when we use Lemma~\ref{lem:mostly-big-laurent} (via Lemma~\ref{lem:combined-contrapositive} in our main induction step, Proposition~\ref{prop:induction}), we cannot afford bounds that are inverse exponential in $\ell$; the rough scale of control we will have over $\ell$ is given by the size of the window of the \emph{previous} iteration of induction, $\ell\approx R_{prev}\cdot p$, which might be polynomially large in the string size $n$. At a high level, our solution is to ``change the scale'' at which we apply Lemma~\ref{lem:mostly-big-exponential} by blurring the sequence $f-g$ to width $\approx R_{prev}\cdot p$; this will correspond to blurring to variance $Rp$, achieved by convolving with $\Bin{R}{p}$---recall that $R\approx R_{prev}^2$ since we roughly square the window size at each iteration.

Lemma~\ref{lem:mostly-big-laurent} is essentially what one would expect from changing the length scale of Lemma~\ref{lem:mostly-big-exponential} by $\frac{1}{\sqrt{Rp}}$ factor by blurring with $\Bin{R}{p}$: the sharp cutoff of $f$ to the left of $0$ becomes a condition that ``$f$ decays roughly at the rate of $\Bin{R}{p}$ to the left of 0''; the Fourier conclusion only applies for frequencies $\lesssim\frac{1}{\sqrt{Rp}}$; and the exponential decay of the Fourier lower bound as a function of $\ell$ becomes $\frac{1}{\sqrt{Rp}}$ times slower.
 
In short: convolving with $\Bin{R}{p}$ in the input of Lemma~\ref{lem:mostly-big-laurent} lets us change the meaning of ``abruptly'' to count in a different pixel size.

\medskip\noindent{\bf Using Sections~\ref{sec:sequences},~\ref{sec:alpha}, and~\ref{sec:deconvolution} in the Induction:}

As we have argued above, convolving the sequences of statistics $x_{[d-R:n]}^{(s+j)},y_{[d-R:n]}^{(s+j)}$ with $\Bin{R}{p}$ gives us crucial properties that enable the rest of our analysis. Fortuitously, the deletion channel itself has the effect of \textit{implicitly} convolving with $\Bin{R}{p}$: Lemma~\ref{lem:binomial-shift} says that $x_{[d-2R:n]}^{(s+j')}$ is essentially the convolution of $x_{[d-R:n]}^{(s+j)}$ with $\Bin{R}{p}$. Thus the main induction proposition starts with the assumption $|x_{[d-2R:n],p}^{(s+\ell)}-y_{[d-2R:n],p}^{(s+\ell)}|\geq \tau$, starting $2R$ back from $d$. We then define the sequences $f,g$ correspondingly as (weighted versions of) the statistics, $w(j) x^{(s+j)}_{[d-R:n],p}, w(j) y^{(s+j)}_{[d-R:n],p}$, but then ultimately apply the lemmas of Section~\ref{sec:sequences} to $f_b:=f\ast \Bin{R}{p}, g_b:=g\ast \Bin{R}{p}$ which emulate $x_{[d-2R:n],p}^{(s+j)},y_{[d-2R:n],p}^{(s+j)}$ respectively, to give us the guarantee that $|f_b(\ell)-g_b(\ell)|\gtrsim \tau$.

Section~\ref{sec:sequences} then shows how to construct a double or triple product that distinguishes $f,g$ respectively even when summed.

Section~\ref{sec:alpha} starts with a simple lemma that shows that we can simulate a product of two or three expected statistics of order $k$ as a linear combination of expected statistics of order $\leq 3k$; the pigeonhole principle lets us extract a single statistic $S$ of order $\leq 3k$ with nonnegligible discrepancy between $x,y$ when summed over all shifts $j$ with coefficient $A(j)$. The rest of the section shows that the domain bounds and smoothness bounds of the weight function $w$ are essentially inherited by the coefficients $A(j)$ of the shifts.

Section~\ref{sec:deconvolution} is an analysis of how to \emph{deconvolve} by a binomial, eventually showing how to construct a sequence $A'$ such that $A'\ast \Bin{R'}{P}\approx A$. Reexpressing the coefficients in terms of a convolution with $\Bin{R'}{P}$ sets us up for a final application of Lemma~\ref{lem:binomial-shift}, moving the start of our window back by $R'$. Knowing that $\left|\sum_j A'(j) \big(x_{[d-R-R':n],P}^{(S+j)}-y_{[d-R-R':n],P}^{(S+j)}\big)\right|\geq \tau^c$, we then use pigeonhole a final time to extract a single shift of $S$ on which $x,y$ differ nonnegligibly.

To recap, we thus use Lemma~\ref{lem:binomial-shift} twice per induction step: once at the beginning to convolve with $\Bin{R}{p}$ to add smoothness before the lemmas of Section~\ref{sec:sequences}, effectively moving the deletion channel start point from $d-2R$ forward to $d-R$, and once at the end to \emph{deconvolve} with $\Bin{R'}{P}$ to let us simulate moving the deletion channel start point back to $d-R-R'$.

\medskip\noindent{\bf Achieving Quasipolynomial Sample Complexity:}

We discuss here what it takes to get an induction step  ``powerful'' enough to achieve quasipolynomial sample complexity. Beyond the ingredients mentioned above, the main challenge is maintaining extremely tight control of both the support size and the smoothness of the sequences we analyze. Throughout the induction, $\tau$ serves as a ``smallness parameter'', where we ignore errors that are smaller than $poly(\tau)$; this thus gives us the relevant definition of size and smoothness: throwing out up to $poly(\tau)$ mass, what is the length $L$ of the smallest interval on which the function is supported, and what is the frequency cutoff $\eps$ beyond which its Fourier transform vanishes? 

The induction step chooses a weight function $w$ and two parameters $R',\tau'$ defining scales for the \emph{next} step of induction; we discuss these choices here. 
The weight function $w$ defines the crucial weighted statistics $f,g$, which we run through the machinery of Section~\ref{sec:sequences} to get a summed double or triple product. We run the double or triple product through the Section~\ref{sec:alpha} machinery to reexpress it as a linear combination of shifts of a slightly higher order statistic $S$, where we show that the coefficients $A$ of this linear combination roughly inherit the support size and smoothness of $w$. We then run the coefficients $A$ through the deconvolution machinery of Section~\ref{sec:deconvolution} which says that, provided $A$ is similarly compact to the pdf of $\Bin{R'}{P}$, and constant-factor smoother than it, then we can deconvolve $A$ by $\Bin{R'}{P}$ to get a sequence $A'$ that is similarly compact and smooth as both $\Bin{R'}{P}$ and $A$.

\textbf{If} Sections~\ref{sec:sequences},~\ref{sec:alpha}, and~\ref{sec:deconvolution} have, cumulatively, lost \textbf{nothing} except constant factors in terms of both domain size \textbf{and} smoothness during the process of going from $w$ to $A'$ then: since we can pick $w$, we pretend for this exposition that, up to constants, we can pick $A'$. Essentially, suppose we can pick any valid combination of domain size and smoothness for $A'$, both measured via our smallness cutoff $poly(\tau)$.

This can be viewed as picking a domain size $L$ (which you can think of as $Rp$), and asking for the smoothest function in this domain, measured after throwing out tails $\leq poly(\tau)$ in both the real and Fourier domain. As it turns out, up to constant factors in the width, the set of binomial distributions of different numbers of samples represent the \emph{optimal} tradeoff between support size and Fourier decay; so we may simply compare to binomials to determine possible length and smoothness scales. We now describe the two simple constraints that control the success of induction.

For deconvolution by $\Bin{P'}{R}$ to have a chance, we need $w$ to be smoother than $\Bin{P'}{R}$, (namely, $w$ has Fourier transform that decays faster than the Fourier transform of $\Bin{P'}{R}$, up to $poly(\tau)$ error). Thus we need $w$ to be supported on a domain larger than the length scale at which the pdf of $\Bin{R'}{P}$ decays to $poly(\tau)$ (intuitively, the radius of the domain needs to be at least $\sqrt{\log\frac{1}{\tau}}$ standard deviations, where the standard deviation of the binomial is $\approx\sqrt{R'P}$): this gives us the constraint \[ \sqrt{R'P\log\frac{1}{\tau}}\lesssim L\]

Meanwhile, if we control location to within our window size $L$, then $L$ becomes our bound on the ``abruptness'' of the sequence of the differences of statistics between $x,y$ at the \emph{next} iteration of the induction. In other words, $L$ is, roughly, a bound on the gap between where the statistic $S$ will have nonnegligible discrepancy and the point to the left of which the discrepancy will decay rapidly. The input condition for the next step of induction then reads \[L\lesssim\sqrt{R'P \log\frac{1}{\tau'}}\]

These last two equations tell us how to choose parameters: first choose $R'$ small enough that the first equation is satisfied; then, given $R'$, choose $\tau'$ small enough that the second equation is satisfied.

Finally, combining these two equations, we see that $\log\frac{1}{\tau'}$ is a constant factor larger than $\log\frac{1}{\tau}$---and hence $\tau'$ is a constant power of $\tau$---\textbf{only} if both inequalities are tight up to constants. Since, overall, we repeat the induction step $\log\log n$ times, quasipolynomial performance is possible only if each of these steps involves merely \textit{polynomial} decay of our discrepancy $\tau$ (since $\frac{1}{\tau^2}$ controls our sample complexity and thus we ultimately need the final $\tau$ to be at least inverse quasipolynomial). Thus: quasipolynomial performance is possible only because \textit{every} step across Sections~\ref{sec:sequences},~\ref{sec:alpha}, and~\ref{sec:deconvolution} has maintained constant-factor optimal control simultaneously of both domain size and smoothness during the process of going from $w$ to $A'$ in our induction step.

\section{Low Order Statistics of Sequences}\label{sec:sequences}

The main results of this section are Lemmas~\ref{lem:contrapositive} and~\ref{lem:mostly-big-laurent}, which are combined in Lemma~\ref{lem:combined-contrapositive}, extended by Corollary~\ref{cor:shift}.

The results of this section are stated without reference to the deletion channel. The goal of this section is to show that for any sequences $f\neq g$ that satisfy certain properties, we can distinguish them via second or third order statistics, even when summed over all shifts: there are offsets $\ell_0=0,\ell_1,\ell_2$ for which the following is nonnegligible
\begin{equation}\label{eq:third-order-example}\Big|\sum_j f(j+\ell_0)f(j+\ell_1)f(j+\ell_2)-g(j+\ell_0)g(j+\ell_1)g(j+\ell_2)\Big|\end{equation}

See Section~\ref{sec:overview} for an in-depth discussion of how to interpret the results of this section and how they fit into the global structure of the paper.

A key component of our approach is a ``3-point linearity testing'' lemma, Lemma~\ref{lem:linearity-tester}, which is motivated by but somewhat different from the standard BLR linearity testing result originally proved in~\cite{10.1145/100216.100225} and stated in more general form below for reference.

\begin{lemma}[From~\cite{https://doi.org/10.1002/rsa.20182}]
Given a function $f:A\rightarrow B$ between two finite groups, then if $f$ passes $\geq 1-\eps$ fraction of 3-point linearity tests $f(x)+f(y)=f(x+y)$---with $``+''$ denoting the group operation in the respective domains---then $f$ is $O(\eps)$ close to a homomorphism from $A$ to $B$ (namely, there is a homomorphism $f'$ that agrees with $f$ on $\geq 1-O(\eps)$ fraction of the inputs).
\end{lemma}

In the context of the below lemma only, let $x\approx_\delta y$ denote that $|x-y|\leq \delta$.

\begin{lemma}\label{lem:linearity-tester}
For integer $\gamma>0$, let $S\subseteq [-\gamma,\gamma]\cap\mathbb{Z}$ contain all but $\leq\frac{\gamma}{4}-1$ integers from this interval. If for some $\delta>0$ a function $H:S\rightarrow \mathbb{R}\mod 2\pi$ satisfies $H(x)+H(y)\approx_\delta H(x+y)$ whenever $x,y,x+y\in S$, then there is an $\alpha\in\mathbb{R}$ such that $H(x)\approx_{12\gamma\delta} \alpha x$ for all $x\in S\cap [-\frac{1}{2}\gamma,\frac{1}{2}\gamma]$.    
\end{lemma}

\begin{proof}
Let $I:=[-\gamma,\gamma]\cap\mathbb{Z}$, where we will also use $\frac{1}{2}I$ to denote $[-\frac{1}{2}\gamma,\frac{1}{2}\gamma]$ etc.

We first prove that there is a function $inc(\sigma)$ that robustly captures how much $H$ changes when incrementing its input by $\sigma\in\{1,\ldots,\lfloor\frac{1}{4}\gamma\rfloor\}$. Namely we claim $\exists \,inc:\{1,\ldots,\lfloor\frac{1}{4}\gamma\rfloor\}\rightarrow \mathbb{R}/2\pi\mathbb{Z}$ such that for all $\sigma\in \{1,\ldots,\lfloor\frac{1}{4}\gamma\rfloor\}$, and for all $x$ such that both $x,x+\sigma\in \frac{1}{2}I\cap S$, then $H(x+\sigma)-H(x)\approx_{4\delta} inc(\sigma)$.

To prove this, fix $\sigma\in \{1,\ldots,\lfloor\frac{1}{4}\gamma\rfloor\}$. We claim that for any $x,y$ for which all four of $x,y,x+\sigma,y+\sigma$ are in $\frac{1}{2}I\cap S$, we have $H(x+\sigma)-H(x)\approx_{4\delta} H(y+\sigma)-H(y)$, which would imply that we can define $inc(\sigma):=H(x+\sigma)-H(x)$ for an arbitrary such $x$ if one exists, and therefore have its value be within $4\delta$ of the value of all other possibilities.
 
Given any such $x,y$, we claim that there exists $z\in\frac{1}{2}I$ such that all four of $z,z+\sigma,x+z+\sigma,y+z+\sigma$ are in $S$. We first point out that all of these quantities are in $I$, since each can be expressed  as the sum of $z\in\frac{1}{2}I$ with another number that was assumed to be in $\frac{1}{2}I$ by the previous paragraph. Then, picking one of these four expressions, the number of $z$ for which the expression is \textit{not} in $S$ is $\leq\frac{1}{4}\gamma-1$ (since all but $\leq\frac{1}{4}\gamma-1$ entries of $I$ are in $S$); and thus, multiplying this by 4, we have the number of failing $z\in\frac{1}{2}I$ is $<\gamma$, yet there are at least $\gamma$ integers in $\frac{1}{2}I$, so at least one $z$ must succeed.

Thus, given that all 8 of $x,y,z,x+\sigma,y+\sigma,z+\sigma,x+z+\sigma,y+z+\sigma$ are in $S$: we use the relations $H(x)+H(z+\sigma)\approx_\delta H(x+z+\sigma)$ and $H(x+\sigma)+H(z)\approx_\delta H(x+z+\sigma)$ to conclude first that $H(x+\sigma)-H(x)\approx_{2\delta} H(z+\sigma)-H(z)$. The corresponding expression with $y$ replacing $x$ correspondingly implies $H(y+\sigma)-H(y)\approx_{2\delta} H(z+\sigma)-H(z)$. Combining lets us conclude that $H(x+\sigma)-H(x)\approx_{4\delta} H(y+\sigma)-H(y)$ as claimed.

This concludes the proof that there is a consistent definition of $inc(\sigma)$ for $\sigma\in\{1,\ldots,\lfloor\frac{1}{4}\gamma\rfloor\}$.

Next we show by induction on $\sigma$ that for all $\sigma\in \{1,\ldots,\lfloor\frac{1}{4}\gamma\rfloor\}$, we have \[inc(\sigma)\approx_{12(\sigma-1)\delta} \sigma\cdot inc(1)\] The claim is trivially true for $\sigma=1$. To show that the claim for $\sigma-1$ implies the claim for $\sigma$, consider the set of $z\in [-\frac{1}{2}\gamma,\frac{1}{4}\gamma]\cap\mathbb{Z}$ for which all three of $z,z+\sigma-1,z+\sigma$ are in $S$. The set $[-\frac{1}{2}\gamma,\frac{1}{4}\gamma]\cap\mathbb{Z}$ has size at least $\frac{3}{4}\gamma-1$; and for each of the three expressions there are $\leq\frac{1}{4}\gamma-1$ invalid values for $z$; so thus there exists some $z$ such that all of $z,z+\sigma-1,z+\sigma$ are in $\frac{1}{2}I\cap S$.

Thus from the ``consistency of $inc()$'' we can compare $H$ on any two of these three values to approximately deduce a value of $inc$: we have $inc(\sigma)\approx_{4\delta} H(z+\sigma)-H(z)$, and $inc(\sigma-1)\approx_{4\delta} H(z+\sigma-1)-H(z)$, and $inc(1)\approx_{4\delta} H(z+\sigma)-H(z+\sigma-1)$. Combining these three expressions yields $inc(\sigma)\approx_{12\delta} inc(\sigma-1)+inc(1)$, proving the induction.

Next, there exists an element $\sigma$ in both $S$ and $\{1,\ldots,\lfloor\frac{1}{4}\gamma\rfloor\}$, since the number of positive integers up to $\frac{1}{4}\gamma$ is greater than $\frac{1}{4}\gamma-1$, the number of integers of $I$ not in $S$. Further, for this $\sigma$ there exists $z$ such that both $z,z+\sigma$ are in $\frac{1}{2}I\cap S$ (by the same counting argument as usual), and thus the relation $H(z)+H(\sigma)\approx_{\delta} H(z+\sigma)$ implies that \[inc(\sigma)\approx_{5\delta} H(\sigma)\]

Thus, for any $x\in \frac{1}{2}I\cap S$, consider the intersection of $S$ with the interval between $x$ and $\sigma$: this consists of entries of $S$, with gaps between consecutive entries $\leq \frac{1}{4}\gamma$---since there are $\leq\frac{1}{4}\gamma-1$ integers in $I\backslash S$. Then, for each of the gaps $g\leq  \frac{1}{4}\gamma$, we applying the claim $inc(g)\approx_{12(g-1)\delta} g\cdot inc(1)$ across the gap, and then sum, to deduce that \[(\sigma-x)\cdot inc(1)+H(x)\approx_{12|\sigma-x|\delta} H(\sigma) \approx_{5\delta} inc(\sigma)\approx_{12(\sigma-1)\delta} \sigma\cdot inc(1)\] from which, since $|\sigma|\leq \frac{1}{4}\gamma$ and $|x|\leq\frac{1}{2}\gamma$,  we conclude the desired bound $H(x)\approx_{12\gamma\delta} x\cdot inc(1)$.
\end{proof}

We note that Lemma~\ref{lem:contrapositive} and the associated standard Lemmas~\ref{lem:autocorrelation} and~\ref{lem:projection-slice} are stated for the \emph{discrete} Fourier transform, mapping sequences of length $N$ (with indices considered mod $N$) to sequences of length $N$; this contrasts with the convention used in the rest of the paper where we use the Fourier transform to map integer sequences to functions of an angle mod $2\pi$. We use the discrete Fourier transform here because the ``linearity testing'' intuition applied to logarithms of Fourier transforms in  Lemma~\ref{lem:contrapositive} is much more natural with a discrete input space.

We now state the standard results about the Fourier transform of the autocorrelation function and its 3-point extension known as the bispectrum. Both results can be viewed as special cases of the projection-slice theorem.

\begin{lemma}\label{lem:autocorrelation}
Let $f$ be a real sequence of length $N$ with discrete Fourier transform $F$. Consider the function $s(\ell) = \sum_{j=1}^N f(j) f(j+\ell)$, with indices interpreted$\mod N$. The Fourier transform of $s$ is $S(\xi) = F(\xi)F(-\xi)=|F(\xi)|^2$.
\end{lemma}
\begin{lemma}\label{lem:projection-slice}
Let $f$ be a sequence of length $N$ with discrete Fourier transform $F$. Consider the two dimensional function $s(\ell_1,\ell_2) = \sum_{j=1}^N f(j) f(j+\ell_1)f(j+\ell_2)$, with indices interpreted$\mod N$. The Fourier transform of $s$ is $S(\xi_1,\xi_2) = F(\xi_1)F(\xi_2)F(-\xi_1-\xi_2)$.
\end{lemma}
\begin{proof}
    By definition, $S(\xi_1,\xi_2) = \sum_{\ell_1,\ell_2}  s(\ell_1,\ell_2) e^{\frac{2\pi i}{N} \cdot (\ell_1\xi_1 + \ell_2\xi_2)}$. From the definition of $s$, we have 
    \begin{equation*}
        S(\xi_1,\xi_2) = \sum_{\ell_1,\ell_2} \sum_j f(j) f(j+\ell_1)f(j+\ell_2) e^{\frac{2\pi i}{N} \cdot ((j+\ell_1)\xi_1 + (j+\ell_2)\xi_2 + j(-\xi_1 - \xi_2))} 
    \end{equation*}
    Reparameterizing the right hand side in terms of $j_1 = j + \ell_1$ and $j_2 = j + \ell_2$, the sum factors as
    \begin{equation*}
        S(\xi_1,\xi_2) = \left( \sum_{j_1} f(j_1) e^{\frac{2\pi i}{N} j_1\xi_1} \right) \left( \sum_{j_2} f(j_2) e^{\frac{2\pi i}{N} j_2\xi_2} \right) \left( \sum_{j} f(j) e^{\frac{2\pi i}{N} j(-\xi_1-\xi_2)} \right) 
    \end{equation*}
    which is precisely $F(\xi_1)F(\xi_2)F(-\xi_1-\xi_2)$. 
\end{proof}

We use the linearity testing result of Lemma~\ref{lem:linearity-tester} to show that, subject to condition on size of Fourier components, any two sequences either can be distinguished by a summed second or third order statistic, or are essentially shifted versions of each other. 

\begin{lemma}\label{lem:contrapositive}
For any $\tau > 0$ and $N \in \mathbb{N}$, the following holds. 
Let $f,g$ be length $N$ sequences of real numbers of size at most 1. Letting $F,G$ be the discrete Fourier transforms of $f,g$ respectively, assume $|F-G|$ is at least $\tau$ large on $\geq\frac{9}{10}$ fraction of (integer) frequencies up to frequency $\gamma$, for $\gamma\geq 22$. If all summed 2nd and 3rd order statistics of $f,g$ are all $\leq \frac{\tau^4}{2000N^4}$ close to each other then ``$f$ is close to a shift of $g$ at low frequencies'' in the sense that there exists a real number $\alpha$ such that for all (integer) frequencies $|\xi|$ up to $\frac{1}{2}\gamma$, we have $F(\xi)$ is within $\tau$ of $G(\xi) e^{i\xi \alpha}$.
\end{lemma}
\begin{proof}
    Let $I = [-\gamma,\gamma]$. Let $S$ be the set of integers $\xi\in I$ for which either $|F(\xi)|$ or $|G(\xi)|$ is at least $\frac{1}{2}\tau$. The set $S$ must contain any points for which $|F(\xi)-G(\xi)|\geq\tau$, since at least one of $F,G$ must have size at least $\frac{\tau}{2}$ for this to happen. Thus $S$ must contain $\geq\frac{9}{10}$ fraction of the points in $I$.

    We have, for $\delta:=\frac{\tau^4}{2000N^4}$, that $f,g$ are $\delta$-close in all summed 2nd and 3rd order statistics. Then by Lemma~\ref{lem:autocorrelation} we have that $||F(\xi)|^2-|G(\xi)|^2|\leq \delta N$ for all (integer) frequencies $\xi$. Since $\delta N\leq (\frac{\tau}{4})^2$, this implies that, for all $\xi$: \begin{equation}\label{eq:close-magnitudes}||F(\xi)|-|G(\xi)||\leq \frac{\tau}{4}\end{equation}
    In particular, for all $\xi \in S$, since at least one of $|F(\xi)|,|G(\xi)|$ is at least $\frac{\tau}{2}$, by Equation~\ref{eq:close-magnitudes}, the other must be at least $\frac{\tau}{4}$ too. Combining this with the previous paragraph, $S$ comprises $\geq\frac{9}{10}$ fraction of $I$, and on the set $S$, \emph{both} $F$ and $G$ have magnitude at least $\frac{\tau}{4}$.

    For $\xi \in S$, define $H(\xi) = imag\big(\log \frac{F(\xi)}{G(\xi)}\big)$, where $\log$ is only defined up to addition of integer multiples of $2\pi i$, so $H$ should be interpreted as a real number mod $2\pi$; equivalently, we can think of $H(\xi)=angle(F(\xi))-angle(G(\xi))$. Using Lemma~\ref{lem:projection-slice}, the Fourier transforms of the summed third-order statistics are $\delta N^2$-close everywhere, meaning that 
    \[\forall \xi_1,\xi_2,\xi_3\textrm{ s.t. }\xi_1+\xi_2+\xi_3=0:\,|F(\xi_1)F(\xi_2)F(\xi_3)-G(\xi_1)G(\xi_2)G(\xi_3)|\leq \delta N^2\]

    Thus, in the case that $\xi_1,\xi_2,\xi_3\in S$, we have that each of $F,G$ has magnitude $\geq \frac{\tau}{4}$, and thus $|F(\xi_1)F(\xi_2)F(\xi_3)|,|G(\xi_1)G(\xi_2)G(\xi_3)| \geq \frac{\tau^3}{64}$. Given that these two products have magnitudes lower bounded by $\frac{\tau^3}{64}$ and distance upper bounded by $\delta N^2$, we conclude that the angle between them is bounded by
    \[|angle(F(\xi_1)F(\xi_2)F(\xi_3))-angle(G(\xi_1)G(\xi_2)G(\xi_3))|\leq \frac{\pi}{2}\cdot\frac{\delta N^2}{\tau^3/64}\]
    Since for all $\xi$, as mentioned above, $H(\xi)=angle(F(\xi))-angle(G(\xi))$, and thus also $H(-\xi)=-H(\xi)$ since $F,G$ are Fourier transforms of real sequences, and since $\xi_1+\xi_2+\xi_3=0$ implies $\xi_1+\xi_2=-\xi_3$, we may thus conclude the following statement about $H$, restricted to inputs in $S$:
    \begin{equation}\label{eq:H-condition}
        \forall \xi_1,\xi_2\textrm{ s.t. } \xi_1,\xi_2,\xi_1+\xi_2\in S: |H(\xi_1)+H(\xi_2)-H(\xi_1+\xi_2) \mod 2\pi | \leq \frac{\pi}{2}\cdot\frac{\delta N^2}{\tau^3/64}
    \end{equation}

We thus apply Lemma~\ref{lem:linearity-tester}
to $H,S,\lfloor \gamma\rfloor, \delta_{lem}:=\frac{\pi}{2}\cdot\frac{\delta N^2}{\tau^3/64}$, where, the set of integers up to size $\gamma$ has size $2\lfloor\gamma\rfloor+1$, and we have assumed at most $\frac{1}{10}$ fraction of these are omitted from $S$; since $\gamma\geq 22$, the fraction omitted is thus at most $\frac{\gamma}{4}-1$, satisfying the input assumption of Lemma~\ref{lem:linearity-tester}.
The lemma concludes that there is an $\alpha\in\mathbb{R}$ such that for all $\xi\in S\cap [-\frac{1}{2}\gamma,\frac{1}{2}\gamma]$ we have $|H(\xi)- \alpha \xi|\leq 12\gamma\delta_{lem}$; namely, $angle(F(\xi))$ is within $12\gamma\delta_{lem}$ of $angle(G(\xi))+\alpha\xi$. Combined with $|F(\xi)|,|G(\xi)|\leq N$ (from our assumption that $|f|,|g|\leq 1$ pointwise) and $||F(\xi)|-|G(\xi)||\leq \frac{\tau}{4}$ (Equation~\ref{eq:close-magnitudes}), we conclude $F(\xi)$ is within $N\cdot 12\gamma\delta_{lem}+\frac{\tau}{4}$ of $G(\xi) e^{i\xi \alpha}$.

Crudely bounding $\gamma$ by our domain size, $N$, thus for $\delta=\frac{\tau^4}{2000N^4}$ the above bound yields that $F(\xi)$ is within $\tau$ of $G(\xi) e^{i\xi \alpha}$.

\end{proof}

The following lemma makes precise the intuition explained in Section~\ref{sec:overview} that ``sequences that start \textit{abruptly} cannot have mostly small Fourier transforms'', extending the simpler versions stated as Lemmas~\ref{lem:mostly-big} and~\ref{lem:mostly-big-exponential}.

\begin{lemma}\label{lem:mostly-big-laurent}
Let $c,c'>0$ and $c''\in(0,1)$ be constants. Given a sequence $f:\mathbb{Z} \rightarrow\mathbb{R}$ with $\sum_{j=-\infty}^{\infty} |f(j)|\leq 1$, and a binomial distribution $\Bin{R}{p}$ with $p\in (0,\frac{1}{4}]$, define the binomially blurred version of $f$ by convolving with the binomial: let $f_b:=f\ast \Bin{R}{p}$. Suppose there exists a bound $\tau\in(0,\frac{1}{2}]$, where, defining $u:=\sqrt{\frac{\log\frac{1}{\tau}}{R p}}$,  there exists an integer location $\ell\leq Rp+c\sqrt{Rp\log\frac{1}{\tau}}$ such that $|f_b(\ell)|\geq\tau$, and suppose $\sum_{j=-\infty}^{\infty} |f(j)|e^{-2uj}\leq \frac{1}{\tau^c}$. Then, defining the series $F_b(z):=\sum_{j=-\infty}^{\infty} f_b(j)z^j$ for complex $z$, we have $|F_b(e^{i\theta})|$ is at least $1/poly_{c,c',c''}(\frac{1}{\tau})$ for at least a $c''$ fraction of the frequencies $\theta$ in the range $[-c'u,c'u]$, assuming $c'u\leq \pi$.
\end{lemma}
\begin{proof}
      Since $\sum_{j=-\infty}^{\infty} |f(j)| \leq 1$, the Laurent series $F(z) \coloneqq \sum_{j=-\infty}^{\infty} f(j)z^j$ converges and satisfies $|F(z)| \leq 1$ for all $|z| = 1$. By the assumption that $\sum_{j=-\infty}^{\infty} |f(j)|e^{-2uj}$ is bounded, $F$ is analytic in the annulus $\{e^{-2u} < |z| < 1\}$. More specifically, for $z$ at some radius $e^{-r}$, we bound $|F(z)|\leq \sum_{j=-\infty}^{\infty} |f(j)|e^{-rj}$; this is a convex function of $r$ since it is a sum of convex functions of $r$; thus since we have the bounds that $|F(z)|\leq 1$ when $|z|=1$ and $|F(z)|\leq \frac{1}{\tau^c}$ when $|z|=e^{-2u}$, we can bound it on the whole annulus by the max of these, $|F(z)|\leq \frac{1}{\tau^c}$ when $|z|\in[e^{-2u},1]$.
      
      The Laurent series $F_b(z)$ is precisely the $z$-transform of $f_b$. Analogous to the Fourier transform, the convolution $f \ast \Bin{R}{p}$ corresponds to elementwise multiplication of $F$ with the generating function of $\Bin{R}{p}$, which is $((1-p)+pz)^R$. By assumption, $|f_b(\ell)| \geq \tau$; applying the Cauchy integral theorem to extract the $z^\ell$ coefficient of $F_b$ from its values on a circle of radius $e^{-u}$, we have
    \begin{equation}\label{eq:Cauchy}
        \tau\leq |[z^\ell](F(z)\cdot ((1-p)+p z)^R)| = \left|\frac{1}{2\pi}\int_0^{2\pi} \frac{F(e^{-u+i\theta})\cdot ((1-p)+p \cdot e^{-u+i\theta})^R}{(e^{-u+i\theta })^\ell}\,d\theta\right|
    \end{equation}
    Viewing the right hand side as an average over angles in $[0,2\pi]$, we conclude that the integrand must have magnitude $\geq \tau$ somewhere. In other words, there exists $z_0=e^{-u+i\theta_0}$ such that $|F(z_0)((1-p)+p z_0)^R z_0^{-\ell}|\geq \tau$.

    Let $G(z) = ((1-p)+pz)^Rz^{-\ell}$, so that integrand of the right hand side of Equation~\ref{eq:Cauchy} is equal to $F(e^{-u+i\theta})G(e^{-u+i\theta})$. For $p \in [0,1]$ and $u > 0$, we have that $\log((1-p)+pe^{-u}) \leq -p + pe^{-u} \leq p(-u+u^2)$. Using this inequality, the definition $u=\sqrt{\frac{\log\frac{1}{\tau}}{R p}}$, and the bound $\ell \leq Rp + c\sqrt{Rp\log\frac{1}{\tau}}$ yields
    \begin{equation*}
        |G(e^{-u+i\theta})| \leq \exp(p(-u+u^2)R+u\ell) \leq \exp\left((1+c)\log\frac{1}{\tau}\right) = \frac{1}{\tau^{1+c}}
    \end{equation*}
    Next, we bound $G(e^{-u+i\theta})$ as a function of $\theta$ to argue that the integrand of Equation~\ref{eq:Cauchy} can only be large at small angle. Consider the ratio $\abs{\frac{G(e^{-u+i\theta})}{G(e^{-u})}} = \abs{\frac{1-p+pe^{-u+i\theta}}{1-p+pe^{-u}}}^R$, which compares $|G|$ at angle $\theta$ to $|G|$ at angle $0$. Taking the logarithm, we have
    \begin{equation*}
        R \log \abs{\frac{1-p+pe^{-u+i\theta}}{1-p+pe^{-u}}} = R \log \abs{1 + \frac{pe^{-u}(e^{i\theta} - 1)}{1-p+pe^{-u}}} \leq -C(c')Rp\theta^2 
    \end{equation*}
    for some constant $C(c')$.
    Since $|F(z_0)G(z_0)| \geq \tau$ and $|F(z_0)| \leq \frac{1}{\tau^c}$, we must have $\log |G(z_0)| \geq (c+1)\log \tau$. Since from above we have $\log |G(e^{-u})| \leq (1+c) \log \frac{1}{\tau}$, combining this with the bound  $\log\abs{\frac{G(e^{-u+i\theta})}{G(e^{-u})}}\leq -C(c')Rp\theta^2$ yields $(c+1)\log \tau \leq \log |G(z_0)| \leq C\cdot(-Rp\theta_0^2 + \log \frac{1}{\tau})$. Solving for $\theta_0$, we have $|\theta_0| \leq C_2\sqrt{\frac{\log \frac{1}{\tau}}{Rp}}=C_2 \cdot u$ for some constant $C_2(c,c')$.  

\medskip\noindent{\bf Poisson kernel argument:}

Recall that $F(z)$ is bounded by $\frac{1}{\tau^c}$; we also showed that $|F(z_0)G(z_0)|\geq \tau$ and $|G(z_0)|\leq \frac{1}{\tau^{1+c}}$, thus $|F(z_0)|\geq \tau^{c+2}$.

Since $F(z)$ is analytic on the annulus of radii between $[e^{-2u},1]$, thus $\log|F(z)|$ is \emph{subharmonic} on the annulus. Thus for $z_0=e^{-u+i\theta_0 }$ we can use the Poisson kernel to upper bound $\log|F(z_0)|$ by an appropriate weighted average of $\log |F(z)|$ around the circles of radii 1 and $e^{-2u}$. For an annulus of radii $[\rho,1]$, 
\cite{Wang1983annulus} shows that, defining

    \begin{equation*}
        P(r,\theta) = \frac{1}{\log\frac{1}{\rho}} \cdot \frac{\cos \left( \frac{\pi}{\log\frac{1}{\rho}} \log \frac{r}{\sqrt{\rho}}\right)}{\cosh \frac{\pi \theta}{\log\frac{1}{\rho}} - \sin \left(\frac{\pi}{\log\frac{1}{\rho}} \log \frac{r}{\sqrt{\rho}}\right)}
    \end{equation*}
    then $p_r(\theta) = \sum_{k=-\infty}^{\infty} P(r,\theta+2\pi k)$ is the \emph{Poisson kernel} on the annulus. This means that, for any subharmonic function $h$ on the annulus, and any radius $r\in (\rho,1)$ and angle $\theta_0$, 
    \begin{equation*}
        h(r\cdot e^{i\theta_0}) \leq \frac{1}{2}\int_0^{2\pi} p_{\rho/r}(\theta-\theta_0) h(\rho e^{i\theta}) d\theta +  \frac{1}{2}\int_0^{2\pi} p_{r}(\theta-\theta_0) h(e^{i\theta}) d\theta
    \end{equation*}

For our particular setting, these equations simplify significantly. Letting $\rho=e^{-2u}$ and $r=e^{-u}$, then $P(r,\theta)=\frac{1}{2u}\frac{1}{\cosh\frac{\pi\theta}{2u}}$. Thus, letting $h$ be our subharmonic function $\log|F|$ we have 
\begin{equation}\label{eq:poisson}
(c+2)\log\tau\leq\log|F(e^{-u+i\theta_0 })|\leq\frac{1}{2}\int_{-\pi}^{\pi} \left(\log|F(e^{-2u+i\theta})|+\log|F( e^{i\theta})|\right)p_r(\theta-\theta_0)\,d\theta
\end{equation}

We point out that $\int_{-\infty}^\infty P(r,\theta)\,d\theta=1$ so that $\int_{-\pi}^\pi p_r(\theta)\,d\theta=1$ so that Equation~\ref{eq:poisson} can be interpreted as saying that $\log|F|$ on the interior of the annulus is bounded by a weighted average of its values on the boundary.

We now show a lower bound for $p_r(\theta-\theta_0)$ when $\theta\in [-c'u,c'u]$. Thus $|\theta-\theta_0|\leq (C_2+c')u$; thus we bound $p_r(\theta-\theta_0)$ by the $k=0$ term, since $1/\cosh$ is positive and decreasing away from 0, as $P(r,\theta-\theta_0)\geq \frac{1}{2u}\frac{1}{\cosh \frac{\pi (C_2+c')u}{2u}}$, which we represent as $\frac{C_3}{u}$.

Markov's inequality says that, since $\log|F(e^{-2u+i\theta})|,\log|F( e^{i\theta})|\leq c\log\frac{1}{\tau}$ then, it \emph{cannot} be the case that $> 1-c''$ fraction of $\theta\in [-c'u,c'u]$ have $\log|F( e^{i\theta})|\leq \frac{2c+2}{(1-c'') c' C_3} \log\tau$, because then the right hand side of Equation~\ref{eq:poisson} would be $<(c\log\frac{1}{\tau})+\frac{1}{2}\frac{C_3}{u}\cdot (1-c'')\cdot |[-c'u,c'u]|\cdot\frac{2c+2}{(1-c'') c' C_3} \log\tau=(c+2)\log\tau$ while actually it is $\geq (c+2)\log\tau$.

Finally, we have  $|F_b(e^{i\theta})|=|F(e^{i\theta})|\cdot |(1-p)+p\cdot e^{i\theta}|^R$, and, for $\theta\in [-c'u,c'u]$, and $p\in (0, \frac{1}{4}]$ we bound $|(1-p)+p\cdot e^{i\theta}|^R\geq e^{-p\theta^2R\cdot C_4}\geq \tau^{{c'}^2C_4}$ for some constant $C_4$,
proving the lemma.
\end{proof}

The next lemma shows how to convert bounds on the fraction of small Fourier elements over a continuous domain into a corresponding statement in a discrete domain. This result is needed to translate between the conclusion of Lemma~\ref{lem:mostly-big-laurent}, which applies to the Fourier transform evaluated on a continuous interval, and the assumptions of Lemma~\ref{lem:contrapositive}, which involve the discrete Fourier transform.

\begin{lemma}\label{lem:fourier-continuous-to-discrete}
Let $f:\mathbb{Z}\rightarrow \mathbb{R}$ be supported on $\{0,\ldots,W-1\}$. Define the Fourier transform $F(\xi)=\sum_j f(j) e^{\xi j i}$. Suppose that for an interval $I$ we have that for $\geq \alpha$ fraction of angles $\xi\in I$ we have $|F(\xi)|\geq \tau$. Then for any positive integer $N$, the number of integers $\ell\in\{0,\ldots,N-1\}$ for which $\frac{2\pi}{N} \ell\in I$ and $|F(\frac{2\pi}{N} \ell)|\geq \tau$ is at least $\frac{N}{2\pi}\alpha|I|-W$.
\end{lemma}
\begin{proof}
Since $f$ is real-valued, we have $|F(\xi)|^2=F(\xi)F(-\xi)$. Define the Laurent polynomial $Q(z)=(\sum_{j=0}^{W-1} f(j) z^j)(\sum_{j=0}^{W-1} f(j) z^{-j})$ so that $Q(e^{\xi i})=|F(\xi)|^2$. The set $S=\{\xi\in I:|F(\xi)|\geq \tau\}$ thus is a union of intervals whose endpoints are either solutions to the equation $Q(e^{\xi i})=\tau^2$ or one of the two endpoints of $I$; for $z=e^{\xi i}$, multiplying both sides by $z^{W-1}$ yields an ordinary polynomial equation of degree $2(W-1)$, implying that the set $S$ has at most $2W$ endpoints and thus consists of a union of at most $W$ intervals. Since for any single interval, the number of integer points in the interval is at least the length of the interval minus 1: scaling the domain and the set $S$ by $\frac{N}{2\pi}$ and summing over the $\leq W$ constituent intervals immediately yields the desired conclusion.
\end{proof}

We are now prepared to prove the main result of this section: if $f-g$ ``starts abruptly'' and $f$ is far from any cyclic shift of $g$, then the ``binomially blurred" versions of $f,g$ differ in some statistic of order at most $3$.

\begin{lemma}\label{lem:combined-contrapositive}
Given an integer $R\geq 250$ and sequences $f,g:\mathbb{Z}\rightarrow\mathbb{R}_{\geq 0}$ supported on $\{-R,\ldots,R\}$ with $\|f\|_1,\|g\|_1\leq \frac{1}{2}$, and a binomial distribution $\Bin{R}{p}$ with $p\in (0,\frac{1}{4}]$, then define the binomially blurred versions of $f,g$ by convolving with the binomial: let $f_b:=f\ast \Bin{R}{p}$ and $g_b:=g\ast \Bin{R}{p}$. Then, if there exists a bound $\tau\in (0,\frac{1}{4}]$ and a location $\ell\leq pR+\sqrt{Rp\log\frac{1}{\tau}}$  such that $|f_b(\ell)-g_b(\ell)|\geq\tau$, and if $\sqrt{\frac{\log\frac{1}{\tau}}{Rp}}\leq \frac{\pi}{8}$ and the left tail of $f-g$ is controlled as $\sum_j |f(j)-g(j)|\cdot e^{-2j\sqrt{(\log\frac{1}{\tau})/(Rp)}}\leq \frac{1}{\tau^c}$ for some constant $c\geq 1$, then: 
\begin{itemize}
    \item[\textup{1)}] Either there exist offsets $\ell_0=0,\ell_1\geq 0,\ell_2\geq 0$ such that $|\sum_j f_b(j+\ell_0)f_b(j+\ell_1)f_b(j+\ell_2)-g_b(j+\ell_0)g_b(j+\ell_1)g_b(j+\ell_2)|\geq 1/poly_c(\frac{1}{\tau},R)$, or the analogous product of two terms ignoring $\ell_2$;
    \item[\textup{2)}] Or there exists $\alpha\in\mathbb{R}$ such that ``$f_b$ is approximately a shifted version of $g_b$'' in the sense that, defining the Fourier transforms $F_b(\xi)=\sum_j f_b(j)\cdot  e^{ij\xi}$ and $G_b$ defined correspondingly in terms of $g_b$, then for all angles $\xi$ that are integer multiples of $\frac{2\pi}{R^2}$ we have $|F_b(\xi)-G_b(\xi)\cdot e^{i\xi \alpha}|\leq \tau^2$.
\end{itemize}

\end{lemma}
\begin{proof}
Apply Lemma~\ref{lem:mostly-big-laurent} to $f-g$, for $c''=0.95$ and $c'=8$ with $c_{lem}=c$ to conclude that, letting $\gamma:=c'\sqrt{\frac{\log\frac{1}{\tau}}{Rp}}$, we have that at least a $0.95$ fraction of the frequencies $\xi\in[-\gamma,\gamma]$ have $|F_b(\xi)-G_b(\xi)|\geq 1/poly_c(\frac{1}{\tau})$.

In order to apply Lemma~\ref{lem:contrapositive} to $f_b,g_b$ we first need to convert the conclusions about the Fourier transforms $F_b,G_b$ to a discrete domain, which we do by applying Lemma~\ref{lem:fourier-continuous-to-discrete} to $f_b-g_b$. The domain size of $f_b-g_b$ is $W_{lem}:=3R+1$. We choose a domain size for the lemma $N_{lem}:=R^2$, which for $R\geq 250$ satisfies $N_{lem} > \frac{20\pi W_{lem} }{ \gamma}=20\pi (3R+1) \frac{\sqrt{Rp}}{c'\sqrt{\log\frac{1}{\tau}}}$. Then we apply Lemma~\ref{lem:fourier-continuous-to-discrete}, checking that $\gamma\frac{N_{lem}}{2\pi}\geq R>22$ satisfies the input condition for $\gamma$, to yield that $\geq 0.9$ fraction of the integers $\theta\in[-\gamma\frac{N_{lem}}{2\pi},\gamma\frac{N_{lem}}{2\pi}]$ have $|F_b(\theta\frac{2\pi}{N_{lem}})-G_b(\theta\frac{2\pi}{N_{lem}})|\geq 1/poly_c(\frac{1}{\tau})$.

We thus apply Lemma~\ref{lem:contrapositive} to $f_b,g_b$ padded with zeros to have size $N_{lem}$, with $\tau_{lem}$ taken to be $\min(1/poly_c(\frac{1}{\tau}),\tau^2)$ from the above bounds; we note that $f_b,f_g$ are pointwise bounded by 1. The conclusion of Lemma~\ref{lem:contrapositive} is phrased as ``if A then B'' which we rephrase as ``either not A, or B'':

\begin{itemize}
    \item[\textup{1)}] Either there exist indices $\ell_0,\ell_1,\ell_2$ for which $|\sum_j f_b(j+\ell_0)f_b(j+\ell_1)f_b(j+\ell_2)-g_b(j+\ell_0)g_b(j+\ell_1)g_b(j+\ell_2)|\geq 1/poly'_c(\frac{1}{\tau},R)$, where without loss of generality, we may take $\ell_0=0$ and $\ell_1,\ell_2\geq 0$; or the corresponding second order variant holds;
    \item[\textup{2)}] Or there exists $\alpha\in\mathbb{R}$ such that for all angles $\xi\in [-\frac{1}{2}\gamma,\frac{1}{2}\gamma]$ that are integer multiples of $\frac{2\pi}{R^2}$ we have $|F_b(\xi)-G_b(\xi)\cdot e^{i\xi \alpha}|\leq \tau^2$.
\end{itemize}

Finally, we use tail bounds on $F_b,G_b$ to extend the conclusion of the second case to \emph{all} frequencies $\xi$. Since $\|f\|_1,\|g\|_1\leq\frac{1}{2}$, we define $F,G$ to be the Fourier transforms of $f,g$ and have that $|F(\xi)|,|G(\xi)|\leq \frac{1}{2}$ everywhere. In the Fourier domain, convolving by $bin(R,\cdot,p)$ has the effect of elementwise multiplying the magnitude of each Fourier transform entry $|F(\xi)|$ by $|(1-p)+p\cdot e^{i\xi}|^R$, which for $p\leq \frac{1}{4}$ and $\xi\in[-\pi,\pi]$ is at most $e^{-\frac{p}{5}\xi^2 R}$. Thus $|F_b(\xi)|,|G_b(\xi)|\leq \frac{1}{2}e^{-\frac{p}{5}\xi^2 R}$. Recalling that $\gamma=c'\sqrt{\frac{\log\frac{1}{\tau}}{Rp}}$ we have that, when $|\xi|\geq \frac{1}{2}\gamma$ we have $|F_b(\xi)|,|G_b(\xi)|\leq \frac{1}{2}e^{-\frac{p}{5}\xi^2 R}\leq \frac{1}{2}\tau^{\frac{1}{5}(\frac{1}{2})^2{c'}^2}$, which is smaller than $\frac{1}{2}\tau^2$ for $c':=8$. Thus for all angles $\xi$ that are integer multiples of $\frac{2\pi}{R^2}$ we have $|F_b(\xi)-G_b(\xi)\cdot e^{i\xi \alpha}|\leq \tau^2$.
\end{proof}

The following corollary continues from Case 2 of Lemma~\ref{lem:combined-contrapositive} to show how to produce a statistic that distinguishes the blurred sequences $f_b,f_g$. We show that a linear, low-frequency weighting of $f_b,g_b$ suffices.

\begin{corollary}\label{cor:shift}
Under the assumptions of Lemma~\ref{lem:combined-contrapositive}, if Conclusion 2) applies, then we can further conclude that \[\Big|\sum_{j} (f_b(j)-g_b(j))\cdot e^{ij\frac{2\pi}{R^2}}\Big|\geq \frac{\tau}{R^2}-\tau^2\]
\end{corollary}
\begin{proof}
For the sake of this proof, we redefine $F_b,G_b$ to be the \emph{discrete} Fourier transforms of $f_b,g_b$ respectively on the discrete domain of size $R^2$. This differs from the $F^{lem}_b,G^{lem}_b$ used in Lemma~\ref{lem:combined-contrapositive} by a scaling of the input: $F_b(\xi)=F^{lem}_b(\frac{2\pi}{R^2}\xi)$; so thus here we take $\alpha$ to be $\frac{2\pi}{R^2}\alpha_{lem}$ in terms of the $\alpha_{lem}$ of Lemma~\ref{lem:combined-contrapositive} so that for all integers $\xi$ we have $|F_b(\xi)-G_b(\xi)\cdot e^{i\xi \alpha}|\leq \tau^2$.

We have, from the assumptions of Lemma~\ref{lem:combined-contrapositive}, that
$|f_b(\ell)-g_b(\ell)|\geq\tau$; since $f_b-g_b$ is the inverse Fourier transform of $F_b-G_b$, then the average value of $|F_b(\xi)-G_b(\xi)|$ must be at least $\tau$, and thus there must exist a particular $\xi'$ for which $|F_b(\xi')-G_b(\xi')|\geq \tau$. Combined with the conclusion of Case 2) of Lemma~\ref{lem:combined-contrapositive} that $|F_b(\xi')-G_b(\xi')\cdot e^{i\xi' \alpha}|\leq \tau^2$, we have that $|G_b(\xi')(1-e^{i\xi'\alpha})|\geq \tau-\tau^2$.

We now analyze $G_b(1)$, the lowest nonzero frequency entry of the Fourier transform. From the proof of Lemma~\ref{lem:combined-contrapositive}, recall that $f_b,g_b$ are supported on a contiguous domain of size $3R+1$ that contains the point $\ell$, and since $R\geq 250,$ we thus have that all points in the support of $f_b-g_b$ are within $\frac{R^2}{20}$ of $\ell$. Now, every point in the support's coefficient $e^{ij\frac{2\pi}{R^2}}$ to Fourier entry $G_b(1)$, when projected to direction $e^{i\ell\frac{2\pi}{R^2}}$ is $\cos((j-\ell)\frac{2\pi}{R^2})\geq\cos(\frac{2\pi}{20})>0.9$. Thus, since $G_b$ is the Fourier transform of a nonnegative real sequence and thus $G_b(0)$ equals simply the sum of the entries of $g$, we have $|G_b(1)|\geq 0.9G_b(0)\geq 0.9 |G_b(\xi')|$. Because $\xi'$ is an integer mod $R^2$, its equivalence class contains an integer of magnitude at most $\frac{R^2}{2}$; thus $|1-e^{i\xi'\alpha}|\leq |\xi'|\cdot |1-e^{i\alpha}|\leq \frac{R^2}{2}|1-e^{i\alpha}|$.
Thus $|G_b(1)|\cdot|1-e^{i\alpha}|\geq 0.9|G_b(\xi')|\cdot\frac{2}{R^2}|1-e^{i\xi'\alpha}|\geq \frac{1.8(\tau-\tau^2)}{R^2}$.

The fact that $F_b$ is $\tau^2$ close to the $\alpha$-shifted $G_b$ means that $|F_b(1)-e^{i\alpha}G_b(1)|\leq \tau^2$. Combining this with the bound of the previous paragraph, the triangle inequality yields that $|F_b(1)-G_b(1)|\geq \frac{1.8(\tau-\tau^2)}{R^2}-\tau^2$. Thus
\[\Big|\sum_{j} (f_b(j)-g_b(j))\cdot e^{ij\frac{2\pi}{R^2}}\Big|\geq \frac{1.8(\tau-\tau^2)}{R^2}-\tau^2\geq \frac{\tau}{R^2}-\tau^2\]
as desired.

\end{proof}

\section{Simulating Multiple Traces from One Trace}\label{sec:alpha}

Suppose the expected value of some statistic $s$ of order $k$ differs significantly between traces from $x$ and $y$. In~\Cref{sec:induction}, we will apply Lemma~\ref{lem:combined-contrapositive} to (weighted versions of) the sequences $x^{(s+j)}_{[d-R:n],p},y^{(s+j)}_{[d-R:n],p}$, to show that there exists a summed $\leq 3$rd order statistic of these sequences on which the two sequences differ. Notice that this is a statistic \emph{of a sequence of statistics}, of the form
\begin{equation}
    \label{eq:third-order}
    \sum_{j} x^{(s+j)}_{[d-R:n],p} x^{(s+j+\ell_1)}_{[d-R:n],p}x^{(s+j+\ell_2)}_{[d-R:n],p}
\end{equation}
for two fixed offsets $\ell_1,\ell_2$. However, the goal of the inductive step is not to produce a summed product of three expectations, but a \emph{single} low-order statistic of a trace. To bridge this gap, in this section, we show how to convert a summed third order statistic of $x^{(s+j)}_{[d-R:n],p}$ into a linear combination of low-order statistics, at the cost of increasing the retention probability by a factor of up to $3$: $P\in [p,3p]$. In particular, we will convert Equation~\ref{eq:third-order} into a weighted sum $\sum_{j,S} A_S(j) x^{(S+j)}_{[d-R:n],P}$ of  statistics $S$ of order at most $3k$. Most of this section is occupied with analyzing the coefficient sequence $A_S(j)$, finding sufficient conditions under which it will be an extremely smooth function of $j$, with bounded support, captured by Lemma~\ref{lem:fourier-alpha}.

\subsection{Simulation}

We begin with the following observation, which simulates three traces with retention probability $p$ using a single trace with retention probability at most $3p$. Since the three way product $x^{(s_1)}_{[d-R:n],p} x^{(s_2)}_{[d-R:n],p}x^{(s_3)}_{[d-R:n],p}$ is the product of expectations across three \emph{independent} traces, the below result is a fundamental transformation, showing, suprisingly, that we can reexpress this as an expectation of a function of a single trace.

\begin{lemma}
    \label{lem:simulation}
    Three independent deletion channels with deletion probability $1-p$ can be simulated from one deletion channel with deletion probability $(1-p)^3$.
\end{lemma}
\begin{proof}
    Given a binary string $x$, for each location $j$, the three independent deletion channels will flip one coin each of bias $1-p$ and keep bit $j$ according to the results of these coins. The probability that at least one of the deletion channels keeps bit $j$ is $1-(1-p)^3$. Thus we can simulate all three deletion channels by conditioning on this event: store via a new deletion channel only those bits kept by at least one of the three channels (which occurs with retention probability $1-(1-p)^3$ ); and then simulate the rest with fresh randomness.  
\end{proof}

As a simple consequence of the above, each term in Equation~\ref{eq:third-order} can be written as a linear combination of statistics of order at most $3k$. 

\begin{lemma}\label{lem:alpha}
Let $P=1-(1-p)^3$. Given three statistics $s_1,s_2,s_3$ we can express $x_p^{(s_1)}x_p^{(s_2)}x_p^{(s_3)}$ as a linear combination $\sum_S \alpha(S) x_P^{(S)}$ with nonnegative coefficients $\alpha_{s_1,s_2,s_3}(S)$.
\end{lemma}
\begin{proof}
    Each quantity in the product is the expectation of a statistic $s_i$ on a trace $U_i$ with retention probability $p$. Simulating independent traces $U_1,U_2,U_3$ by a single trace $U$ with retention probability $1-(1-p)^3$ by Lemma~\ref{lem:simulation}, each bit of $U_1,U_2,U_3$ must correspond to some bit in $U$. Defining $\alpha(S)$ to be the probability that $S$ equals the union of the bits in $U$ that became the bits in $U_1$ of indices $s_1$; the bits in $U$ that became the bits in $U_2$ of indices $s_2$; and the bits in $U$ that became the bits in $U_3$ of indices $s_3$, the lemma follows.
\end{proof}

We now study the coefficients $\alpha(S)$ for each statistic $S$ applied to a trace $U \sim \Del{P}{x}$. Considering only those statistics $S$ which start at the beginning of the trace, i.e., $\min_i S_i = 1$, we consider the coefficients for each shift of $s_1,s_2,s_3,S$. We show that the coefficients $\alpha$ of different shifts can be related to each other by a somewhat simpler function $\beta$.

\begin{lemma}\label{lem:beta}
Let $P=1-(1-p)^3$. Let $s_1,s_2,s_3$ be three order-$k$ statistics each with minimum index $\min_i s_1(i)=\min_i s_2(i)=\min_i s_3(i)=1$. We can express $x_p^{(s_1+j_1)}x_p^{(s_2+j_2)}x_p^{(s_3+j_3)}$ as a mixture $\sum_{j\geq 0, S:\min_i S_i=1} \alpha_{s_1+j_1,s_2+j_2,s_3+j_3}(S+j) x_P^{(S+j)}$ where each $S$ is a statistic of order $\leq 3k$; the coefficients are nonnegative and satisfy \[\alpha_{s_1+j_1,s_2+j_2,s_3+j_3}(S+j)=\sum_{\ell_1,\ell_2,\ell_3:\min(\ell_1,\ell_2,\ell_3)=0} \alpha_{s_1+\ell_1,s_2+\ell_2,s_3+\ell_3}(S) \beta_p(j_1-\ell_1,j_2-\ell_2,j_3-\ell_3,j)\]
where we define $\beta_p(j_1,j_2,j_3,j)$ to be the probability that when using $j$ bits of a $P$-trace to simulate three $p$-traces, the number of retained bits from the three traces equal $j_1,j_2,j_3$ respectively.
\end{lemma}
\begin{proof}
    Consider simulating three $p$-traces $U_1,U_2,U_3$ with a single $P$-trace $U$.
    By \Cref{lem:alpha}, the coefficient $\alpha_{s_1+j_1,s_2+j_2,s_3+j_3}(S+j)$ is the probability that the union of the bits at indices $s_1+j_1,s_2+j_2,s_3+j_3$ of the three $p$-traces is the set of bits at $S+j$ in $U$. This event requires that bit $j+1$ of $U$ corresponds to either bit $j_1+1$ of $U_1$, $j_2+1$ of $U_2$, or $j_3+1$ of $U_3$. Conditioned on the number of the first $j$ bits of $U$ that are retained by $U_1,U_2,U_3$ respectively being $j_1-\ell_1,j_2-\ell_2,j_3-\ell_3$, (with $\min(\ell_1,\ell_2,\ell_3) = 0$), the probability that the union of the three sets is $S+j$ is $\alpha_{s_1+\ell_1,s_2+\ell_2,s_3+\ell_3}(S)$. The probability of the event we conditioned on was defined as $\beta_p(j_1-\ell_1,j_2-\ell_2,j_3-\ell_3,j)$.    
    Summing over all possible tuples $(\ell_1,\ell_2,\ell_3)$ gives the desired result. 
\end{proof}

We will also require a basic support bound on the sequence $\alpha(S+j)$. 

\begin{lemma}\label{lem:alpha-range}
For statistics $s_1,s_2,s_3,S$ each with minimum index 1, the coefficient $\alpha_{s_1+j_1,s_2+j_2,s_3+j_3}(S+j)$ is 0 unless $\min(j_1,j_2,j_3)\leq j\leq j_1+j_2+j_3$.
\end{lemma}
\begin{proof}
    Consider the process described by Lemma~\ref{lem:simulation}, simulating three $p$-traces $U_1,U_2,U_3$ with a single $P$-trace $U$. Bit $1+j_1$ of $U_1$ must correspond to a bit of $U$ at position $\geq 1+j_1$, and similarly for bits $1+j_2$ of $U_2$ and $1+j_3$ of $U_3$. Therefore, the first bit contained in $S+j$---which is the $1+j$th bit of $U$---cannot appear in $s_1,s_2$ or $s_3$ if $j_1,j_2,j_3 > j$. 

    To show the upper bound on $j$, consider the locations in $U$ before the queried bits in $S+j$: each of the bits $1,\ldots,j$ of $U$ must appear in at least one of $U_1,U_2,U_3$. But before the queried locations $s_m+j_m$ in $U_m$, trace $m$ can have at most $j_m$ retained bits. Thus $j\leq j_1+j_2+j_3$, giving the claimed upper bound. 
\end{proof}

\subsection{Fourier Properties}

Now that we understand the functions $\alpha$ and $\beta_p$ relating a three way product of expectations to a single expectation, we show they satisfy certain Fourier properties. In this subsection, our goal is to show that the sequence $\beta_p$, weighted by a reasonably well-behaved function, has bounded Fourier transform at large frequencies. 

To begin, we write down the generating function of $\beta_p$. (Recall that $\beta_p$ is the pdf of a joint probability distribution over three variables $j_1,j_2,j_3$.)

\begin{lemma}
The probability $\beta_p(j_1,j_2,j_3,j)$ equals the coefficient of $z_1^{j_1}z_2^{j_2}z_3^{j_3}$ in the formal power series $\left(1-\frac{1}{P}+\frac{1}{P}\prod_{m=1}^3((1-p)+pz_m)\right)^j$, letting $P = 1-(1-p)^3$.
\end{lemma}
\begin{proof}
    The allocation of a single bit to a single $p$-trace $m \in \{1,2,3\}$ follows a Bernoulli distribution with generating function $(1-p) + pz_m$. The generating function of the joint distribution over all three $p$-traces is therefore $\prod_{m=1}^3 (1 - p + pz_m)$. The probability that the bit is allocated to zero traces is the constant term $1 - P = (1-p)^3$. To condition on the complement of this event, we set the constant term to zero and normalize the remaining coefficients: 
    \begin{equation*}
        \frac{1}{P} \left( \prod_{m=1}^3 (1-p + pz_m) - (1-P) \right) = 1 - \frac{1}{P} + \frac{1}{P} \prod_{m=1}^3 (1-p + pz_m)
    \end{equation*}
    Taking the $j$th power of this expression corresponds to repeating the process for $j$ trials.  
\end{proof}

We now establish several somewhat technical tools for studying the Fourier transform of $\beta_p$. 

\begin{lemma}\label{lem:beta-expression}
For a frequency $\xi$ and nonnegative integer $j$ we have

\[\hspace{-1.5cm}
\sum_{\ell \in\mathbb{Z}}e^{i\xi\ell}\beta_p(j_1+\ell,j_2+\ell,j_3+\ell,j)=\frac{1}{(2\pi)^2}\int_{\substack{\xi_1,\xi_2\in[0,2\pi]\\\xi_3\equiv \xi-(\xi_1+\xi_2) \mod 2\pi}}e^{-i(j_1 \xi_1+j_2\xi_2+j_3\xi_3)}\left(1-\frac{1}{P}+\frac{1}{P}\prod_{m=1}^3((1-p)+p e^{i \xi_m})\right)^j\,d\xi_1\,d\xi_2\]
\end{lemma}
\begin{proof}
This is basically the projection-slice theorem, where the left hand side is a Fourier component along the diagonal line $(j_1,j_2,j_3)+\ell$, and the right hand side is the Fourier transform along the dual plane $\xi_1+\xi_2+\xi_3\equiv\xi \pmod{2\pi}$ of the 3-dimensional Fourier transform of $\beta_p(\cdot,\cdot,\cdot,j)$ described in the previous lemma.
\end{proof}

\begin{lemma}\label{lem:close-angle}
For angles $\xi_1,\xi_2,\xi_3$ and $\eps \leq \frac{1}{4}$, letting $y=e^{i\xi_1}+e^{i\xi_2}+e^{i\xi_3}$, if $|y|\geq 3-\eps$, then the sum of the angles $\xi_1+\xi_2+\xi_3$ is within $\frac{\eps}{3}$ of $3\cdot angle(y)$ (interpreting angles mod $2\pi$).
\end{lemma}
\begin{proof}
Without loss of generality we rotate everything so that $y$ is real (with angle 0). Then our assumptions state that $\sin(\xi_1)+\sin(\xi_2)+\sin(\xi_3)=0$ and that $3-\cos(\xi_1)-\cos(\xi_2)-\cos(\xi_3)\leq \eps$. Thus for each of $\xi_1,\xi_2,\xi_3$ we have $1-\cos(\xi_i)\leq \eps\leq\frac{1}{4}$ meaning that $|\xi_i|\leq 0.8$. We have for any angle $|\xi|\leq 0.8$ that $|\xi-\sin(\xi)|\leq \frac{1}{3}(1-\cos(\xi))$; so thus applying the triangle inequality on each of $\sin(\xi_1)$, $\sin(\xi_2)$, and $\sin(\xi_3)$, and adding, we conclude $|\xi_1+\xi_2+\xi_3|\leq \frac{\eps}{3}$
\end{proof}
\begin{lemma}\label{lem:low-frequency-or-decay}
Let $p\in(0,\frac{1}{4}]$ and let $\xi\in[-\pi,\pi]$. Given $\xi_1+\xi_2+\xi_3\equiv \xi \pmod{2\pi}$, then for any $\tau\leq p$  the quantity $z=\frac{\left(\prod_{m=1}^3((1-p)+p e^{i \xi_m})\right)-(1-p)^3}{1-(1-p)^3}$ either has radius $\leq 1-\frac{\tau}{4}$ or has angle with magnitude $\leq |\xi|+\tau$; and $|z|\leq 1$ always.
\end{lemma}
\begin{proof}
We will assume that the first possibility is violated, that $|z|\geq 1-\frac{\tau}{4}$, and use it to either conclude that the second condition is satisfied, or derive a contradiction.

Define $y=e^{i\xi_1}+e^{i\xi_2}+e^{i\xi_3}$. Then reexpress $z$ as the following expression of $y$ and its complex conjugate: \begin{equation}\label{eq:z-three-terms}z=\frac{(1-p)^2y+p(1-p)\overline{y}e^{i\xi}+p^2e^{i\xi}}{3-3p+p^2}\end{equation}
and thus $|z|\leq \frac{|y|(1-p)+p^2}{3-3p+p^2}=1-(3-|y|)\frac{1-p}{3-3p+p^2}$ and thus, since $\frac{1-p}{3-3p+p^2}\geq \frac{1}{4}$ we have that if $|z|>1-\frac{\tau}{4}$ then $|y|\geq 3-\tau$. We thus apply Lemma~\ref{lem:close-angle} to conclude that $\xi$ is within $\tau$ of $3 \cdot angle(y)$, interpreted mod $2\pi$. 
In the case that the angles do not ``wrap around'', so that for $\xi\in[-\pi,\pi]$ we have $|\xi-3\cdot angle(y)|\leq  \tau$, then each of the three complex coefficients in the numerator of Equation~\ref{eq:z-three-terms}, $y, \bar{y}e^{i\xi},e^{i\xi}$, have angles of magnitudes respectively $\leq \frac{1}{3}|\xi|+\frac{1}{3}\tau, \leq \frac{2}{3}|\xi|+\frac{1}{3}\tau, |\xi|$ respectively, which all have magnitude at most $|\xi|+\tau$.

On the other hand, in the ``wrap around'' case we have, for some integer $j\in\{-2,-1,1,2\}$ that $|j\cdot 2\pi+\xi-3\cdot angle(y)|\leq  \tau$; since $\xi\in[-\pi,\pi]$ we have that $|j\cdot 2\pi+\xi|\geq \pi$. The first two complex coefficients above, $y$ and  $\bar{y}e^{i\xi}$, have angles $angle(y)$ and $\xi-angle(y)$ respectively, with a difference of $\xi-2\cdot angle(y)$, which, multiplying our $\tau$ bound by $\frac{2}{3}$ has magnitude at least $|\xi-2\cdot angle(y)|\geq \frac{2}{3}\pi-\frac{1}{3}|\xi|-\frac{2}{3}\tau$. Since $|\xi|\leq\pi$, and $\tau\leq p\leq\frac{1}{4}$, this angle between $y$ and $\bar{y}e^{i\xi}$ is at least $\frac{1}{3}\pi-\frac{1}{6}$, which has cosine $<0.64$. 

From the law of cosines, for two vectors of lengtha $a,b$ respectively and angle $\theta$ between them, we have that the length of their sum equals $\sqrt{a^2+b^2+2ab\cos(\theta)}=\sqrt{(a+b)^2-2ab(1-\cos(\theta))}$ which is at most $a+b-(1-\cos(\theta))\frac{a b}{a+b}$. Thus we can upper bound $|z|$ by upper bounding the length of the sum of the terms of Equation~\ref{eq:z-three-terms} in directions $y,\bar{y}e^{i\xi}$ and adding to this the length of the third term, $\frac{p^2}{3-3p+p^2}$. We thus bound the length of the sum of the first two terms by substituting in $a=(1-p)^2\frac{|y|}{3-3p+p^2}$ and $b=p(1-p)\frac{|y|}{3-3p+p^2}$ and $1-\cos(\theta)\geq 0.36$ to get $|z|\leq \frac{|y|}{3-3p+p^2}\cdot\left(p(1-p)+(1-p)^2- 0.36\frac{p(1-p)^3}{p(1-p)+(1-p)^2}\right) +\frac{p^2}{3-3p+p^2}$. Simplifying, using $|y|\leq 3$, yields $|z|\leq 1-0.36 p \frac{3(1-p)^2
}{3-3p+p^2}$; the fraction $\frac{3(1-p)^2
}{3-3p+p^2}$ is decreasing in $p$, as can be seen by rewriting it in terms of $q=1-p$ as $\frac{3}{1+1/q+1/q^2}$ which is clearly increasing for $q\in[\frac{3}{4},1]$ and hence decreasing in $p\in [0,\frac{1}{4}]$. Thus since, when $p=\frac{1}{4}$ we have $\frac{3(1-p)^2
}{3-3p+p^2}\geq 0.72$, yielding that $|z|\leq 1-0.36 \cdot 0.72\cdot p\leq 1-\frac{1}{4}p\leq  1-\frac{\tau}{4}$, contradicting our assumption that $|z|>1-\frac{\tau}{4}$.
\end{proof}

The below lemma lets us translate bounds on polynomials on the unit circle to bounds inside the unit circle, using a basic subharmonicity argument.

\begin{lemma}\label{lem:sqrt-bound}
Let $\lambda,\delta,\eps> 0$, with $\eps\leq \frac{\pi}{2}$. Let $f(j)$ be a function from the nonnegative integers to real numbers such that $\sum_{j\geq 0} |f(j)|\leq \lambda$, and for any angle $\xi$ with $|\xi|\geq \eps$ we have $\left|\sum_{j\geq 0} e^{i j \xi} f(j)\right|\leq \delta$. Then for any real $y\in[0,1]$ and any angle $\xi$ with $|\xi|\geq \eps$ we have \[\left|\sum_{j\geq 0} e^{i j \xi} f(j) y^j\right|\leq \sqrt{\lambda\delta}\]
\end{lemma}
\begin{proof}
Without loss of generality we assume $\delta\leq \lambda$ in the proof, for otherwise we set $\delta=\lambda$ without changing the assumption, and strengthening the conclusion.

Define the polynomial $P(z)=\sum_{j\geq 0} f(j) z^j$, which converges on the unit disk in the complex plane. The lemma's conclusion is equivalent to saying that $|P(z)|\leq\sqrt{\lambda\delta}$ when $|z|\leq 1$ and $|angle(z)|\geq \eps$.

To prove this, we point out that the assumption, translated to apply to $P$, states that for inputs on the unit circle, $|P(e^{i \xi})|\leq \lambda$ always, and is at most $\delta$ when $|\xi|\geq\eps$. Further, since $P(z)$ is analytic on the unit disk, thus $\log|P(z)|$ is subharmonic on the unit disk. We thus upper bound $\log|P(z)|$ by its average over the unit circle, weighted by the Poisson kernel. We use the fact (which follows from a basic symmetry argument) that, for $\eps\in[0,\frac{\pi}{2}]$ and $z=re^{i\theta}$, then if $|\theta|\geq \eps$ then the harmonic measure from $z$ of the arc of angles $[-\eps,\eps]$ is at most $\frac{1}{2}$. 
Thus we conclude $\log|P(z)|\leq\frac{1}{2}\log\delta+\frac{1}{2}\log\lambda$; exponentiating gives the desired result.
\end{proof}

\begin{lemma}\label{lem:fourier-beta2}
Let $j_1,j_2,j_3$ be integers with at least one nonnegative, and let $p\in(0,\frac{1}{4}]$. Let $w:\mathbb{Z}\rightarrow\mathbb{R}_{\geq 0}$ have sum $\leq 1$, be supported on a positive interval $[I_-,I_+]$ and, letting $\hat{w}$ be the Fourier transform of $w$, for parameters $\eps\in(0,p], \delta\geq 0$ we assume, for $|\xi|\geq \eps$ that $|\hat{w}(\xi)|\leq \delta$. Then for frequency $|\xi'|\geq 3\eps$ we have, letting $\lambda=3I_++j_1+j_2+j_3$ that
\begin{equation}\label{eq:fourier-beta}
\left|\sum_j \left(e^{i\xi'j}\sum_\ell w(\ell)\beta_p(j_1+\ell,j_2+\ell,j_3+\ell,j)\right)\right|\leq \lambda\delta+\max(\lambda,\frac{8\pi}{\eps})\cdot e^{-\eps I_-/12}
\end{equation}
if $I_-<\lambda$, and otherwise, the left hand side is 0.
\end{lemma}
\begin{proof}

Define the interval $I^{wide}=[\frac{1}{3}I_-\,,3I_++\frac{2}{3}I_-+j_1+j_2+j_3]$ and define the function $w':=B_{I^{wide},\eps,\frac{2}{3}I_-}$ by Lemma~\ref{lem:cumulative-epsilon}, which thus has value 1 on the interval $[I_-,3I_++j_1+j_2+j_3]$, has sum $3I_++j_1+j_2+j_3-\frac{1}{3}I_-$, and for angles $|\xi|\geq \eps$ has Fourier transform of magnitude $\leq \frac{8\pi}{\eps} e^{-I_-\eps/6}$. 

Defining $f(j)=\sum_\ell w(\ell)\beta(j_1+\ell,j_2+\ell,j_3+\ell,j)$, we point out that $\beta$ from its definition as a probabilistic process is nonzero only when its last argument is at least each of its first three arguments and at most their sum. Thus for $B(j_1+\ell,j_2+\ell,j_3+\ell,j)$ to be nonzero we need $j\in [\ell,3\ell+j_1+j_2+j_3]$; since $w(\ell)$ is nonzero only for $\ell\in [I_-,I_+]$ we have that $f(j)$ is nonzero only when $j\in [I_-,3I_++j_1+j_2+j_3]$---and thus, if this interval is empty, the left hand side of Equation~\ref{eq:fourier-beta} equals 0. Otherwise, recall that $w'$ was defined to equal 1 exactly on this range, so that the left hand side of Equation~\ref{eq:fourier-beta}, reexpressed as $\sum_j e^{i\xi'j} f(j)$, equals $\sum_j e^{i\xi'j} w'(j) f(j)$.

Lemma~\ref{lem:beta-expression} says that, letting $\hat{w}(\xi)$ be the Fourier transform of $w$, so that $w(\ell)=\frac{1}{2\pi}\int_{\xi\in[0,2\pi]} \hat{w}(\xi)e^{-i\xi\ell}\,d\xi$, thus $f(j)= \frac{1}{2\pi}\int_{\xi\in[0,2\pi]} \hat{w}(\xi) \sum_\ell e^{i\xi\ell}\beta(j_1+\ell,j_2+\ell,j_3+\ell,j) \,d\xi$, and thus, defining $z_{\xi_1,\xi_2,\xi_3}=1-\frac{1}{P}+\frac{1}{P}\prod_{m=1}^3((1-p)+p e^{i \xi_m})$, then we have \[f(j)=\int_{\xi\in[0,2\pi]}\frac{\hat{w}(\xi)}{(2\pi)^3}\int_{\substack{\xi_1,\xi_2\in[0,2\pi]\\\xi_3\equiv \xi-(\xi_1+\xi_2) \mod 2\pi}}e^{-i(j_1 \xi_1+j_2\xi_2+j_3\xi_3)}(z_{\xi_1,\xi_2,\xi_3})^j\,d\xi_1\,d\xi_2\,d\xi\]

Thus since the left hand side of Equation~\ref{eq:fourier-beta} equals $\sum_{j\geq N} e^{i\xi'j} w'(j) f(j)$, we reexpress this from the above expression for $f(j)$, bringing the sum and the Fourier coefficient inside:
\begin{equation}\label{eq:fourier-beta-integral}
\int_{\xi\in[0,2\pi]}\frac{\hat{w}(\xi)}{(2\pi)^3}\int_{\substack{\xi_1,\xi_2\in[0,2\pi]\\\xi_3\equiv \xi-(\xi_1+\xi_2) \mod 2\pi}}e^{-i(j_1 \xi_1+j_2\xi_2+j_3\xi_3)}\sum_{j\geq N} e^{i\xi'j} w'(j)(z_{\xi_1,\xi_2,\xi_3})^j\,d\xi_1\,d\xi_2\,d\xi\end{equation}

We bound the inner sum using Lemma~\ref{lem:sqrt-bound}. We first reexpress $z_{\xi_1,\xi_2,\xi_3}=y e^{i \xi''}$ for $y\in[0,1]$, letting us merge the terms $e^{i\xi'j}e^{i\xi'' j}$ into effectively a single angle $\xi'+\xi''$; we also shift $j$ by $-I_-/3$, since $w'$ is supported on inputs $\geq I_-/3$; shifting $j$ like this changes the magnitude by a factor of $y^{I_-/3}$ which we account for separately. Thus, by Lemma~\ref{lem:sqrt-bound}, we have $|\sum_{j\geq 0} e^{i\xi'j} w'(j+I_-/3)(z_{\xi_1,\xi_2,\xi_3})^{j+I_-/3}|\leq y^{I_-/3} \sqrt{\lambda\frac{8\pi}{\eps}e^{-\eps I_-/6}}$ when $|\xi'+\xi''|\geq\eps$, and by $y^{I_-/3}\lambda$ in general.

When the outer integration variable satisfies $|\xi|> \eps$ we have $|\hat{w}(\xi)|\leq \delta$, and we use the above bounds, trivially bounding $y\leq 1$. The contribution of such $\xi$ to Equation~\ref{eq:fourier-beta-integral} (which equals the left hand side of Equation~\ref{eq:fourier-beta} that we seek to bound) we thus bound by $\delta$ times our general bound $y^{I_-/3}\lambda$ from above, whose product we bound by $\delta \lambda$,
 for the contribution from 
$|\xi|> \eps$.

On the other hand, for $|\xi|\leq \eps$ we apply Lemma~\ref{lem:low-frequency-or-decay} which says that $y\leq 1$ and for any $\tau\in [0,p]$ we have either $y\leq 1-\frac{\tau}{4}$ or $|\xi''|\leq |\xi|+\tau$. We choose $\tau=\eps$, since by assumption, $\eps\leq p$.

Recall we assume $|\xi'|\geq 3\eps$. If the case of Lemma~\ref{lem:low-frequency-or-decay} where $|\xi''|\leq |\xi|+\tau\leq 2\eps$ applies, we have that $|\xi'+\xi''|\geq \eps$, in which case our bound $y^{I_-/3} \sqrt{\lambda\frac{8\pi}{\eps}e^{-\eps I_-/6}}\leq \sqrt{\lambda\frac{8\pi}{\eps}e^{-\eps I_-/6}}$ applies. Else, in the other case of Lemma~\ref{lem:low-frequency-or-decay} we have $y\leq 1-\frac{\tau}{4}\leq e^{-\tau/4}=e^{-\eps/4}$ and we have the bound $y^{I_-/3}\lambda\leq e^{-\eps I_-/12}\lambda$. Multiplying by the bound $|\hat{w}(\xi)|\leq 1$ and averaging yields a bound of the max of these two prior bounds, which we bound as $\max(\lambda,\frac{8\pi}{\eps})e^{-\eps I_-/12}$, leading to our final bound.
\end{proof}

\subsection{Main Lemma}
The below lemma defines a sequence of coefficients $A(j)$ that express via the definition of $\alpha$, the coefficient of the statistic $S+j$ needed to reexpress a triple product of shifted versions of statistics $s_1,s_2,s_3$.
This lemma shows that, because of the properties of the coefficients $\alpha$, we can effectively transfer bounds on the Fourier smoothness of the weight function $w$ to Fourier bounds on the coefficients $A$. Thus in the context of the induction step of Proposition~\ref{prop:induction}, we can ensure smoothness of the coefficients $A(j)$ of our statistic $S+j$ by picking a smooth weight function $w$ at the beginning.

\begin{lemma}\label{lem:fourier-alpha}
Let $j_1,j_2,j_3$ be nonnegative integers and let $p\in(0,\frac{1}{4}]$. Let $w:\mathbb{Z}\rightarrow\mathbb{R}_{\geq 0}$ have sum $\leq 1$, be supported on a nonnegative interval $[I_-,I_+]$ and, letting $\hat{w}$ be the Fourier transform of $w$, for parameters $\eps\in(0,p], \delta\geq 0$ we assume, for $|\xi|\geq \eps$ that $|\hat{w}(\xi)|\leq \delta$. Let $s_1,s_2,s_3,S$ be statistics each with minimum index $\min_i s_1(i)=\min_i s_2(i)=\min_i s_3(i)=\min_i S_i=1$.  Define $A(j):=\sum_{\ell} w(\ell)\alpha_{s_1+j_1+\ell,s_2+j_2+\ell,s_3+j_3+\ell}(S+j)$. Then, letting $\lambda=3I_++j_1+j_2+j_3$ we have that $A$ is supported in the interval $[I_-,\lambda]$, sums to $\leq 1$, and for frequency $|\xi'|\geq 3\eps$ we have:

\begin{equation}\label{eq:fourier-alpha}\left|\sum_{j} e^{i\xi' j}A(j)\right|\leq 3\max_i(S_i)^2\left(\lambda\delta+\max(\lambda,\frac{8\pi}{\eps})\cdot e^{-\eps I_-/12}\right)
\end{equation}
while for general $\xi'$ the expression is bounded by $1$.
\end{lemma}

\begin{proof}
We bound the sum of $A$ by $\sum_{j,\ell} w(\ell)\alpha_{s_1+j_1+\ell,s_2+j_2+\ell,s_3+j_3+\ell}(S+j)$ from the definition of $A$. The $\alpha$ coefficients represent different probabilities from a probabilistic process as we vary $j$, so thus for fixed $\ell$ the sum over $j$ is at most $w(\ell)$; and since $\sum_\ell w(\ell)\leq 1$ we get the intended bound of $\leq 1$.

For the case of frequency $|\xi'|\geq 3\eps$ we reexpress the left hand side of Equation~\ref{eq:fourier-alpha} using Lemma~\ref{lem:beta} as

\[\left|\sum_{j,\ell} e^{i\xi' j}w(\ell)\sum_{\ell_1,\ell_2,\ell_3:\min(\ell_1,\ell_2,\ell_3)=0} \alpha_{s_1+\ell_1,s_2+\ell_2,s_3+\ell_3}(S) \beta_p(j_1+\ell-\ell_1,j_2+\ell-\ell_2,j_3+\ell-\ell_3,j)\right|\]

We crudely bound each $\alpha\leq 1$ since they represent probabilities; we move the sum over $\ell_1,\ell_2,\ell_3$ outside, and apply Lemma~\ref{lem:fourier-beta2}---since $\min(\ell_1,\ell_2,\ell_3)=0$ and all of $j_1,j_2,j_3$ are nonnegative, we have that at least one of $j_1-\ell_1,j_2-\ell_2,j_3-\ell_3$ is nonnegative, as required by the input conditions of Lemma~\ref{lem:fourier-beta2}. We then bound the range of the triple sum over $\ell_1,\ell_2,\ell_3$ by $3\max_i(S_i)^2$, since the $\alpha$ coefficient will be 0 if any of the queried indices are outside the range of $S$, or if all of $\ell_1,\ell_2,\ell_3$ are strictly positive. Thus we conclude with the bound of Lemma~\ref{lem:fourier-beta2} times $3\max_i(S_i)^2$, as claimed; we note that $\lambda_{lem}$ in the context of Lemma~\ref{lem:fourier-beta2} might be smaller than $\lambda$ here (since in the lemma we decrease $j_1,j_2,j_3$ by $\ell_1,\ell_2,\ell_3$ respectively), which just leads to a stronger bound.

Finally, we bound the support of $A(j)$ by invoking Lemma~\ref{lem:alpha-range} to see that $\alpha_{s_1+j_1+\ell,s_2+j_2+\ell,s_3+j_3+\ell}(S+j)$ is 0 unless $\ell+\min(j_1,j_2,j_3)\leq j\leq 3\ell+j_1+j_2+j_3$. Combined with the definition $w(\ell)$ is nonzero only when $I_-\leq \ell\leq I_+$, we conclude that these conditions can both be satisfied only when $j\in [I_-,3I_++j_1+j_2+j_3]$ as claimed.
\end{proof}

\section{Deconvolution}\label{sec:deconvolution}
This section describes how to \emph{deconvolve} by $\Bin{N}{p}$ in the sense that, given a function $f:\mathbb{Z}\rightarrow \mathbb{R}$, we want to find a function $g$ such that $g\ast \Bin{N}{p}\approx f$. However, as discussed in Section~\ref{sec:overview}, we need extremely tight control of both the support of $g$ and its Fourier decay. While convolving with $\Bin{N}{p}$ improves a function's Fourier decay since the Fourier transform of $\Bin{N}{p}$ decays analogously to a Gaussian, deconvolution is famously ill-conditioned, since in the Fourier domain, deconvolution corresponds to dividing by a rapidly decaying function. Paradoxically, we aim to achieve what sounds impossible: bounding the support of deconvolution by essentially the same bound (up to constants) as one would expect for convolution with $\Bin{N}{p}$.

Explicitly, Corollary~\ref{cor:deconvolution} says that, given an accuracy bound $\delta$, and provided the Fourier transform of $f$ decays to size $\delta$ for similar frequencies as $\Bin{N}{p}$, then we construct a deconvolution $g$ so that $g\ast \Bin{N}{p}$ is $O(\delta)$-close to $f$, and $g$ will be supported on an interval $O(r)$ larger than the support of $f$, where $r$ is the size of the interval on which $\Bin{N}{p}\geq \delta$.

The following lemma bounds the $m^\textrm{th}$ derivative of the inverse of the Fourier transform of $\Bin{N}{p}$.

In the bound below, the initial term $\frac{m!e^m}{m^m}$ is polynomial in $m$ by Stirling's approximation and should be ignored; the final term is essentially the square of the inverse Gaussian approximation of the function we are differentiating; the interesting term is the one in the middle, taken to the $m^\textrm{th}$ power, expressing essentially the length scale at which the derivative grows, independent of the common inverse-Gaussian factor.

\begin{lemma}\label{lem:inverse-derivative}
For $p\in [0,\frac{1}{4}]$ and $x\in [-\frac{\pi}{4},\frac{\pi}{4}]$, the $m^{\textrm{th}}$ derivative of $\frac{1}{((1-p)e^{-i p x}+p\cdot e^{i (1-p)x})^N}$ has magnitude at most $\frac{m!e^m}{m^m}(e^3\max(2.5 m,\sqrt{mNp},|x|Np))^{m}e^{pNx^2}$.
\end{lemma}

\begin{proof}
We extend our function analytically to the complex plane and use the Cauchy formula for the derivative: 
\[\frac{d^m}{dx^m}\frac{1}{((1-p)e^{-i p x}+p\cdot e^{i (1-p)x})^N}=\frac{m!}{2\pi r^m}\int_{0}^{2\pi } \frac{e^{-i m\theta}}{((1-p)e^{-i p (x+r\cdot e^{i\theta})}+p\cdot e^{i (1-p)(x+r\cdot e^{i\theta})})^N}\,d\theta\]
provided that the denominator is nonzero on the disk of radius $r$ about $x$.

Reexpressed slightly, suppose for some positive $r$ we can get a lower bound $\eps_r>0$, over all complex $z$ within distance $r$ of $x$ of the magnitude of the function $(1-p)e^{-ipz}+p\cdot e^{i(1-p)z}$, then this gives us a corresponding overall upper bound for our derivative of $\frac{m!}{r^m} \eps_r^{-N}$. We derive such a bound now and then choose $r$.

Letting $g(z)=(1-p)e^{-ipz}+p\cdot e^{i(1-p)z}$, its derivative is $g'(z)=i p(1-p)(e^{i (1-p)z}-e^{-i p z})$. Because the arguments of the two exponentials, $i(1-p)z$ and $-ipz$, are separated by $iz$ (which has magnitude at most $|x|+r$) and because both arguments have real part at most $r$ (since $z$ is within distance $r$ of the real number $x$, so $real(i(1-p)z)\leq r$ and $real(-ipz)\leq r$) so we have the bound $|g'(z)|\leq p (|x|+r)e^r$, for any $z$ within distance $r$ of $x$. Thus $\eps_r=\min_{z:|z-x|\leq r}|g(z)|\geq |g(x)|-p r (|x|+r) e^r$.

Since the real part of $g(x)$ for real arguments trivially has second derivative $\geq -p$, we have $|g(x)|\geq 1-\frac{p}{2} x^2$. Thus $\eps_r\geq 1-p(\frac{1}{2}x^2 +r(|x|+r)e^r)$. For $r\leq 0.4$ and our bounds $|x|\leq\frac{\pi}{4}$, $p\leq \frac{1}{4}$, this is at least $0.7$ and since $1-\alpha\geq e^{-1.3\alpha}$ when $1-\alpha\in[0.7,1]$ we thus have $\eps_r\geq e^{-1.3 p(\frac{1}{2}x^2 +r(|x|+r)e^r)}$. Trivially bounding $e^r\leq 1.5$ for $r\leq 0.4$, and since $1.3\cdot 1.5<2$, we thus change the initial multiplier from $1.3$ to 2, giving $\eps_r\geq e^{-p(x^2+2r(|x|+r))}$, yielding our bound on the derivative of $\frac{m!}{r^m}e^{pN(x^2+2r(|x|+r))}$, for any $r\in [0, 0.4]$.

We chose $r=\min(0.4,\sqrt{\frac{m}{Np}},\frac{m}{Np|x|})$. Then in all cases, $Npr^2\leq m$ and $Npr|x|\leq m$. Thus the quantity inside the exponential term in our bound can be bounded as $pN(x^2+2r(|x|+r))\leq pNx^2+2m+2m$, which simplifies our overall bound to $\frac{m!}{r^m}e^{4m}e^{pNx^2}$. Thus since $\frac{1}{r}=\max(2.5,\sqrt{\frac{Np}{m}},\frac{Np|x|}{m})$, we thus conclude our overall bound of $\frac{m!e^m}{m^m}(e^3\max(2.5 m,\sqrt{mNp},|x|Np))^{m}e^{pNx^2}$.
\end{proof}

Below we bound the derivatives of a standard construction known as the cardinal B-spline.

\begin{lemma}\label{lem:spline}
For any integer $m\geq 1$ there is a function $B_m(x)$ that has value $0$ for $x\leq 0$, has value $1$ for $x\geq 1$, takes values $\in[0,1]$ for $x\in[0,1]$, and for any integer $r\in \{0,\ldots,m\}$ has $r^{\textrm{th}}$ derivative bounded in magnitude by $(2m)^r$.
\end{lemma}
\begin{proof}
Let $f_m(x)$ be the function that has value $m$ for $x\in[0,\frac{1}{m}]$ and 0 otherwise. Consider the $m$ way convolution of $f_m$ with itself, which we denote $f_m^{*m}$, which we may regard as the pdf of a distribution. Let $B_m$ be the cdf of this distribution; namely, let $B'(x)=f_m^{*m}(x)$. 

We point out that convolving $f_m$ with itself any number of times cannot increase its maximum value, which was originally $m$, since $f_m$ is a distribution so convolving a function with a distribution returns a weighted average of its values at each point, and thus cannot increase its maximum.

Consider the derivative of $f_m$, which is a signed measure, which we denote $g_m$, consisting of $m$ times a delta function at $x=0$ minus $m$ times a delta function at $x=\frac{1}{m}$. Thus the $L_1$ norm of this signed measure is $2m$. Since convolution commutes with differentiation, we now express the $r^{\textrm{th}}$ derivative of $B_m$ as the $r-1^{\textrm{st}}$ derivative of $f_m^{*m}(x)$, which we express by absorbing one derivative into each of $r-1$ different copies of $f_m$, yielding $g_m^{*r-1}*f_m^{*m-r+1}$. The $f_m^{*m-r+1}$ term has maximum value $\leq m$; and this is scaled by at most $2m$ for each of the $r-1$ copies of $g_m$ we convolve with, giving us our final bound of $m(2m)^{r-1}$, which we harmlessly round up to $(2m)^r$ for simplicity.
\end{proof}

The following lemma combines the two previous results to define a kernel that we use to deconvolve, via its Fourier transform $H$: we define $H$ to be the inverse of the Fourier transform of the pdf of $\Bin{N}{p}$, but mollified by multiplying by the smooth cutoff function $B$ from Lemma~\ref{lem:spline}.

\begin{lemma}\label{lem:H-def}
Given a positive integer $N$, a probability $p\in (0,\frac{1}{4}]$, a frequency cutoff $\eps>0$, and two parameters defining a smooth cutoff: a positive integer order $m$, and a width $W>0$ such that $\eps+W\leq\frac{\pi}{4}$, define the function $H(x)=\frac{B_m(\frac{\eps+W-|x|}{W})}{((1-p)e^{-i p x}+p\cdot e^{i (1-p)x})^N}$ where $B_m$ is the smooth cutoff function from Lemma~\ref{lem:spline} and the denominator is the Fourier transform of the binomial distribution $\Bin{N}{p}$, centered to have mean 0. We claim that: the numerator of $H(x)$ equals 1 when $x\in[-\eps,\eps]$, and $H(x)$ is 0 when $x\notin [-(\eps+W)\,,\eps+W]$; and $|H^{(m)}(x)|\leq e\sqrt{m}\cdot e^{pN(\eps+W)^2} \left((\frac{2m}{W})+(e^3\max(2.5 m,\sqrt{mNp},(\eps+W) Np))\right)^m$ everywhere. And $|H(x)|\leq e^{pN|x|^2}$, and 0 outside of its support, $x\in[-(\eps+W)\,,\eps+W]$.
\end{lemma}

\begin{proof}
The first two properties follow from Lemma~\ref{lem:spline} where we rescale the input of $B_m$ so that $B_m(\frac{\eps+W-|x|}{W})$ smoothly goes from 1 to 0 over the range $\eps$ to $\eps+W$, and symmetrically from $-\eps$ to $-(\eps+W)$, and takes values in $[0,1]$ everywhere.

Define $g(x)=\frac{1}{((1-p)e^{-i p x}+p\cdot e^{i (1-p)x})^N}$, namely, $H$ without the $f$ numerator term.
From Lemma~\ref{lem:inverse-derivative} we have that $|g^{(r)}(x)|\leq\frac{r!e^r}{r^r}(e^3\max(2.5 r,\sqrt{rNp},|x|Np))^{r}e^{pNx^2}$. By Stirling's approximation, $\frac{r!e^r}{r^r}\leq \max(1,e\sqrt{r})$; for $r\leq m$ this is at most $e\sqrt{m}$.

From Lemma~\ref{lem:spline} we have that $|B_m^{(m-r)}(x)|\leq (2m)^{m-r}$, so that the corresponding derivative of $B_m(\frac{\eps+W-|x|}{W})$ has magnitude bounded by $(\frac{2m}{W})^{m-r}$. 

In general, $\frac{d^m}{dx^m} f(x)g(x)=\sum_{r=0}^m \binom{m}{r}f^{(m-r)}(x)g^{(r)}(x)$ and we thus bound 
\begin{align*}|H^{(m)}(x)| &\leq e\sqrt{m}\cdot e^{pNx^2} \sum_{r=0}^m \binom{m}{r}(\frac{2m}{W})^{m-r}(e^3\max(2.5 m,\sqrt{mNp},|x|Np))^{r}\\
&=e\sqrt{m}\cdot e^{pNx^2} \left((\frac{2m}{W})+(e^3\max(2.5 m,\sqrt{mNp},|x|Np))\right)^m\end{align*}
In the first step we used $r\leq m$ to replace $r$ in the $g^{(r)}$ derivative expression by its upper bound $m$. For the final conclusion, since $H$ and all its derivatives vanish outside $x\in[-(\eps+W),\eps+W]$, we thus replace $|x|$ by its upper bound $\eps+W$ to yield our final expression.

We bound $|H(x)|$ via the ``0 derivative'' bound of Lemma~\ref{lem:inverse-derivative}.
\end{proof}

The following lemma uses the fact that the $m^{\textrm{th}}$ moment of a function becomes the $m^{\textrm{th}}$ derivative of its Fourier transform; thus we can use bounds on high derivatives of the Fourier transform to get good tail bounds on the function.

\begin{lemma}\label{lem:h-local}
Given a function $H:\mathbb{R}\rightarrow\mathbb{C}$ and a real number $\rho$ such that $H(\xi)e^{-i \rho \xi}$ is periodic of period $2\pi$, having Fourier series $h:\mathbb{Z}\rightarrow\mathbb{C}$, and for which we have a bound $\|H^{(m)}\|_{\infty}$ on the $m^{\textrm{th}}$ derivative of $H$ everywhere then, given a radius $\ell$ we have 
\[\forall j\notin (-\rho-\ell,-\rho+\ell),\;|h(j)|\leq \frac{\|H^{(m)}\|_{\infty}}{\ell^m}\]
\end{lemma}

\begin{proof}
Given a function $G(\xi)$ such that $G(\xi)e^{-i\rho\xi}$ is periodic with period $2\pi$, and letting the Fourier series of $G(\xi)e^{-i\rho\xi}$ be $g(j)$, we claim that the Fourier series of $G'(\xi)e^{-i\rho\xi}$ equals $(j+\rho)i\cdot g(j)$. To show this, we compute the Fourier series of $G'(\xi)e^{-i\rho\xi}$ via integration-by-parts in terms of the Fourier series of $G(\xi)e^{-i\rho\xi}$ which is $g(j)$ by definition: at location $j$ the Fourier series of $G'(\xi)e^{-i\rho\xi}$ is \[\frac{1}{2\pi}\int_{-\pi}^\pi G'(\xi)e^{-i\rho\xi} e^{-ij\xi}\,d\xi=\left.\frac{1}{2\pi}G(\xi)e^{-i\rho\xi} e^{-ij\xi}\right|_{-\pi}^\pi-(-\rho-j)i\cdot\frac{1}{2\pi}\int_{-\pi}^\pi G(\xi)e^{-i\rho\xi} e^{-ij\xi}\,d\xi=(j+\rho)i\cdot g(j)\]
where the first term vanishes because $G(\xi)e^{-i\rho\xi}$ is periodic with period $2\pi$.

Thus we can repeatedly apply this relation $m$ times to $H$ (since $H$ when multiplied by $e^{-i\rho\xi}$ has period $2\pi$, all its derivatives have the same period) to see that the Fourier series of $H^{(m)}(\xi)e^{-i\rho\xi}$ equals $((j+\rho)i)^m\cdot h(j)$.

The Fourier series of $H^{(m)}(\xi)e^{-i\rho\xi}$ (which equals $((j+\rho)i)^m\cdot h(j)$), has magnitude bounded by $\|H^{(m)}\|_{\infty}$ and thus \[|(j+\rho)^m\cdot h(j)|\leq  \|H^{(m)}\|_{\infty}\]

When $j\notin (-\rho-\ell,-\rho+\ell)$ then $|j+\rho|\geq \ell$. So thus
\[|h(j)|\leq \frac{\|H^{(m)}\|_{\infty}}{\ell^m}\]
\end{proof}

The goal of the following lemma is to approximately deconvolve by the binomial distribution $\Bin{N}{p}$, in a way that lets us control the support and the error.

\begin{lemma}\label{lem:deconvolution}
Let $f:\mathbb{Z}\rightarrow\mathbb{C}$ supported on an interval $[I_-,I_+]$ with $\sum_{j=I_-}^{I_+} |f(j)| \leq 1$. Further, letting $F:[-\pi,\pi]\rightarrow\mathbb{C}$ be the Fourier transform of $f$, we assume that, for a given $\eps,\delta>0$ we have that $|F(\xi)|\leq\delta$ when $|\xi|\geq \eps$. Let $\Bin{N}{p}$ be a binomial distribution with $p\in (0,\frac{1}{4}]$. Define $h:\mathbb{Z}\rightarrow\mathbb{R}$ to be the Fourier series of $H(\xi)e^{-ipN\xi}$ for $H$ defined in Lemma~\ref{lem:H-def} (depending on associated parameters $\ell,m,W$, where $\eps+W\leq\frac{\pi}{4}$) and let $h_\ell(j)=h(j)\cdot\mathbbm{1}_{[|j+pN|< \ell]}$. Then $f\ast h_\ell$ is ``a local deconvolution of $f$ with $\Bin{N}{p}$'' in the following sense: $f\ast h_\ell$ is supported on $[I_--pN-\ell,I_+-pN+\ell]$, with the pointwise bound $|f\ast h_\ell|\leq  \max(e^{pN\eps^2}, e^{pN(\eps+W)^2}\delta)+\frac{ \|H^{(m)}\|_{\infty}}{\ell^m}$, and, when convolved with $\Bin{N}{p}$ is close to $f$ in the sense $\|f\ast h_\ell\ast \Bin{N}{p}-f\|_\infty\leq \delta+\frac{\|H^{(m)}\|_{\infty}}{\ell^m}$
\end{lemma}

\begin{proof}

We first bound $\|f\ast h\ast \Bin{N}{p}-f\|_\infty\leq \delta$ by analyzing in the Fourier domain; and we will then use the triangle inequality to relate this to the corresponding expression involving $h_\ell$. Recall that $\Bin{N}{p}$ is the sequence of Fourier coefficients of $((1-p)+p\cdot e^{i \xi})^N$. Also, by definition, $\widehat{h}(\xi):=H(\xi)e^{-i p N \xi}=\frac{B_m(\frac{\eps+W-|\xi|}{W})}{((1-p)+p\cdot e^{i \xi})^N}$, for $B_m$ as defined in Lemma~\ref{lem:spline}. Thus we have $\widehat{h}(\xi)\widehat{\Bin{N}{p}}(\xi)=B_m(\frac{\eps+W-|\xi|}{W})$. Thus the Fourier transform of the 3-way convolution $f\ast h\ast \Bin{N}{p}$ exactly equals the Fourier transform of $f$ at all frequencies $|\xi|\leq \eps$. For frequencies of magnitude $>\eps$, we point out that the Fourier transform of $f-f\ast h\ast \Bin{N}{p}$ equals $F(\xi)(1-B_m(\frac{\eps+W-|\xi|}{W}))$ where $1-B_m(\frac{\eps+W-|\xi|}{W})\in [0,1]$ and thus this error term has magnitude $\leq \delta$ in the Fourier domain; thus $\|f\ast h\ast \Bin{N}{p}-f\|_{\infty}\leq \delta$.

We bound $\|h-h_\ell\|_\infty$ by applying Lemma~\ref{lem:h-local} to our $H$, yielding $\|h-h_\ell\|_\infty\leq \frac{ \|H^{(m)}\|_{\infty}}{\ell^m}$, since $h_\ell$ is defined by truncating $h$ to 0 outside of $(-pN-\ell,-pN+\ell)$, and Lemma~\ref{lem:h-local} bounds $h$ exactly outside this range. Next, convolution by a function of $L_1$ norm $\leq 1$ preserves bounds on $L_\infty$ norm, so thus $\|(h-h_\ell)\ast f\ast \Bin{N}{p}\|_\infty\leq \frac{ \|H^{(m)}\|_{\infty}}{\ell^m}$. Thus the triangle inequality gives us our final error bound of \[\|f\ast h_\ell\ast \Bin{N}{p}-f\|_{\infty}\leq \delta+\frac{ \|H^{(m)}\|_{\infty}}{\ell^m}\]

Finally, the support bound is trivial since convolution of two functions each bounded on intervals is bounded on the set-sum of the intervals.

The pointwise bound on $|f\ast h_\ell|$ we prove by first bounding $|f\ast h|$ and then adding in our above bound $\|h-h_\ell\|_\infty\leq \frac{\|H^{(m)}\|_{\infty}}{\ell^m}$. The Fourier transform of $f\ast h$ equals the pointwise product of the Fourier transforms of $f$ and $h$. We use Lemma~\ref{lem:H-def} to bound the Fourier transform of $h$ at frequency $\xi$ by $e^{pN|\xi|^2}$ The Fourier transform of $f$, which we defined as $F$, has magnitude $\leq 1$ everywhere; and thus the product of the Fourier transforms of $f$ and $h$, for $|\xi|\leq \eps$, has magnitude at most $e^{pN\eps^2}$. Otherwise, when $|\xi|\in [\eps,\eps+W]$, we have by assumption that $|F(\xi)|\leq \delta$, giving us the bound $e^{pN(\eps+W)^2}\delta$. Thus the Fourier transform of $f\ast h$ is pointwise bounded in magnitude by $\max(e^{pN\eps^2}, e^{pN(\eps+W)^2}\delta)$, meaning that $|f\ast h|$ is pointwise bounded by this as well; adding the initial triangle inequality term gives us our final bound of \[|f\ast h_\ell|\leq  \max(e^{pN\eps^2}, e^{pN(\eps+W)^2}\delta)+\frac{ \|H^{(m)}\|_{\infty}}{\ell^m}\]

\end{proof}

\begin{corollary}\label{cor:deconvolution}
Let $f:\mathbb{Z}\rightarrow\mathbb{C}$ supported on an interval $[I_-,I_+]$ with $\sum_{j=I_-}^{I_+} |f(j)| \leq 1$. Further, letting $F:[-\pi,\pi]\rightarrow\mathbb{C}$ be the Fourier transform of $f$, we assume that, for a given $\eps,\delta>0$ we have that $|F(\xi)|\leq\delta$ when $|\xi|\geq \eps$. 
If for $P\in(0,\frac{1}{4}]$ and positive integer $R'$ we have $e^{-0.16PR'+1}\leq\delta\leq e^{-4PR'\eps^2}<1$ then there exists a $g(x)$ that is ``an $r$-local deconvolution of $f$ with $\Bin{R'}{P}$'' in the following sense: for $r=300\sqrt{R'P\lceil\log\frac{1}{\delta}\rceil}$ we have $g$ is supported on $[I_--R'P-r,I_+-R'P+r]$, with the pointwise bound $|g(x)|\leq 11\cdot e^{R'P\eps^2}$, and, $\|g\ast \Bin{R'}{P}-f\|_\infty\leq 9\delta$.
\end{corollary}

\begin{proof}
Recall our upper bound $\|H^{(m)}\|_{\infty}\leq e\sqrt{m}\cdot e^{R'P(\eps+W)^2} \left((\frac{2m}{W})+(e^3\max(2.5 m,\sqrt{mR'P},(\eps+W) R'P))\right)^m$ from Lemma~\ref{lem:H-def}.

Choose $m=\lceil\log\frac{1}{\delta}\rceil$ and choose $W=\sqrt{\lceil\log\frac{1}{\delta}\rceil/R'P}-\eps$. We point out that by our lower bound on $\delta$ we have $\lceil\log\frac{1}{\delta}\rceil\leq 0.16 R'P$, and thus $\eps+W\leq 0.4\leq\frac{\pi}{4}$ so thus $\eps,W$ satisfy the requirements of Lemma~\ref{lem:H-def}. Then $\max(2.5 m,\sqrt{mR'P},(\eps+W) R'P)=\sqrt{R'P\lceil\log\frac{1}{\delta}\rceil}$ since the last two terms in the $\max$ are equal by definition, and the first term is smaller by our assumption lower-bounding $\delta$. Also, from our assumption upper-bounding $\delta$ we have $\eps\leq \frac{1}{2}\sqrt{(\log\frac{1}{\delta})/R'P}$ and thus $W\geq \frac{1}{2}\sqrt{\lceil\log\frac{1}{\delta}\rceil/R'P}$ and so $\frac{2m}{W}\leq 4\sqrt{R'P\lceil\log\frac{1}{\delta}\rceil}$. Combining, since $R'P(\eps+W)^2=\lceil\log\frac{1}{\delta}\rceil\leq\log\frac{e}{\delta}$ we have $\|H^{(m)}\|_{\infty}\leq \frac{e}{\delta}e\sqrt{m}\left((4+e^3)\sqrt{R'P\lceil\log\frac{1}{\delta}\rceil}\right)^m$

Thus, since $r=300\sqrt{R'P\lceil\log\frac{1}{\delta}\rceil}$ and $\log\frac{300}{4+e^3}\geq 2.5$ we have $\frac{\|H^{(m)}\|_{\infty}}{r^m}\leq e^2\sqrt{\lceil\log\frac{1}{\delta}\rceil}\cdot\delta^{1.5}\leq e^2\delta$. Thus from Lemma~\ref{lem:deconvolution} our final bound is $(1+e^2)\delta\leq 9\delta$.

The pointwise bound of Lemma~\ref{lem:deconvolution} becomes, under our parameter choices and from the above inequalities, \[|f\ast h_r|\leq \max(e^{R'P\eps^2}, e^{R'P(\eps+W)^2}\delta)+\frac{ \|H^{(m)}\|_{\infty}}{r^m}\leq \max(e^{R'P\eps^2},e)+e^2\delta\leq 11\cdot e^{R'P\eps^2}\]
\end{proof}

\section{Induction}
\label{sec:induction} 
The section proves Proposition~\ref{prop:induction}, the main induction step of Theorem~\ref{thm:main}. See Section~\ref{sec:overview} for in-depth discussion and context about how this brings together pieces from all the sections of this paper.

The following is a basic calculation showing how bounds on the Fourier transform combine under pointwise multiplication.

\begin{lemma}\label{lem:product}
Let $f,g:\mathbb{Z}\rightarrow\mathbb{R}$ have Fourier transforms $F,G$ respectively, defined on angles $\xi\in (-\pi,\pi]$, interpreted mod $2\pi$. If $\sum_j |f(j)|\leq 1$ and for all angles $|\xi|\geq \eps$ we have $|F(\xi)|\leq \delta$, and suppose we have $\sum_j |g(j)|\leq 1$ and for all angles $|\xi|\geq \eps'$ we have $|G(\xi)|\leq \delta$. Then the pointwise product $h(j)=f(j)g(j)$ satisfies $\sum_j |h(j)|\leq 1$ and for any angle $|\xi|\geq \eps+\eps'$ we have, letting $H$ be the Fourier transform of $h$, that $|H(\xi)|\leq \delta$. 
\end{lemma}
\begin{proof}
    We note that, since $|f|,|g|$ are both bounded pointwise by $1$, $\sum_j |h(j)| = \sum_j |f(j)||g(j)| \leq \sum_j f(j) \leq 1$. Similarly, $F,G$ are bounded in magnitude by $1$ everywhere. 
    
    The pointwise product in the spatial domain corresponds to convolution in the Fourier domain, so we have $H(\xi) = \frac{1}{2\pi} \int_{-\pi}^{\pi} F(\omega)G(\xi - \omega) d\omega$. If $|\xi| \geq \eps + \eps'$, then for any value of $\omega \in (-\pi,\pi]$, we must have either $|\omega| \geq \eps$ or $|\xi - \omega| \geq \eps'$, so we have $|F(\omega)G(\xi-\omega)| \leq \delta$. Since $H(\xi)$ is the average of this expression over an interval of length $2\pi$, we must have $|H(\xi)| \leq \delta$, as desired. 
\end{proof}

The following lemma shows how, given a triple product of expected statistics with $\tau$ discrepancy between $x,y$, we can construct a single statistic $S$ that, when summed over shifts $j$, weighted by $A(j)$, still has large discrepancy between $x,y$, yet the weights $A(j)$ are extremely smooth and have bounded support.

\begin{lemma}\label{lem:three-to-one}
Let $x,y$ be two binary strings and let $p\in(0,\frac{1}{4}]$ be a retention probability.
Suppose there is a weight function $w:\mathbb{Z}\rightarrow\mathbb{R}_{\geq 0}$ supported on a positive integer interval $[I_-,I_+]$, having sum $\leq 1$, and for frequencies $|\xi|\geq \eps$ has Fourier transform of magnitude $\leq \delta$, for $\eps\in(0,\frac{1}{3}p]$. Suppose there is a $k$-tuple of distinct positive integers $s$ with $\min_i s_i=1$ and $\max_i s_i=\sigma$ such that, for three offsets $\ell_0=0,\ell_1,\ell_2\geq 0$, defining $f(j):=w(j) x_p^{(s+j)}$ and $g(j):=w(j) y_p^{(s+j)}$, with $|\sum_j f(j)f(j+\ell_1)f(j+\ell_2)-g(j)g(j+\ell_1)g(j+\ell_2)|\geq \tau$ or the corresponding product of two terms omitting $\ell_2$, for some $\tau>0$. Defining $\lambda=3I_++\ell_0+\ell_1+\ell_2$ then: there exists a probability $P\in[p,3p]$, a tuple $S$ of $\leq 3k$ distinct positive integers having minimum value 1 and maximum value $\leq 10(I_++\sigma)$, and there exists a weight function $A:\mathbb{Z}\rightarrow\mathbb{R}_{\geq 0}$ that is supported within $[I_-,\lambda]$, sums to $\leq 1$, and where $A$ has Fourier transform of magnitude $\leq 300 (I_++\sigma)^2\cdot\left(\lambda\delta+\max(\lambda,\frac{8\pi}{3\eps})\cdot e^{-\eps I_-/4}\right)$ at frequencies greater than $9\eps$ such that $|\sum_j A(j)(x_P^{(S+j)}-y_P^{(S+j)})|\geq \frac{\tau-3\cdot 2^{-I_+}}{(10(I_++\sigma))^{3k}}$.
\end{lemma}

\begin{proof}

Using Lemma~\ref{lem:beta} we can reexpress a triple product of expected statistics as a convex combination of statistics of order $\leq 3k$, defining $P=1-(1-p)^3$:
\[x^{(s+j)}_p x^{(s+j+\ell_1)}_px^{(s+j+\ell_2)}_p=\sum_{j'\geq 0, S:\min_i S_i=1} \alpha_{s+j,s+j+\ell_1,s+j+\ell_2}(S+j') x_P^{(S+j')}\]

Multiplying by the corresponding triple product of the weight function $w$ and summing over $j$ yields, by definition of $f$, that
\[\sum_j f(j)f(j+\ell_1)f(j+\ell_2)=\sum_{j'\geq 0, S:\min_i S_i=1} x_P^{(S+j')}\sum_j \alpha_{s+j,s+j+\ell_1,s+j+\ell_2}(S+j') \prod_{m=0}^2 w(j+\ell_m)\]
and this is thus at least $\tau$ different from the corresponding mixture of statistics of $y$.

We combine the 3 versions of $w$ into 1 by defining $w^{comb}(j):=\prod_{m=0}^2 w(j+\ell_m)$.

In the subscript of $\alpha$, each statistic $s+j+\ell_m$ is shifted by an amount which appears in the factor $w(j+\ell_m)$ at the end, and thus we may ignore any statistics for which $j+\ell_m$ is outside the interval $[I_-,I_+]$ defined as the support of $w$. As a reminder, entries of $s$ are all $\leq \sigma$ by definition of $\sigma$. Thus the subscript of $\alpha$ effectively consists of statistics whose maximum indices are all at most $I_++\sigma$. Recall the probabilistic process defining $\alpha_{s_0,s_1,s_2}$: for each location in a trace (that we think of as having retention probability $P$), we flip three independent coins of bias $p$ but \emph{condition} on at least one of them being heads; heads in coin $m\in\{0,1,2\}$ means we use that bit in simulated trace $m$. For statistics $s_0,s_1,s_2$ of span $\leq I_++\sigma$, we claim that none of these statistics will use any bits beyond location $10(I_++\sigma)$ unless for the first $10(I_++\sigma)$ coin tosses, at most $I_++\sigma$ heads occur for at least 1 of the 3 coin types. Since at least one of the coins must be heads with each flip, the probability of each single coin being heads is $\geq \frac{1}{3}$; thus Chernoff bounds on binomial tail probabilities easily show that this occurs with probability $\leq 3\cdot 2^{-(I_++\sigma)}$, which we bound by the simpler expression $\leq 3\cdot 2^{-I_+}$.

Thus the difference between

\[\sum_{\substack{j'\geq 0, S:\min_i S_i=1\\\max_i S_i\leq 10(I_++\sigma)}}x_P^{(S+j')}\sum_j \alpha_{s+j,s+j+\ell_1,s+j+\ell_2}(S+j') \cdot w^{comb}(j)\]

and the corresponding expression for $y$ has magnitude at least $\tau-3\cdot 2^{-I_+}$.

We thus use pigeonhole to extract a single $S$ for which the rest of the expression has significant difference: since $S$ consists of $\leq 3k$ (distinct) entries between 1 and $10(I_++\sigma)$, the number of such $S$ is at most $(10(I_++\sigma))^{3k}$. And thus there must exist a single $S$ for which the difference between 
\begin{equation}\label{eq:w-combined}\sum_{j'\geq 0}x_P^{(S+j')}\sum_j \alpha_{s+j,s+j+\ell_1,s+j+\ell_2}(S+j') \cdot w^{comb}(j)\end{equation}
and the corresponding expression of $y$ has magnitude at least $\frac{\tau-3\cdot 2^{-I_+}}{(10(I_++\sigma))^{3k}}$

We now define the coefficients claimed by the lemma: let $A(j'):=\sum_\ell  w^{comb}(\ell)\cdot \alpha_{s+\ell,s+\ell+\ell_1,s+\ell+\ell_2}(S+j')$, so that Equation~\ref{eq:w-combined} is exactly
$\sum_{j'} A(j')x_P^{(S+j')}$, with an analogous expression for  $y$, so that $|\sum_j A(j)(x_P^{(S+j)}-y_P^{(S+j)})|\geq \frac{\tau-3\cdot 2^{-I_+}}{(10(I_++\sigma))^{3k}}$ as claimed.

We now apply Lemma~\ref{lem:fourier-alpha} to ``transfer'' the Fourier bounds on $w^{comb}$ to analogous bounds on the Fourier transform $\hat{A}$.

We bound the Fourier transform of $w^{comb}(j)=\prod_{m=0}^2 w(j+\ell_m)$ by considering it as a 3-way product of shifted versions of $w$, and applying  Lemma~\ref{lem:product} for each product, to conclude that we can retain the same Fourier bounds for $w^{comb}(j)$ of $\delta$, just starting at frequency $3\eps$ instead of $\eps$.

Thus Lemma~\ref{lem:fourier-alpha}, applied for $3\eps$ instead of $\eps$, yields that $A$ sums to $\leq 1$, and letting $\lambda=3I_++\ell_0+\ell_1+\ell_2$ we have that $A$ is supported in the interval $[I_-,\lambda]$, and for any frequency $|\xi|\geq 9\eps$ we have $|\hat{A}(\xi)|\leq 300 (I_++\sigma)^2\cdot\left(\lambda\delta+\max(\lambda,\frac{8\pi}{3\eps})\cdot e^{-\eps I_-/4}\right)$.

The two-product case is identical, using the two-trace analog of the
$\alpha,\beta$ decomposition; equivalently, one may use an empty third
statistic whose expectation is $1$. This only improves the order, support,
and tail bounds shown above, so the same conclusions hold.
\end{proof}

We now prove the main induction step. We break the proof into sub-lemmas for readability.

\begin{proposition}\label{prop:induction}
Given $n$-bit strings $x,y$ where $d=\min\{i:x_i\neq y_i\}$ is the point of first discrepancy, and a retention probability $p\in(0,\frac{1}{12}]$. For integer $R$ larger than some universal constant, if $\tau$ satisfies $4k\log R\leq \log\frac{1}{\tau}\leq R^{0.8}p^2$ and if there exists a statistic $s$ consisting of $k$ integers in $\{1,\ldots,\sigma\}$ with $\min_i s_i=1$, where $\sigma\leq  \sqrt{Rp\log\frac{1}{\tau}}$  
(``$s$ has order $k$ and span $\sigma$'') such that for some integer $\ell\leq 2Rp+ \sqrt{R p\log\frac{1}{\tau}}$ we have $|x_{[d-2R:n],p}^{(s+\ell)}-y_{[d-2R:n],p}^{(s+\ell)}|\geq \tau$ then: there exists an integer $R'\in[3R^{1.1},R^{2}]$ and a probability $P\in [p,3p]$ and a bound $\tau'\geq 1/poly(1/\tau)$ for which there is an order $\leq 3k$ statistic $S$ of span $\leq  \sqrt{\frac{1}{2}R'P\log\frac{1}{\tau'}}$ and $12k\log R'\leq \log\frac{1}{\tau'}\leq ({R'}/2)^{0.8}P^2$ and a location $\ell'\leq R'P+ \sqrt{\frac{1}{2}R' P\log\frac{1}{\tau'}}$ for which $|x_{[d-R-R':n],P}^{(S+\ell')}-y_{[d-R-R':n],P}^{(S+\ell')}|\geq \tau'$.
\end{proposition}

\begin{proof}
Because we assume that $\log\frac{1}{\tau}\leq R^{0.8}p^2$, we can thus say that for $\emph{any}$ constant $c^*>0$ we have $\log\frac{1}{\tau}\leq c^* R p$ for sufficiently large $R$. We will repeatedly use variants of this bound to show that tail probabilities of the form $e^{-c^*Rp}$ are smaller than $\tau$---or even $\tau$ to any constant power---under the assumptions of the proposition: that $R$ exceeds an appropriately chosen universal constant.

As a first use of this fact, we show that $\ell\in [1.95Rp,2.05Rp]$. The upper bound follows from our assumption that $\ell\leq 2Rp+ \sqrt{R p\log\frac{1}{\tau}}$, since from above, $\log\frac{1}{\tau}$ is smaller than any desired constant multiple of $Rp$. On the other hand since $s+\ell$ has a maximum index $\sigma+\ell$, then a difference in $x_{[d-2R:n],p}^{(s+\ell)}$ versus $y_{[d-2R:n],p}^{(s+\ell)}$ can only happen if some bit at location $\geq d$ in the original strings ends up at location $\leq \sigma+\ell$ in the trace; this happens when $<\sigma+\ell$ of the initial $2R$ bits of the substring $x_{[d-2R:n]}$ are retained in the trace, which happens with probability $\leq \Pr[\Bin{2R}{p}<\sigma+\ell]$, which thus bounds $\tau$. Since by assumption $\sigma\leq  \sqrt{Rp\log\frac{1}{\tau}}$, which is smaller than $0.025Rp$ for large enough $R$, thus if $\ell<1.95Rp$ then we have $\tau\leq  \Pr[\Bin{2R}{p}<1.975Rp]$, which violates our assumptions for large enough $R$, since the right hand side decays exponentially with $Rp$.

We now choose a frequency cutoff $\eps=c' \frac{1}{Rp}\log\frac{1}{\tau}$, for a constant $c'$ to be chosen later (see Equation~\ref{eq:constants}), and use this to choose a smooth function $w:=\frac{1}{Rp} B_{[0.6Rp,1.4Rp],\eps,0.3Rp}$ as defined in Lemma~\ref{lem:cumulative-epsilon}, which from Lemma~\ref{lem:cumulative-epsilon} thus is supported on integers in $[0.6Rp,1.4Rp]$, takes values in the range $[0,\frac{1}{Rp}]$ and is equal to $\frac{1}{Rp}$ on integers in $[0.9Rp,1.1Rp]$, has sum $0.5$, and its Fourier transform at angles $|\xi|\geq \eps$ has magnitude at most $\delta:=\frac{8\pi}{\eps} e^{-0.3 Rp \eps/4}=\frac{8\pi}{\eps} \tau^{0.075 c'}$.

Define the nonnegative sequences \[f(j-Rp):=w(j) x^{(s+j)}_{[d-R:n],p} \quad\textrm{and}\quad g(j-Rp):=w(j) y^{(s+j)}_{[d-R:n],p}\] namely, the expected value of any shift $j$ of the statistic $s$, for a deletion channel of retention probability $p$ starting $R$ to the left of the point of first discrepancy $d$, scaled by the weights $w(j)$, and where we shift the sequences $Rp$ to the left. The shift by $-Rp$ makes the expected location of the point of first discrepancy, in a deletion channel starting at $d-R$, end up at location 0 in $f$ and $g$, so thus we effectively ``center $f,g$ around the expected location of the point of first discrepancy in the trace.''

From the support bounds on $w$, we have that the sequences $f,g$ are supported in $[-0.4Rp,0.4Rp]$. And since statistics take values in $[0,1]$, we have $\|f\|_1,\|g\|_1\leq \sum_j w(j) =0.5$. 

We will apply Lemma~\ref{lem:combined-contrapositive} to the sequences
\begin{equation*}
    f_b := f \ast \Bin{R}{p}, \quad g_b := g \ast \Bin{R}{p}
\end{equation*}
Lemma~\ref{lem:combined-contrapositive} requires certain conditions to hold for the sequences $f,g$ and $f_b,g_b$, which we now verify.

\begin{lemma}
    \label{lem:contrapositive-conditions}
    For $R$ large enough, there exists $\ell \in [1.95RP,2.05Rp]$ and $\tau_{lem} := \frac{\tau}{2Rp}$ such that $\abs{f_b(\ell-Rp) - g_b(\ell-Rp)} \geq \tau_{lem}$, and $\sqrt{\frac{\log \frac{1}{\tau_{lem}}}{Rp}} \leq \frac{\pi}{8}$, and \begin{equation}
        \label{eq:left-tail}
        \sum_j |f(j)-g(j)|\cdot e^{-2j\sqrt{(\log\frac{1}{\tau_{lem}})/(Rp)}}\leq poly\left(\frac{1}{\tau_{lem}}\right)
    \end{equation}
\end{lemma}
\begin{proof}
    We relate $f_b-g_b$ to our assumption that $|x_{[d-2R:n],p}^{(s+\ell)}-y_{[d-2R:n],p}^{(s+\ell)}|\geq \tau$ to set up the application of Lemma~\ref{lem:combined-contrapositive}. From Lemma~\ref{lem:binomial-shift} we have that 
    \begin{equation}
        \label{eq:2R-difference}
        x_{[d-2R:n],p}^{(s+\ell)}-y_{[d-2R:n],p}^{(s+\ell)}=e_x-e_y+\sum_{j=0}^R bin(R,j,p)\left(x_{[d-R:n],p}^{(s+\ell-j)}-y_{[d-R:n],p}^{(s+\ell-j)}\right)
    \end{equation}
    where  $0\leq e_x,e_y\leq \Pr[\Bin{R}{p}\geq \ell+1]$. Since, as proven at the start, $\ell\in [1.95Rp,2.05Rp]$, we have that $|e_x-e_y|$ decays exponentially in $Rp$ and thus is at most $\frac{1}{10}\tau$ for large enough $R$. Next, we will compare the above binomial sum with the very similar expression
    \begin{equation}
        \label{eq:smoothed difference}
        f_b(\ell-Rp)-g_b(\ell-Rp)=\sum_{j=0}^R bin(R,j,p) w(\ell-j) \left(x^{(s+\ell-j)}_{[d-R:n],p}-y^{(s+\ell-j)}_{[d-R:n],p}\right)
    \end{equation}
    When $\ell-j\in[0.9Rp,1.1Rp]$ then $w(\ell-j)=\frac{1}{Rp}$, and thus the summand is exactly $\frac{1}{Rp}$ times the summand in Equation~\ref{eq:2R-difference}. The other case, $\ell-j\notin[0.9Rp,1.1Rp]$, since $\ell\in [1.95Rp,2.05Rp]$, is thus possible only when $j\notin [0.95Rp,1.05Rp]$; and thus we bound the total contribution from this case, using the range bound $w(j)\in [0,\frac{1}{Rp}]$, to conclude that Equation~\ref{eq:smoothed difference} is within $\frac{1}{Rp}\Pr[\Bin{R}{p}\notin [0.95Rp,1.05Rp]]$ of the value of Equation~\ref{eq:2R-difference} without the $e_x-e_y$ terms and when multiplied by $\frac{1}{Rp}$. Thus, from Chernoff bounds we have that $|f_b(\ell-Rp)-g_b(\ell-Rp)|\geq \frac{\tau}{2Rp}$ for $R$ above a large enough constant. 

    We point out that $\sqrt{\frac{\log\frac{1}{\tau_{lem}}}{Rp}}\leq \frac{\pi}{8}$ for sufficiently large $R$, satisfying the second condition.

    Finally, we prove Equation~\ref{eq:left-tail}. As argued above, $|f(j)-g(j)|\leq \frac{1}{Rp}\Pr[\Bin{R}{p}<Rp+j+\sigma]$, which, by Chernoff bounds, is at most $\frac{1}{Rp}\cdot e^{-\frac{(j+\sigma)^2}{3Rp}}$ when $j+\sigma\leq 0$. Thus we can bound the $j\leq -\sigma$ portion of the sum of Equation~\ref{eq:left-tail} by $\frac{1}{Rp}\sum_j e^{-\frac{(j+\sigma)^2}{3Rp}}e^{-2j\sqrt{(\log\frac{1}{\tau_{lem}})/(Rp)}}$. We reexpress the exponential by completing the square, getting the equivalent expression $\frac{1}{Rp}\frac{1}{\tau_{lem}^3}e^{2\sigma\sqrt{\frac{\log\frac{1}{\tau_{lem}}}{Rp}}}\sum_j e^{-\frac{\left(j+\sigma+3\sqrt{Rp\log\frac{1}{\tau_{lem}}}\right)^2}{3Rp}}$. The sum here is bounded by $O(\sqrt{Rp})$, since it decays exponentially fast outside an interval of size $\sqrt{Rp}$ around $j=-\sigma-3\sqrt{Rp\log\frac{1}{\tau_{lem}}}$; meanwhile, the multipliers outside the sum are polynomial in $\frac{1}{\tau_{lem}}$, since, in particular, $\sigma\leq  \sqrt{Rp\log\frac{1}{\tau}}$ by assumption, and thus we bound the exponential factor as $e^{2\sigma\sqrt{\frac{\log\frac{1}{\tau_{lem}}}{Rp}}}\leq \frac{1}{\tau_{lem}^{3 }}$, using the fact that $\log\frac{1}{\tau_{lem}}\leq 2\log\frac{1}{\tau}$ for large enough $R$; thus we get a $poly(\frac{1}{\tau_{lem}})$ bound for the portion of Equation~\ref{eq:left-tail} coming from $j\leq -\sigma$.

    For the remaining portion of the sum, we note that $\sum_j |f(j)-g(j)|\leq 1$, and that, for $j>-\sigma$ we may bound the multiplier $e^{-2j\sqrt{(\log\frac{1}{\tau_{lem}})/(Rp)}}\leq e^{2\sigma\sqrt{(\log\frac{1}{\tau_{lem}})/(Rp)}}\leq \frac{1}{\tau_{lem}^{3 }}$ as above, leading to an overall $poly(\frac{1}{\tau_{lem}})$ bound for Equation~\ref{eq:left-tail}.
\end{proof}

We may now apply Lemma~\ref{lem:combined-contrapositive} to $f,g$. We show that both conclusions of the lemma imply the following general statement about summed statistics. 

\begin{lemma}
    \label{lem:coefficients-helper}
    There exists a sequence $A: \mathbb{Z} \rightarrow \mathbb{C}$ supported on $[0.6Rp,5.8Rp]$ with $\|A\|_1 \leq 1$; there exists a statistic $S$ of min value 1, of order $\leq 3k$ and of span $\leq 15Rp$; and there exists a probability $P \in [p,3p]$ such that 
    \begin{equation}\label{eq:combined-discrepancy-bound}
        \abs{\sum_j A(j) (x^{(S+j)}_{[d-R:n],P} - y^{(S+j)}_{[d-R:n],P}} \geq \tau^{c+1}
    \end{equation}
    for some constant $c \geq 1$. Further, the Fourier transform of $A$ has magnitude $\leq \tau^{0.075c'-1}$ for frequencies $|\xi| \geq 9\eps$. 
\end{lemma}
\begin{proof}
    Having checked the conditions in Lemma~\ref{lem:contrapositive-conditions}, we apply Lemma~\ref{lem:combined-contrapositive} to $f,g$ with threshold $\tau_{lem}=\frac{\tau}{2Rp}$, yielding either Conclusion 1 or Conclusion 2 of Lemma~\ref{lem:combined-contrapositive}. We consider both cases. \\
    
    \noindent \textbf{Conclusion 1}: In this case, there exist offsets $\ell_0=0,\ell_1\geq 0,\ell_2\geq 0$ such that $|\sum_j f_b(j+\ell_0)f_b(j+\ell_1)f_b(j+\ell_2)-g_b(j+\ell_0)g_b(j+\ell_1)g_b(j+\ell_2)|\geq \tau^c$, or the analogous double product---where, since we assume $\frac{1}{\tau}\geq R$, we have absorbed the polynomial dependence on $R$ in the lemma into just a power of $\tau$, for $R$ larger than some absolute constant. We assume $c\geq 1$, and otherwise round $c$ up to 1, preserving the inequality. 

    Recalling that $f_b,g_b$ are the convolution of $f,g$ respectively with $\Bin{R}{p}$, thus a triple product of $f_b$ implicitly consists of an average of triple products of $f$ shifted by a triple of integers $j_0,j_1,j_2\sim \Bin{R}{p}$. Thus, we now remove all terms for which any of $j_0,j_1,j_2\notin [0.9Rp,1.1Rp]$, which changes each of $f_b,g_b$ by an amount exponentially small in $Rp$; for sufficiently large $R$ this change is at most $\frac{1}{2}\tau^c$. By the pigeonhole principle, thus there is a single $j_0,j_1,j_2\in [0.9Rp,1.1Rp]$ whose discrepancy is at least the average, and thus we conclude that, for this $j_0,j_1,j_2\in [0.9Rp,1.1Rp]$ we have \begin{equation}\label{eq:triple-bound}|\sum_j f(j+\ell_0-j_0)f(j+\ell_1-j_1)f(j+\ell_2-j_2)-g(j+\ell_0-j_0)g(j+\ell_1-j_1)g(j+\ell_2-j_2)|\geq \frac{1}{2}\tau^c\end{equation}

    We further point out that $f,g$ both have support $[-0.4Rp,0.4Rp]$, and thus all of $\ell_0-j_0,\ell_1-j_1,\ell_2-j_2$ must be within $0.8Rp$ of each other for Equation~\ref{eq:triple-bound} to be nonzero.

    Next, we apply Lemma~\ref{lem:three-to-one} to the bound of Equation~\ref{eq:triple-bound}, setting $\tau_{lem}:=\frac{\tau^c}{2}$, and using the weight function $w$ supported on the interval from $I_-^{lem}:=0.6Rp$ to $I_+^{lem}:=1.4Rp$; the input offsets are $\ell_0-j_0,\ell_1-j_1,\ell_2-j_2$, though to match the input conditions of the lemma, we subtract off the min of these three from all of them and relabel it as $\ell_0^{lem}=0$, with $\ell_1^{lem},\ell_2^{lem}\in[0,0.8Rp]$ from our above observation that all of $\ell_0-j_0,\ell_1-j_1,\ell_2-j_2$ must be within $0.8Rp$ of each other. Thus in the context of Lemma~\ref{lem:three-to-one} we have $\lambda_{lem}:= 3I_+^{lem}+\ell_1^{lem}+\ell_2^{lem}\leq 5.8Rp$. We check the input condition of the lemma that $\eps\leq\frac{1}{3}p$: since $\eps:=c' \frac{1}{Rp}\log\frac{1}{\tau}$, by our assumption that $\log\frac{1}{\tau}\leq R^{0.8}p^2$ we have $\eps\leq c' R^{-0.2} p$, which is thus $\leq\frac{1}{3}p$ for large enough $R$. We thus invoke the lemma, which yields that: there exists a tuple $S$ of $\leq 3k$ distinct positive integers having minimum value 1 and maximum value $\leq 10(I_+^{lem}+\sigma)\leq 15Rp$ and there exist nonnegative coefficients $A:\mathbb{Z}\rightarrow \mathbb{R}_{\geq 0}$ supported on $[0.6Rp,5.8Rp]$, that sum to $\leq 1$, has discrepancy $|\sum_j A(j)(x_{[d-R:n],P}^{(S+j)}-y_{[d-R:n],P}^{(S+j)})|\geq \frac{\tau^c}{2(15Rp)^{3k}}\geq\tau^{c+1}$, and has Fourier transform that at frequencies $|\xi|\geq 9\eps$ has magnitude  $\leq 300 (I_+^{lem}+\sigma)^2\cdot\left(\lambda_{lem}\delta+\max(\lambda_{lem},\frac{8\pi}{3\eps})\cdot e^{-\eps I_-^{lem}/4}\right)$. Recall from above that $\eps=c' \frac{1}{Rp}\log\frac{1}{\tau}$ and $\delta=\frac{8\pi}{\eps} \tau^{0.075 c'}$. Thus our overall Fourier bound, since $I_-^{lem}=0.6Rp\geq 0.3Rp$, is $\leq \tau^{0.075c'} (Rp)^3$ times a constant, which we can thus bound by $\tau^{0.075c'-1}$ when $R$ exceeds some global constant. \\

    \noindent \textbf{Conclusion 2:} In this case we further apply Corollary~\ref{cor:shift}, which yields that $\big|\sum_{j} (f_b(j)-g_b(j))\cdot e^{ij\frac{2\pi}{R^2}}\big|\geq \frac{\tau_{lem}}{R^2}-\tau_{lem}^2$, which is at least $2\tau^2$ for $R$ larger than some constant since $\tau_{lem}=\frac{\tau}{2Rp}$, and we assumed $\tau\leq \frac{1}{R^4}$ and $p\leq \frac{1}{4}$. 
    Analogously to the Conclusion 1 case, this conclusion is modified by an amount exponentially small in $Rp$ if we limit the binomial convolutions in $f_b,g_b$ to $j'\in[0.9Rp,1.1Rp]$ and we can use pigeonhole to choose a single such $j'$ with at least the average discrepancy. Thus for some $j'\in[0.9Rp,1.1Rp]$ we have $\big|\sum_{j} (f(j-j')-g(j-j'))\cdot e^{ij\frac{2\pi}{R^2}}\big|\geq \tau^2$. However, the shift $j'$ does not affect the magnitude of the discrepancy, so we may drop $j'$ to simply conclude that in the Conclusion 2 case of Lemma~\ref{lem:combined-contrapositive} we have \begin{equation}\label{eq:shift-conclusion}\big|\sum_{j} (f(j)-g(j))\cdot e^{ij\frac{2\pi}{R^2}}\big|\geq \tau^2\end{equation}
    We define $A(j):=w(j) e^{-i(j-Rp)\frac{2\pi}{R^2}}$, so that Equation~\ref{eq:shift-conclusion} says that $|\sum_j A(j)(x_{[d-R:n],p}^{(s+j)}-y_{[d-R:n],p}^{(s+j)})|\geq \tau^2$. Thus this result is analogous to that for the Conclusion 1 case, except we use the original $s,p$ instead of the new $S,P$, and thus the original bounds on $s,p$ are strong enough. We still have strong Fourier bounds on $A$: elementwise multiplying $w(j)$ by $e^{-i(j-Rp)\frac{2\pi}{R^2}}$ simply shifts its Fourier transform by angle $\frac{2\pi}{R^2}$ and multiplies by a phase; thus since since $\frac{2\pi}{R^2}\leq \eps:=c'\frac{1}{Rp}\log\frac{1}{\tau}$ for large enough $R$, we conclude that, in the Conclusion 2 case, $A$ has Fourier transform that at frequencies $|\xi|\geq 2\eps$ has magnitude at most $\frac{8\pi}{\eps} \tau^{0.075 c'} \leq \tau^{0.075 c' - 1}$. 

    Since the desired conclusions hold in either case, this completes the proof.
\end{proof}

We define $R':=c''\frac{R^2 p}{\log\frac{1}{\tau}}$ for a constant $c''$ that we specify now, along with the constant $c'$ used to define $\eps$ at the start of this proof, so that they relate to the constant $c$ of Equation~\ref{eq:triple-bound} induced by Lemma~\ref{lem:combined-contrapositive}. Specifically, we choose $c'$ to be any large enough constant and $c''>0$ to be any constant small enough relative to $c'$ so that
\begin{equation}\label{eq:constants}
0.075c'\geq 3+c\quad \textrm{and}\quad c'c''\leq 0.000004
\end{equation}

\begin{lemma}\label{lem:statistics-deconvolution}
    For large enough $R$, there exists a probability $P \in [p,3p]$, a statistic $S$ of min value 1 and order $\leq 3k$ and span $\leq 15Rp$, a location $j'' \leq R'P + 6.1Rp$ for which
    \begin{equation}\label{eq:final-discrepancy}|x_{[d-R-R':n],P}^{(S+j'')}-y_{[d-R-R':n],P}^{(S+j'')}|\geq \tau^{c_f}\end{equation}
    for any sufficiently large constant $c_f$. 
\end{lemma}
\begin{proof}
    Consider the sequence $A$ given by Lemma~\ref{lem:coefficients-helper} and the associated $P\leq 3p\leq\frac{1}{4}$. We apply Corollary~\ref{cor:deconvolution} to $A(-j)$ to approximately deconvolve the sequence by $\Bin{R'}{P}$. Crucially, we will \emph{flip} $A$ before deconvolving: define the notation $\overleftarrow{A}$ so that $\overleftarrow{A}(j):=A(-j)$. We will denote the output of Corollary~\ref{cor:deconvolution} as $\overleftarrow{A}'$, and thus $A'$ will be the flipped version of the output. The parameters of Corollary~\ref{cor:deconvolution} correspond to the bounds obtained from Lemma~\ref{lem:coefficients-helper}: we let $\delta_{lem}:=\tau^{0.075c'-1}$ and $\eps_{lem}:=9\eps$. 

    We first check the input conditions of Corollary~\ref{cor:deconvolution}. For the condition $e^{-0.16PR'+1}\leq\delta_{lem}$, recall that $\delta_{lem}:=\tau^{0.075c'-1}$ while the left hand side is $e^{-0.16PR'+1}\leq e\cdot e^{-0.16 c''\frac{R^2 p^2}{\log\frac{1}{\tau}}}$; since $\log\frac{1}{\tau}\leq R^{0.8}p$, we bound our expression by $e\cdot e^{-0.16 c''R^{0.4}\log\frac{1}{\tau}}$, which is smaller than $\tau^{0.075c'-1}$ for sufficiently large $R$.

    For the input condition $\delta_{lem}\leq e^{-4PR'\eps_{lem}^2}$, the right hand side is $e^{-4Pc''\frac{R^2 p}{\log\frac{1}{\tau}}(9c' \frac{1}{Rp}\log\frac{1}{\tau})^2}\geq\tau^{3\cdot 4\cdot 9^2c''\cdot {c'}^2}$ which is thus at least $\tau^{0.075c'-1}$ if $972c''\cdot {c'}^2\leq 0.075c'-1$; this equation is implied by the constant conditions in Equation~\ref{eq:constants} since $0.075c'\geq 3$ and thus the right hand side is at least $0.075c'-1\geq 0.05c'$, and thus this condition reduces to $c''c'\leq \frac{0.05}{972}$, which is implied by the second condition of Equation~\ref{eq:constants}.
    
    Thus Corollary~\ref{cor:deconvolution} defines a radius $r_{lem}=300\sqrt{p_{lem}R'\lceil\log\frac{1}{\delta_{lem}}\rceil}$ and shows that there exists a sequence with $|A'(j)|\leq 11\cdot e^{PR'\eps_{lem}^2}$ where the support of $A'$ is the support of $A$ shifted by $R'P$ and expanded by radius $r$; and, crucially, $\overleftarrow{A}'\ast \Bin{R'}{P}$ is within $9\delta_{lem}$ of $\overleftarrow{A}$ pointwise. From our above bounds, the coefficients are bounded as $|A'(j)|\leq 11\cdot e^{p_{lem}R'\eps^2}\leq 11\cdot (\frac{1}{\tau})^{3\cdot 9^2c''\cdot {c'}^2}$. 

    Next, the radius added to the support width by deconvolution is bounded by $r_{lem}=300\sqrt{p_{lem}R'\lceil\log\frac{1}{\delta_{lem}}\rceil}$; since $\lceil\log\frac{1}{\delta_{lem}}\rceil\leq 1+(0.075c'-1)\log\frac{1}{\tau}$ we bound this as $\leq 0.075c'\log\frac{1}{\tau}$ for large enough $R$, and thus $300\sqrt{p_{lem}R'\lceil\log\frac{1}{\delta_{lem}}\rceil}\leq 300\sqrt{3p c''\frac{R^2 p}{\log\frac{1}{\tau}} 0.075c'\log\frac{1}{\tau}}=Rp\cdot 300\sqrt{3\cdot 0.075\cdot c'c''}$. We will show that $r_{lem}\leq 0.3Rp$, which requires $300\sqrt{3\cdot 0.075\cdot c'c''}\leq 0.3$, which is satisfied when $c'c''\leq 0.000004$, which is the second condition of Equation~\ref{eq:constants}. Thus, since Lemma~\ref{lem:coefficients-helper} showed that $A$ is supported within $[0.6Rp,5.8Rp]$, then, shifting this domain by $R'P$ and expanding it by radius $r_{lem}\leq 0.3Rp$ lets us conclude that $A'$ is supported within $[R'P+0.3Rp,R'P+6.1Rp]$.

    Thus since $\overleftarrow{A}'\ast \Bin{R'}{P}$ is within $9\cdot \tau^{0.075c'-1}$ of $\overleftarrow{A}$ pointwise, and the sum of the sizes of their supports is $\leq R'+6R$, we have that $\|\overleftarrow{A}'\ast \Bin{R'}{P}-\overleftarrow{A}\|_1\leq \frac{1}{2}\tau^{0.075c'-2}$ for large enough $R$. Since Equation~\ref{eq:constants} ensures that $0.075c'-2\geq c+1$ then, combining these bounds with Equation~\ref{eq:combined-discrepancy-bound}, for $S$ defined in Lemma~\ref{lem:coefficients-helper} we have that: \begin{equation}\label{eq:binomial-discrepancy}|\sum_{j} (\overleftarrow{A}'\ast \Bin{R'}{P})(-j)\cdot(x_{[d-R:n],P}^{(S+j)}-y_{[d-R:n],P}^{(S+j)})|\geq \frac{1}{2}\tau^{c+1}\end{equation}

    The final step is applying Lemma~\ref{lem:binomial-shift} showing that binomially weighted statistics of substrings $x_{[d-R:n]},y_{[d-R:n]}$ may be reexpressed as statistics of substrings starting $R'$ farther back, thus completing our induction step, ``zooming out'' from scale roughly $R$ to scale roughly $R'$ while preserving a nonnegligible difference in statistics of $x$ versus $y$.

    Explicitly, we expand the left hand side of Equation~\ref{eq:binomial-discrepancy}:
\[|\sum_{j} \sum_{j'} \overleftarrow{A}'(-j-j')\cdot bin(R',j',P)\cdot(x_{[d-R:n],P}^{(S+j)}-y_{[d-R:n],P}^{(S+j)})|\]

    Instead of summing over all $j$, instead define $j''=j+j'$ and equivalently sum over all $j',j''$, where we replace all occurrences of $j$ by the equivalent $j''-j'$:

    \[|\sum_{j''} \sum_{j'} \overleftarrow{A}'(-j'')\cdot bin(R',j',P)\cdot(x_{[d-R:n],P}^{(S+j''-j')}-y_{[d-R:n],P}^{(S+j''-j')})|\]

    We now apply Lemma~\ref{lem:binomial-shift} for each $j''$, to conclude:

    \begin{equation}\label{eq:R'-discrepancy}|\sum_{j''} A'(j'')\cdot (x_{[d-R-R':n],P}^{(S+j'')}-y_{[d-R-R':n],P}^{(S+j'')})|+\sum_{j''} |A'(j'')\Pr[\Bin{R'}{P}\geq 1+j'']|\geq \frac{1}{2}\tau^{c+1}\end{equation}

    Since $A'$ is supported on $j''\geq R'P+0.3Rp$, we bound $\Pr[\Bin{R'}{P}\geq 1+j'']$ by standard Chernoff bounds by $e^{-\frac{(0.3Rp)^2}{3R'P}}\leq \tau^{\frac{(0.3)^2}{9c''}}$; thus since the size of the domain of $A'$ is $\leq 6Rp\leq \frac{1}{50\tau}$ for large enough $R$, and $|A'(j'')|\leq 11\cdot (\frac{1}{\tau})^{3\cdot 9^2c''\cdot {c'}^2}$ from above, then, provided that $\frac{(0.3)^2}{9c''}-3\cdot 9^2c''\cdot {c'}^2 \geq c+2$, then we conclude that the second sum of Equation~\ref{eq:R'-discrepancy} is at most $\frac{1}{4}\tau^{c+1}$ and thus the first sum is at least $\frac{1}{4}\tau^{c+1}$. This condition on $c,c',c''$ is implied by Equation~\ref{eq:constants}, since the condition $c'c''\leq 0.000004$ yields $\frac{(0.3)^2}{9c''}\geq 2500c'$ and $-3\cdot 9^2c''\cdot {c'}^2\geq -0.000972c'$, and the condition $c'\geq c+3$ lets us conclude the required inequality $\frac{(0.3)^2}{9c''}-3\cdot 9^2c''\cdot {c'}^2 \geq c+2$. Thus we have shown:

    \[|\sum_{j''} A'(j'')\cdot (x_{[d-R-R':n],P}^{(S+j'')}-y_{[d-R-R':n],P}^{(S+j'')})|\geq \frac{1}{4}\tau^{c+1}\]

    Now we use pigeonhole to extract a single $j''$ for which the discrepancy of the statistics $|x_{[d-R-R':n],P}^{(S+j'')}-y_{[d-R-R':n],P}^{(S+j'')}|$ is large: we divide our bound $\frac{1}{4}\tau^{c+1}$ by the bound $|A'(j'')|\leq 11\cdot (\frac{1}{\tau})^{3\cdot 9^2c''\cdot {c'}^2}$ and our bound on the domain size $6Rp\leq \frac{1}{50\tau}$. Thus, for large enough $R$, there exists $j''\in [R'P+0.3Rp,R'P+6.1Rp]$ for which 
    \begin{equation}|x_{[d-R-R':n],P}^{(S+j'')}-y_{[d-R-R':n],P}^{(S+j'')}|\geq \tau^{c+2+3\cdot 9^2c''\cdot {c'}^2}\end{equation}
    Choosing $c_f$ to be at least $c+2+3\cdot 9^2c''\cdot {c'}^2$ yields the lemma.
\end{proof}

It remains only to prove bounds on the statistic $S$ and the values $R',\tau'$ as desired by the proposition statement. From Lemma~\ref{lem:statistics-deconvolution}, we have constructed $S$ as having min value 1 and span $\leq 15Rp$ and order $\leq 3k$.

Define $\tau':=\tau^{c_f}$, where we will choose the constant $c_f$ to be at least $c+2+3\cdot 9^2c''\cdot {c'}^2$ so that Equation~\ref{eq:final-discrepancy} yields a $\tau^{c_f}$ discrepancy bound.

To prove the remaining bounds, we first note that $R':=c''\frac{R^2 p}{\log\frac{1}{\tau}}\in [3R^{1.1},R^2]$ for sufficiently large $R$. Then we check that our bound $15Rp$ on the span of $S$ is at most the requirement $ \sqrt{\frac{1}{2}R'P\log\frac{1}{\tau'}}$; this expression is at least $Rp \cdot  \sqrt{\frac{1}{2}c'' c_f}$, which is thus at least $15Rp$ for large enough constant $c_f$. The location $\ell'$ in the conclusion is defined to be $j''$ which we bounded as $\leq R'P+6.1Rp$, and thus is also bounded by $R'P+ \sqrt{\frac{1}{2}R'P\log\frac{1}{\tau'}}$, for large enough constant $c_f$.

We next show $12k\log R'\leq \log\frac{1}{\tau'}\leq ({R'}/2)^{0.8}P^2$. For the first inequality, recall our assumption that $4k\log R\leq \log\frac{1}{\tau}$; since $R'\leq R^2$, and $\log\frac{1}{\tau'}=c_f\log\frac{1}{\tau}$, we have that this inequality is satisfied when $c_f\geq 6$. For the second inequality, recall our assumption $\log\frac{1}{\tau}\leq {R}^{0.8}p^2$; since $P\geq p$ and since $R'\geq R^{1.1}$ from above, we have that $({R'}/2)^{0.8}P^2$ exceeds any constant $c_f$ times ${R}^{0.8}p^2$, for large enough $R$. We have thus concluded all the desired bounds.
\end{proof}

\section{Main Result}\label{sec:main}
In this section we prove the main result of our paper, Theorem~\ref{thm:main}; beforehand, we show the base case needed for induction, Lemma~\ref{lem:mean-based}.

\subsection{Base Case}
The base case analysis relies on standard techniques from the mean-based analysis of traces, but with some adaptations to ensure that the location of the bit $j$ that we measure is as far to the left as possible. See~\cite{de2019optimal,peres17Average} for prior bounds along these lines. We start out with a fact used in all this work bounding Littlewood polynomials.

\begin{fact}[From~\cite{borwein1997littlewood}, Corollary 3.2]
    \label{fact:arcbound}
    Let $F(z)$ be a polynomial with all coefficients bounded by $1$ in magnitude and $|F(0)| = 1$. There exists an absolute constant $C$ such that for any $\theta \in (0,\pi]$, we have $\sup_{t : |t| \leq \theta} |F(e^{it})| \geq \exp(-C/\theta)$. 
\end{fact}

In this subsection we index $x$ starting at $0$, i.e., $x = x_0x_1\dots x_{n-1}$. 

\begin{lemma}[From~\cite{de2019optimal}]
    \label{lem:mean-based}
    Let $x \in \{0,1\}^n$, and let $U \sim \mathbf{Del}_p(x)$. For any $w \in \mathbb{C}$, we have
    \begin{equation*}
        \sum_{j \geq 0} \Exp[U_j] \cdot w^j = p \sum_{i=0}^{n-1} x_i (1-p + pw)^{i} 
    \end{equation*}
\end{lemma}

Our base case result is the following.

\begin{lemma}\label{lem:base}
    Let $x,y \in \{0,1\}^n$ be two strings with $d = \min \{i :x_i \neq y_i\}$. Let $p\in[8 d^{-1/2},\frac{1}{2}]$. Define $\tilde{a}_i:=\Exp_{U\sim \Del{p}{x}}[U_i] - \Exp_{U\sim \Del{p}{y}}[U_i]$. For some fixed constant $C^*>0$, there exists $j \leq pd + C^* p^{1/3} d^{2/3}$ such that $|\tilde{a}_j| \geq \exp(-C^*d^{1/3}p^{-1/3})$, for $d$ larger than some fixed constant. 
\end{lemma}

\begin{proof}
    Let $a_i = x_i - y_i$, and define the following polynomials:
    \begin{align*}
        A(z) &= \sum_{j=0}^{n-1} a_j z^j \\
        B(z) &= \sum_{j=0}^{n-d-1} a_{d+j} z^j
    \end{align*}
    By assumption, the $z^d$ term is the first nonzero term in $A$, so $A(z) = z^{d} B(z)$. Also by assumption $|B(0)| = 1$. Let $\theta = \frac{p^{1/3}}{d^{1/3}}$ and let $r=1-\theta=1-\frac{p^{1/3}}{d^{1/3}}$. Applying~\Cref{fact:arcbound} to the polynomial $B(r z)$, thus there exists $\theta_0$ with $|\theta_0| \leq \theta$ such that $|B(re^{i\theta_0})| \geq \exp(-C/\theta)$. Defining $w_0 = 1 + \frac{r e^{i\theta_0}-1}{p}$ so that $1-p+pw_0=r e^{i\theta_0}$ and applying \Cref{lem:mean-based}, we have
    \begin{align*}
        \abs{\sum_{j \geq 0} \tilde{a}_j w_0^j} &= \abs{p \sum_{j \geq 0} a_j (1-p+pw_0)^j} \\
        &= p\abs{A(1-p+pw_0)} \\
        &= p\abs{(r e^{i\theta_0})^d B(r e^{i\theta_0})} \\
        &\geq p r^d \exp(-C/\theta)
    \end{align*}
    Since $1-r=\frac{p^{1/3}}{d^{1/3}} \leq \frac{1}{2}$ for large enough $d$, we further lower-bound this last expression from the fact that $\log 1-x\geq -x-x^2$ for $x\leq \frac{1}{2}$, as
    \begin{equation}\label{eq:mean-bound} \abs{\sum_{j \geq 0} \tilde{a}_j w_0^j}\geq p \cdot\exp(d(-\theta-\theta^2)-C/\theta)=p\cdot \exp(-d^{2/3}p^{1/3}-d^{1/3}(p^{2/3}+C p^{-1/3}))\end{equation}
    
    We now bound the right tail of the power series; we consider the portion of the sum $j \geq pd+C'p^{1/3}d^{2/3}$, for constant $C'$ which we will relate at the end to the constant $C^*$ of the lemma statement.
    
    From the definition of $w_0$  we have
    \[        |w_0|^2 = \abs{1 + \frac{r e^{i\theta_0}-1}{p}}^2 =(1+\frac{r-1}{p})^2+\frac{2r(1-p)(1-\cos (\theta_0))}{p^2}\]
    Since $1-\cos(\theta_0)\leq \frac{\theta_0^2}{2}$ and $|\theta_0|\leq \frac{p^{1/3}}{d^{1/3}}$ and $r\leq 1$, we bound the second term by $\frac{\theta_0^2}{p^2}\leq \frac{1}{d^{2/3}p^{4/3}}$. We use this to bound $\log|w_0^2|$, where we first bound the log of the first term only, since $1-r=d^{-1/3}p^{1/3}$ and $\log 1-x\leq -x$, by $-\frac{2}{d^{1/3}p^{2/3}}$; then, by our assumption relating $p$ and $d$ we have that $\frac{1}{d^{1/3}p^{2/3}}\leq\frac{1}{4}$; using the fact that $\log$ has derivative $\leq 2$ when evaluated above $\frac{1}{2}$ we conclude that \[\log|w_0|^2\leq -\frac{2}{d^{1/3}p^{2/3}}+\frac{2}{d^{2/3}p^{4/3}}\]
    The right tail is thus bounded by the formula for a geometric series, since $|\tilde{a}_j|\leq 1$:
    \[\abs{\sum_{j \geq pd+C'p^{1/3}d^{2/3}} \tilde{a}_j w_0^j} \leq \frac{|w_0|^{pd+C'p^{1/3}d^{2/3}}}{1-|w_0|} \leq \frac{\exp((-\frac{1}{d^{1/3}p^{2/3}}+\frac{1}{d^{2/3}p^{4/3}})(pd+C'p^{1/3}d^{2/3}))}{1-|w_0|}\]

    For large enough constant $C'$ this expression is negligible compared with our upper bound for the sum over all $j$ computed in Equation~\ref{eq:mean-bound}: dividing this expression by Equation~\ref{eq:mean-bound}, the terms proportional to $d^{2/3}$ in the exponential cancel, leaving a ratio of $\frac{1}{p(1-|w_0|)}\exp\left(d^{1/3}\left(\frac{-C'+C+1}{p^{1/3}}+p^{2/3}\right)+\frac{C'}{p}\right)$ which can be easily seen to be at most $\exp(-d^{1/3}p^{-1/3})$ for sufficiently large $C'$ and $d$.

    On the other hand, for $j \leq pd$, we can use Chernoff bounds to bound $\tilde{a}_j$: entry $j$ of the trace can only show discrepancy between $x$ and $y$ if it comes from location $\geq d$ in the original string; this can only occur if, of the bits $0,\ldots,d-1$ in the original string, $\leq j$ bits were retained in the trace. Chernoff bounds show this has probability at most $e^{-\frac{(dp-j)^2}{2dp}}$. Thus 
\[\abs{\sum_{j < pd-C'p^{1/3}d^{2/3}} \tilde{a}_j w_0^j} \leq \sum_{j < pd-C'p^{1/3}d^{2/3}}   e^{-\frac{(dp-j)^2}{2dp}} |w_0|^j\]
where for large enough constant $C'$, the expression for each valid $j$ is negligible compared to the bound of Equation~\ref{eq:mean-bound}: our above bounds show that, for $j = pd-C'p^{1/3}d^{2/3}$, the $j^{\textrm{th}}$ term of the above equation is bounded by \[e^{-\frac{(dp-j)^2}{2dp}}e^{(-\frac{1}{d^{1/3}p^{2/3}}+\frac{1}{d^{2/3}p^{4/3}})j}=e^{-d^{2/3}p^{1/3}+d^{1/3}p^{-1/3}(C'+1-{C'}^2/2)  -C'/p}\]
and it is straightforward to check that our bounds decay as $j$ decreases.

Thus since the left and tails of the sum of Equation~\ref{eq:mean-bound} are negligible compared to the overall sum, we conclude from the pigeonhole principle that there exists $j_0$ in the central region  $[pd-C'p^{1/3}d^{2/3},pd+C'p^{1/3}d^{2/3})$ for which $|\tilde{a}_{j_0}\cdot w_0^{j_0}|\geq \frac{p}{4C'p^{1/3}d^{2/3}}\cdot \exp(-d^{2/3}p^{1/3}-d^{1/3}(p^{2/3}+C p^{-1/3}))$, where this expression is half of the sum from Equation~\ref{eq:mean-bound} divided by the size of the range. We bound $|\tilde{a}_{j_0}|$ itself by dividing this bound by the upper bound that, for $j_0$ in this range, we have $|w_0|^{j_0}\leq \exp((-\frac{1}{d^{1/3}p^{2/3}}+\frac{1}{d^{2/3}p^{4/3}})(pd-C'p^{1/3}d^{2/3}))$, letting us conclude that \[|\tilde{a}_{j_0}|\geq \frac{p}{4C'p^{1/3}d^{2/3}}\cdot \exp(-d^{1/3}((C'+C+1)p^{-1/3}+p^{2/3})+C'/p)\]

Thus choosing $C^*\geq C'+C+2$ lets us conclude the lemma, for large enough $d$.
\end{proof}

\subsection{Proof of Main Theorem}
\begin{theorem}\label{thm:main}
There exists a constant $c>0$ such that for any retention probability $p>0$, trace reconstruction on $n>1$ bit strings can be done from $e^{p^{-7/3}(\log_2 n)^c}$ traces.
\end{theorem}

\begin{proof}
We first show how to distinguish any two $n$-bit strings $x,y$ using quasipolynomial traces, using Proposition~\ref{prop:induction} as the induction step with Lemma~\ref{lem:base} as the base case. Let $d$ be the point of first discrepancy between $x,y$. If at any point in this proof we consider a deletion channel that extends ``off the end'' of our strings, the intention is to implicitly pad the string by prepending and appending known 0s and simulate traces from the enlarged string, which can be easily done from traces of the original string, $x$ or $y$ by prepending or appending known 0s of appropriate randomly sampled lengths.

We must choose parameters for the base case lemma in light of the expected trajectory of the induction, so we analyze how parameters evolve over the induction first. At each invocation of Proposition~\ref{prop:induction}, we start at some location $d-2R$ which is $2R$ back from the point of first discrepancy, $d$, and end with an analogous conclusion at location $d-R-R'$, where Proposition~\ref{prop:induction} guarantees that $R+R'\geq 3R^{1.1}>(2R)^{1.1}$. Namely, each application of Proposition~\ref{prop:induction} takes the deletion channel's starting location relative to $d$ and raises it to the power of $\geq 1.1$. Thus if we start at any length $\geq 2$, the number of times we need to raise this to the $1.1$ power until we exceed $n$ is at most $\log_{1.1} \log_2 n$. Thus let $b:=\lceil\log_{1.1} \log_2 n\rceil$ be our bound on the number of iterations of Proposition~\ref{prop:induction} needed until we are guaranteed to reach the beginning of the string.

Meanwhile, each application of Proposition~\ref{prop:induction} increases the retention probability $p_{lem}$ by a factor of up to 3, so to make sure that after the last iteration our retention probability is still smaller than the actual retention probability $p$ of our traces, we need to choose a base case value $p_{base}:=3^{-b}\cdot \min(\frac{1}{4},p)$ (where the min with $\frac{1}{4}$ is chosen so that, at the last invocation of Proposition~\ref{prop:induction}, we still have $p_{lem}\leq\frac{1}{12}$ as required by the lemma).

We will thus use Lemma~\ref{lem:base} with $p_{base}$ as a base case for our induction to find a mean-based difference between traces of $x,y$; set $d_{base}=2R_{base}=\max(1/p_{base}^6,C_1,C_2,C_3)$ where $C_1$ is the lower bound on $d$ of Lemma~\ref{lem:base} and $C_2$ is the lower bound on $2R$ from Proposition~\ref{prop:induction} and $C_3$ is a constant to be specified below; this guarantees, since $p_{base}\leq\frac{1}{4}$, that the input condition of Lemma~\ref{lem:base}, that $p_{base}\geq 8d_{base}^{-1/2}$, is satisfied. Then Lemma~\ref{lem:base} guarantees that, setting $\tau_{base}=\exp(-\max(1,2{C^*}^2) \cdot d_{base}^{1/3} p_{base}^{-1/3})$, then, since strings $x_{[d-2R_{base}:n]},y_{[d-2R_{base}:n]}$, when zero-indexed as Lemma~\ref{lem:base} does, have first discrepancy $d_{base}=2R_{base}$, then there exists $j\leq 2R_{base}p_{base}+\sqrt{R_{base}p_{base}\log\frac{1}{\tau_{base}}}$ for which the order-1 (``mean based'') statistics differ as: \[|x_{[d-2R_{base}:n],p_{base}}^{(1+j)}-y_{[d-2R_{base}:n],p_{base}}^{(1+j)}|\geq \tau_{base}\]

We now will invoke the induction step, Proposition~\ref{prop:induction}, starting from $k_{base}=1$ with $s_{base}$ the order-1 statistic consisting of just entry 1, which thus has span $\sigma_{base}=1$. The remaining input condition of Proposition~\ref{prop:induction} requires that $4 \log R_{base}\leq\log\frac{1}{\tau_{base}}\leq R_{base}^{0.8} p_{base}^2$; since $\log\frac{1}{\tau_{base}}=\max(1,2C^*)^2 \cdot 2^{1/3}R_{base}^{1/3} p_{base}^{-1/3}$ and since $2R_{base}\geq p_{base}^{-6}$, it is easy to check that our two inequalities must hold when $2R_{base}$ is larger than some constant, which we denote $C_3$, as referenced above.

Thus we apply Proposition~\ref{prop:induction}, and in fact show that we can apply it up to $b$ times in succession, stopping when the location $d-2R$ becomes $\leq 1$ and we thus may interpret its conclusion as applying to a deletion channel on the entire strings $x,y$. It is straightforward to check that the output conditions of Proposition~\ref{prop:induction} imply exactly the required input conditions for the next application of  Proposition~\ref{prop:induction}: the lemma starts with a statistic $s$ of order $k$ having discrepancy $\geq \tau$ between $x,y$ when applied on a deletion channel of retention probability $p$ starting at location $d-2R$; the lemma concludes with a new statistic $S$ of order $\leq 3k$ having discrepancy $\tau'$ when applied to a deletion channel of retention probability $P\in[p,3p]$ starting at location $d-R'-R$. The span of $S$ is bounded by $\sqrt{\frac{1}{2}R'P\log\frac{1}{\tau'}}$ which, when considering the value of $R$ needed for the next invocation of the lemma, $R_{next}=\frac{1}{2}(R'+R)\geq \frac{1}{2}R'$, means that the span of $S$ is indeed bounded by $\sqrt{R_{next}P\log\frac{1}{\tau'}}$ as required. The other conditions of the lemma for the next iteration are all easily seen to be implied by the analogous conclusions of the lemma in the previous iteration, including: $4(3k)\log R_{next} \le 12k\log R' \le \log(1/\tau')$ and $(R'/2)^{0.8}P^2 \le R_{next}^{0.8}P^2$.

We thus analyze the conclusions of Proposition~\ref{prop:induction} after its last iteration to conclude our main result. The discrepancy $\tau$ decays inverse polynomially in each iteration; and since $\tau$ is bounded away from 1 we can thus find a constant $C$ encoding the bounds of Proposition~\ref{prop:induction} so that in each iteration $\tau'\geq \tau^C$. After $b$ iterations, we thus conclude that there is some offset $\ell$ and statistic $S$ of order $\leq 3^b$ such that for retention probability $P$ bounded by the retention probability of the actual channel $p$, the discrepancy is bounded as \[|x_{[d-R-R':n],P}^{(S+\ell)}-y_{[d-R-R':n],P}^{(S+\ell)}|\geq \tau_{base}^{C^b}\]

We thus estimate the statistic $S+\ell$ from traces to distinguish $x$ from $y$: pad traces by prepending with 0s when $d-R-R'<1$, and simulate a $P<p$ retention probability from the actual retention probability $p$ traces by throwing out each bit i.i.d. with probability $1-\frac{P}{p}$. 
Because a statistic is a 0/1-valued random variable, we can thus estimate the expected value of our statistic $S+\ell$ to within error $<\frac{1}{2}\tau_{base}^{C^b}$ with probability $\geq\frac{2}{3}$ using $4\cdot\tau_{base}^{-2\cdot C^b}$ samples (from Chebyshev's inequality). Thus the sample complexity of distinguishing $x$ from $y$ from traces of retention probability $p$ is $\leq 4\cdot\tau_{base}^{-2\cdot C^b}$.

Plugging in the definitions that  $\tau_{base}=\exp(-\max(1,2C^*)^2 \cdot d_{base}^{1/3} p_{base}^{-1/3})$ and $d_{base}=\max(1/p_{base}^6,C_1,C_2,C_3)$ we see that $\tau_{base}\geq \min(\frac{1}{C'},e^{-C''p_{base}^{-7/3}})$ for some constants $C',C''>0$; plugging in $p_{base}=3^{-b}\cdot \min(\frac{1}{4},p)$ yields $\tau_{base}\geq \min(\frac{1}{C'},e^{-C''3^{(7/3)b}\cdot\max(4^{7/3},p^{-7/3})})$ and now raising to the $-2C^b$ power and multiplying by 4 yields a sample complexity bounded by $4\cdot \max({C'}^{2C^b},e^{2 C^b\cdot C''\cdot3^{(7/3)b}\max(4^{7/3},p^{-7/3})})$. Since $b:=\lceil\log_{1.1} \log_2 n\rceil$, for $n\geq 3$ we can upper bound any constant to the $b$ power by an appropriate constant power of $\log_2 n$, and thus we can bound all the constant and constant-to-the-$b$ terms in the exponent by $(\log_2 n)^c$ for some $c$, leading to our bound of $e^{p^{-7/3}(\log_2 n)^c}$.

It is a standard fact that, if one can distinguish pairs of traces in a certain sample complexity, then $O(n)$ factor more traces suffice to reconstruct an arbitrary $x$, and we can increase $c$ to encompass this small extra factor.

Finally, for the remaining $n=2$ case we can trivially calculate that with probability $p^2$ a trace will retain both bits and thus reveal the original string; given $j$ traces, none of them will be of length $2$ with probability $\leq (1-p^2)^j\leq e^{-jp^2}$; so setting $j=\frac{\log 3}{p^2}$, we will see the whole string with probability $\geq \frac{2}{3}$. It is easy to check that our bound $e^{p^{-7/3}}$ always exceeds $\frac{\log 3}{p^2}$, proving the $n=2$ case.
\end{proof}

\section{Technical Ingredients}\label{sec:misc}
The following lemmas lead to Lemma~\ref{lem:cumulative-epsilon} that is used to define the weight function $w$ in the main induction result of Proposition~\ref{prop:induction}, but also as part of Section~\ref{sec:alpha}.

\begin{lemma}\label{lem:bump-pdf-alpha}
Given two integers $\alpha,m >0$, define the function $B^{pdf}_{\alpha,m }:\mathbb{Z}\rightarrow[0,1]$ that, on input $j$, equals the binomial probability $bin(\alpha,j,\frac{1}{2})$, except clamped to 0 if $|j-\frac{\alpha}{2}|\geq \frac{m}{2}$, and with the tail mass added back in arbitrarily to locations $j\in (\frac{\alpha}{2}-\frac{m}{2},\frac{\alpha}{2}+\frac{m}{2})$. Then its Fourier transform at angle $\xi$ has magnitude at most $e^{-\alpha\xi^2/8}+4\cdot e^{-m^2/(2\alpha)}$.
\end{lemma}
\begin{proof}
The Fourier transform of $\Bin{\alpha}{\frac{1}{2}}$ at angle $\xi$ has magnitude $\cos(\frac{\xi}{2})^{\alpha}\leq e^{-\alpha\xi^2/8}$, for $\xi\in[-\pi,\pi]$.

Meanwhile, the probability mass of the binomial at distance $\geq \frac{m}{2}$ from the mean is bounded, by Hoeffding's inequality, to be $\leq 2\cdot e^{-m^2/(2\alpha)}$; so removing this mass at the tails and adding it in elsewhere changes the distribution by at most $4\cdot e^{-m^2/(2\alpha)}$ in $\ell_1$ distance. Adding these two bounds yields the lemma.    
\end{proof}
\begin{lemma}\label{lem:bump-pdf-epsilon}
Given integer $m >0$ and real number $\eps\in(0,1]$, there exists a function $B^{pdf}:\mathbb{Z}\rightarrow\mathbb{R_{\geq 0}}$ with sum 1 and supported on $(\frac{m}{\eps} - \frac{m}{2}, \frac{m}{\eps} + \frac{m}{2})$ such that for any $|\xi|\geq \eps$, the magnitude of the Fourier transform of $B^{pdf}$ is at most $8\cdot e^{-\eps m/4}$. 
\end{lemma}
\begin{proof}
Let $\alpha=\lfloor\frac{2m}{\eps}\rfloor$ and apply Lemma~\ref{lem:bump-pdf-alpha} to yield $B^{pdf}_{\alpha,m}$. From Lemma~\ref{lem:bump-pdf-alpha}, at angle $\xi$, since $|\xi|\leq \pi$, its Fourier transform has magnitude at most $e^{-(\frac{2m}{\eps}-1)\xi^2/8}+4\cdot e^{-m^2/(2\frac{2m}{\eps})}$; since $e^{\pi^2/8}<4$, this is bounded by $8\cdot e^{-\eps m/4}$.
\end{proof}

\begin{lemma}\label{lem:cumulative-epsilon}
Given integer $m>0$ and real number $\eps\in(0,\pi]$, let $b$ denote the function from Lemma~\ref{lem:bump-pdf-epsilon}, shifted to have support in $\{1,\ldots,m\}$. Then, for an integer interval $I=\{I_-,\ldots,I_+\}$ define the function $B_{I,\eps,m}(j):=\sum_{\ell=I_-}^{I_+-m-1} b(j-\ell)$. Then $B_{I,\eps,m}$ is supported in $I$; equal to 1 on the subinterval that is $m$ smaller, $\{I_-+m,\ldots,I_+-m\}$; has sum $I_+-I_- -m$; takes values in $[0,1]$; and its Fourier transform at angle $|\xi|\geq \eps$ has magnitude at most $\frac{8\pi}{\eps} e^{-m\eps/4}$.
\end{lemma}
\begin{proof}
The support bounds come from fact that $b$ is supported on $\{1,\ldots,m\}$ from Lemma~\ref{lem:bump-pdf-epsilon}; the bounds on the values and sum of $B_{I,\eps,m}$ come from the fact that $b$ is the pdf of a distribution, and $B_{I,\eps,m}(j):=\sum_{\ell=I_-}^{I_+-m-1} b(j-\ell)$ is a sum of $b(t)$ on a subset of those $t$ in the support $\{1,\ldots,m\}$ and on the claimed flat interval, that subset contains all of $\{1,\ldots,m\}$, yielding $B(j)=1$ there.

The Fourier bound follows from defining $\widehat{b}$ to be the Fourier transform of $b$ and $\widehat{B}$ to be the Fourier transform of $B_{I,\eps,m}$ where, since $B_{I,\eps,m}(j):=\sum_{\ell=I_-}^{I_+-m-1} b(j-\ell)$ we have that $\widehat{B}(\xi)=\widehat{b}(\xi)\sum_{\ell=I_-}^{I_+-m-1} e^{i\xi \ell}$. Thus from the formula for the sum of a geometric series, for $|\xi|\leq \pi$ we have $|\widehat{B}(\xi)|\leq|\widehat{b}(\xi)|\cdot\frac{2}{|1-e^{i\xi}|}\leq |\widehat{b}(\xi)|\cdot\frac{\pi}{|\xi|}$, which yields our final bound, after substituting the Lemma~\ref{lem:bump-pdf-epsilon} bound that for $|\xi|\geq \eps$ we have $|\widehat{b}(\xi)|\leq 8\cdot e^{-\eps m/4}$.
\end{proof}

\section*{Acknowledgements}
The authors would like to thank Zachary Chase for many insightful discussions, and Alex Wein for pointing out the connection to the problem of multiple reference alignment.
\bibliographystyle{alpha}
\bibliography{reference}

\end{document}